\documentclass[letterpaper,12pt]{article}
\usepackage{amsfonts}
\usepackage{amsmath, amsthm, amssymb, bbm,verbatim}
\usepackage{xcolor, natbib}
\usepackage{graphicx}

\numberwithin{equation}{section} 

\newtheorem{assumption}{Assumption}
\newtheorem{theorem}{Theorem}
\newtheorem{example}{Example}

\newtheorem{lemma}{Lemma}

\newtheorem{remark}{Remark}
\newtheorem{algorithm}{Algorithm}

\newcommand{\plim}{\operatorname*{plim}}
\newcommand{\diag}{{\rm diag}}
\newcommand{\argmin}{\operatorname*{argmin}}
\newcommand{\argmax}{\operatorname*{argmax}}
\newcommand{\E}{\mathbb{E}}
\newcommand{\EE}{\overline{\mathbb{E}}}
\newcommand{\Ep}{\mathbb{E}}

\newcommand{\mD}{\mathcal{D}}
\newcommand{\mB}{\mathcal{B}}

\def\ft#1#2{{\textstyle {\frac{#1}{#2}} }}

\begin{document}

\title{Nonlinear Factor Models for Network and Panel Data\thanks{
       Preliminary versions of this paper were presented at several conferences and seminars.
       We thank the participants to these presentations, the editor, an associate editor, two anonymous referees, Shuowen Chen, Riccardo D'Adamo, Siyi Luo and Carlo Perroni for helpful comments.   
  Fern\'andez-Val  gratefully acknowledges support from the National Science Foundation, and Spanish State Research Agency MDM-2016-0684 under the Mar\'ia de Maeztu Unit of Excellence Program. 
       Weidner gratefully acknowledges support from the Economic and Social Research Council through the ESRC Centre for Microdata Methods and Practice grant RES-589-28-0001 and from the European Research Council grant ERC-2014-CoG-646917-ROMIA.}
     }

\author{\setcounter{footnote}{2} 
       Mingli Chen\footnote{Department of Economics, University of Warwick, Gibbet Hill Road, Coventry CV4 7AL, UK. 
       Email: {\tt m.chen.3@warwick.ac.uk}
      }
    \and 
        Iv\'{a}n Fern\'{a}ndez-Val\footnote{
     Department of Economics, Boston University,
     270 Bay State Road,
     Boston, MA 02215-1403, USA.
     Email: {\tt ivanf@bu.edu}
      }
   \and
      Martin Weidner\footnote{
                   Department of Economics,
                   University College London,
                   Gower Street,
                   London WC1E~6BT,
                   UK,
                   and CeMMAP.
                   Email: {\tt m.weidner@ucl.ac.uk}
                   }
                   }

\date{\bigskip \today}

\maketitle

\abstract{\noindent
Factor structures or interactive effects are convenient devices to incorporate latent variables in panel data models. We consider fixed effect estimation of nonlinear panel  single-index models with factor structures in the unobservables, which include logit, probit, ordered probit and Poisson specifications. We establish that  fixed effect estimators of model parameters and average partial effects have normal distributions when the two dimensions of the panel grow large, but might suffer from incidental parameter bias. We show how models with  factor structures can also be applied to capture important features of network data such as reciprocity, degree heterogeneity, homophily in latent variables, and clustering. We illustrate this applicability with an empirical example to the estimation of a gravity equation of international trade between countries using a Poisson model with multiple factors.

\bigskip


{\bfseries Keywords:} Panel data, network data, interactive fixed effects, factor models, bias correction, incidental parameter problem, gravity equation 

{\bfseries JEL:} C13, C23.

\newpage

\section{Introduction}

Factor structures or interactive effects are convenient devices to incorporate latent variables in panel data models. They are commonly used to capture  aggregate shocks that might have heterogeneous impacts on the agents  in macroeconomic models, and multidimensional individual heterogeneity that might have time varying effects in microeconomic models.  More generally, the inclusion of these structures serves to account for dependences along the cross-section and time series dimensions in a parsimonious fashion.  While methods for linear factor models  are well-established, there are very few studies that develop methods for nonlinear factor models. (We provide a literature review at the end of this section.) Nonlinear models are commonly used when the outcome variable is discrete or has a limited support. In this paper we introduce factor structures in single-index nonlinear specifications such as the logit, probit, ordered probit and Poisson models. 

\medskip

The model that we consider is semiparametric. It includes an outcome, strictly exogenous covariates, and a fixed number of factors and factor loadings. The parametric part is the distribution of the outcome conditional on the covariates, factors and loadings, which is specified up to a finite dimensional parameter. The nonparametric part is the distribution of the factors and loadings conditional on the covariates. In other words, our model is of the ``fixed effects'' type because we do not impose any restriction on the relationship between the observed covariates and the unobserved factors and loadings. This flexibility allows us to capture features of economic behavior  more realistically, but poses important challenges to estimation and inference. The objects of interest are the model parameter and average partial effects (APEs), which  are averages of functions of the data,  parameter, factors and loadings. The APEs measure the effect of covariates on moments of the outcome conditional on the covariates, factors and loadings. We consider a fixed effects estimation approach that treats the factors and loadings as parameters to be estimated.  As it is well-known in the panel data literature, the resulting estimators generally suffer from the incidental parameter problem  coming from the high-dimensionality of the estimated parameter \citep{NeymanScott1948}.   

\medskip

We derive asymptotic theory for our estimators of the model parameter and APEs under sequences where the two dimensions of the panel pass to infinity with the sample size. Even establishing consistency is complicated in our setting because the dimension of the estimated parameters increases with the sample size. We develop a new proof of consistency that relies on concavity of the log-likelihood function  on a single-index that captures the dependence on covariates, parameter, factors and loadings. However, unlike \cite{FW16}, we need to deal with the complication that our log-likelihood function is not concave in all the estimated parameters because the factors and loadings enter multiplicatively in the index.  We also establish that our estimators are normally distributed in large samples, but might have biases of the same order as their standard deviations. For example, we find that the estimator of the model parameter is asymptotically unbiased in the Poisson model, but is biased in logit and probit models. Following the recent panel data literature,  we develop analytical and split-sample corrections for the case where the estimator has asymptotic bias.  One specific feature of our estimator is that the bias depends on the number of factors. In particular, we show  that the bias grows proportionally with the number of factors in examples. 

\medskip

We discuss implementation details of our methods including the computation of the estimator and selection of the number of factors. Thus, we propose an EM-type algorithm  based on \cite{C14} and a concrete proposal to estimate the number of factors based on the eigenvalue ratio test of  \cite{AhnHorenstein2013}. 
The estimator of the number of factors requires to specify an upper bound for the number of factors, but does not rely on any arbitrary choice of penalty function or other tuning parameter.  We do not provide asymptotic theory for this estimator, but show that it  performs well in numerical simulations.
Formally deriving the theory is rather challenging, because it requires to study the
asymptotic properties of the initial fixed effects estimators of the parameters and factor structure obtained from a specification with too many factors, which is a difficult problem even in linear panel factor model
(\citealt{MoonWeidner2015}). We  leave this analysis to future research.

\medskip

We also introduce factor structures as practical tools to model network data. We show how the inclusion of latent factors is useful to incorporate important features of the network such as reciprocity, degree heterogeneity, homophily on latent variables, and clustering \citep{S11, G15}. 
 We focus on directed networks with  unweighted and weighted outcomes. These cover binary response models for network formation where the outcome is an indicator for the existence of a  link between  sender and receiver, and count data models for  network flows where the outcome is a measure of the volume of flow between  sender and receiver. As we shall discuss, our factor model provides a parsimonious reduced-form specification that captures the important network features  mentioned above.  The statistical treatment of the network factor model is identical to the panel factor model after noticing that a network is isomorphic to a panel after labeling the senders as individuals and the receivers as time periods. 

\medskip

We illustrate the use of the factor structure in network data with an application to gravity equations of trade between countries. We estimate a Poisson model where the outcome is the volume of trade and the covariates include typical gravity variables such as the distance between the countries or whether the country pair  belongs to a currency union or a free trade area. The unobserved factors and loadings serve to account for scale and multilateral resistance effects, unobserved partnerships, presence of multinational firms, and differences in natural resources or industrial composition. We find that accounting for these multiple unobserved factors changes the effects of the gravity variables, making all of them to have the expected signs while keeping most of them to be statistically significant.

\paragraph{Literature review: }  This paper contributes to the econometric panel data and network data literatures.  Regarding the panel data literature, our statistical analysis relies on the recent developments in fixed effects methods. We refer to \cite{ARE} for a recent review on fixed effects estimation of nonlinear panel models with additive individual and time effects, and to \cite{BP16} for a recent review on fixed effects estimation of linear factor or interactive effects panel models. Since the first draft of this paper appeared in \cite{CFW14}, \cite{BL17} and \cite{AB16} have considered special cases of nonlinear factor models. \cite{BL17} analyzed a probit model using the common correlated random effects approach of \cite{Pesaran2006}, and \cite{AB16} a logit model using a Bayesian approach with data augmentation. Our analysis is different in the modeling assumptions and estimation method.\footnote{We refer to \cite{BL17} and \cite{AB16} for more detailed comparisons with our analysis.}  The most closely related work is   \cite{FaWang2018}. This paper derives the asymptotic distribution of the estimators of the factors and loadings in non-linear single index models without covariates. By contrast, we focus on covariate coefficients and average partial effects and treat the factors and loadings as nuisance parameters. Accordingly, we view our results as complementary to the results in   \cite{FaWang2018}.

\medskip

In terms of the network literature, our paper is related to the recent work on the application of panel fixed effects methods to network data including \cite{FW16}, \cite{Yan2016statistical}, \cite{Stata2017}, \cite{Dzemski2017}, \cite{Graham17},  and \cite{yan2018}. These papers account for degree heterogeneity by including additive unobserved sender and receiver effects. Additive effects, however, do not capture other network features such as homophily in latent factors and clustering. \cite{Graham2016} considered a binary response model of network formation with all these features plus state dependence, for the case where the network is observed at multiple time periods. Compared to \cite{Graham2016}, our method can capture all these features, except for state dependence, applies to  ordered and count outcomes in addition to binary outcomes, and only requires observing the network at one time period.   A stream of the statistic literature has considered nonlinear factor network models using a random effects approach including \cite{HRH02}, \cite{hoff05}, \cite{KHRH09}, and \cite{HRT07}. Unlike the fixed effects approach that we adopt, the random effects approach assumes independence between covariates and factors and between  covariates and loadings. This assumption is regarded as implausible for most economic applications where the loadings reflect unobserved individual heterogeneity and  some of the covariates are individual choice variables.  There is also a recent econometric literature on structural models  of strategic network formation where the main focus is on  identification.  We refer to \cite{depaula2017} for an excellent up-to-date review on this topic. The focus of our paper is on estimation and inference.

\medskip

Finally, there is an extensive literature in international economics on the estimation of the gravity equation  including \cite{Harrigan1994}, \cite{EatonKortum2001}, \cite{AndersonWincoop2003},  \cite{SantosSilvaTenreyro2006}, \cite{Helpman01052008}, \cite{Charbonneau2011} and \cite{jochmans2017two}. We refer to \cite{HEAD2014} for a recent review on this literature.  These papers estimate models with additive unobserved sender and receiver country effects to account for scale or multilateral resistence effects. Our innovation to this literature is the inclusion of multiple unobserved factors to account for not only scale effects, but also unobserved partnerships, and homophily induced by differences in natural resources, industrial composition or other country characteristics. 

\medskip

To sum-up, our paper makes the following contributions. First, we derive asymptotic theory for fixed effects estimators of model parameters and APEs in a class of nonlinear single-index factor models that include logit, probit, ordered probit and Poisson models. Second, we provide bias corrections for fixed effects estimators of model parameters  and APEs. Third, we propose an estimator of the number of factors in nonlinear single-index models with factor structure. Fourth, we bring in the factor structure to model important features of network data such as reciprocity, degree heterogeneity, homophily in latent factors and clustering in a reduced form fashion. Fifth, we apply our methods to the estimation of a gravity equation of trade between countries and confirm the importance of the gravity variables even after conditioning on multiple unobserved latent factors.

\medskip

\paragraph{Outline: } In Section \ref{sec:model}, we introduce the model and estimators. Section  \ref{sec:numeric}  discusses the statistical issues in the estimation and inference of factor models with a simple example. Section  \ref{sec:app_panel} 
derives asymptotic theory for our estimators. Section \ref{sec:implementation} provides implementation details for the estimators of the parameters and number of factors. Section \ref{sec:empirics} describes the results of the empirical application to the gravity equation and a calibrated simulation.
The proofs of the main results and other technical details are given in the Appendix.

\section{Model and Estimators}
\label{sec:model}

\subsection{Model} 
We observe the data $\{(Y_{ij}, X_{ij}) : (i,j) \in \mD \}$, where $Y_{ij}$ is a scalar outcome variable  and $X_{ij}$ is a $d_x$-dimensional vector of covariates.  The subscripts $i$ and $j$ index individuals and  time periods in traditional panels, but they might index different dimensions in other data structures such as network data. In our empirical application, for example, we use country trade network data where $Y_{ij}$ is the volume of trade between country $i$ and country $j$, and $X_{ij}$ includes gravity variables such as the distance between country $i$ and country $j$.  Both $i$ and $j$  index countries as exporters and importers respectively. The set  $\mD$ contains the indexes of the units that are observed.  It is a subset of the set of all possible pairs $\mD_0  := \{(i,j) :  i = 1,\dots,I; j = 1, \dots, J \}$, where  $I$ and $J$ are the dimensions of the data set. We introduce $\mD$ to  allow for missing data that are common in panel and network applications. For example, in the trade application  $I=J$ and $\mD = \mD_0 \setminus  \{(i,i) :  i = 1,\dots,I \}$ because we do not observe trade of a country with itself. We denote the total number of observations by $n$, i.e. $n = |\mD|$. 

\medskip

We assume that the outcome is generated by
\begin{align}\label{eq:model}
Y_{ij}  \mid X_{ij}, \, \beta, \, \alpha, \, \gamma   &\sim f(\cdot \mid z_{ij}), 
&
z_{ij} &:= X_{ij}' \beta + \pi_{ij},
&
\pi_{ij} :=  \alpha_i' \, \gamma_j ,
\end{align}
where  $f$ is a known density function with respect to some dominating measure,
 $\beta$ is $d_{x}$-dimensional parameter vector,
and $\alpha_i$ and $\gamma_j$ are $R$-vectors of unobserved effects. 
We collect these effects in the $I \times R$ matrix
 $\alpha = (\alpha_1,\ldots, \alpha_I)'$, and
 the $J \times R$ matrix
 $\gamma = (\gamma_1,\ldots, \gamma_J)'$, which are further 
stacked in the $R(I+J)$-vector $\phi_n = ({\rm vec}(\alpha)', {\rm vec}(\gamma)')'$. We make explicit in $\phi_n$  that the number of unobserved effects changes with the sample size because it will have important effects on the asymptotic theory.  We assume that the dimension of the unobserved effects $R$ is known, and provide a practical method to estimate $R$ in Section \ref{sec:implementation}. 
 The effects $\alpha_i$ and $\gamma_j$ are unobserved factors  and factor loadings. In panel data they represent individual and  time effects that in economic applications capture individual heterogeneity and aggregate shocks, respectively.  In network data 
$\alpha_i$ and $\gamma_j$ represent unobserved characteristics of senders and receivers that affect the network flow. 
The model is semiparametric because  we do not specify the distribution of the unobserved effects nor their relationship with the covariates. This flexibility is important for economic applications where some of the covariates are choice variables with values determined in part by the unobserved effects. The conditional distribution $f$ represents the parametric part of the model. 

\medskip

 The model has a single-index specification because the covariates and unobserved effects  enter $f$ through the index $ z_{ij} = X_{ij}'\beta + \alpha_i' \gamma_j$. The parameter $\beta$ is a quantity of interest because it  measures the effect of the covariates on the distribution of the outcome controlling for the unobserved effects. For example, in network data $\beta$ can measure homophily in an observable characteristic $W$ if $X_{ij}$ includes $(W_i - W_j)^2$ as one of its  components. The unobserved effects have a factor or interactive structure because they enter the index $z_{ij}$ multiplicatively through $\pi_{ij} =  \alpha_i' \, \gamma_j$. The standard additive structure $\alpha_{1i} +  \gamma_{1j}$ can be seen as a special case of the factor structure with $R=2$, $\alpha_i = (\alpha_{1i},1)'$, and $\gamma_j = (1, \gamma_{1j})'$.  More generally, in panel data applications the factor structure allows one to incorporate multiple aggregate shocks $\gamma_t$ with heterogeneous effects across agents $\alpha_i$, or multidimensional individual heterogeneity $\alpha_i$ with time-varying returns $\gamma_t$. For example, we can have productivity and monetary shocks with heterogeneous effects across industries, or multiple dimensions of individual ability and skills with time-varying returns in the labor market. 
 
 \medskip
 
 One of the contributions of the paper is to introduce factor structures to network data. In this case the factor structure serves to capture important network features in an unspecified or reduced-form fashion. For example, degree heterogeneity can be captured with the additive structure $\alpha_{1i} +  \gamma_{1j}$ mentioned above, and reciprocity  by allowing $Y_{ij}$ to be arbitrarily related to $Y_{ji}$ even after conditioning on the covariates and unobserved effects. Another important feature is homophily in latent factors, in addition to the homophily on observed factors captured by $X_{ij}$. Assume that there is a latent factor $\xi_i$ such that the flow between $i$ and $j$ increases or decreases with the distance between $\xi_i$ and $\xi_j$ as measured by $(\xi_i - \xi_j)^2$. This type of homophily can also be captured by a factor structure with $R=3$, $\alpha_i=(\xi_i^2,1,-2\xi_i)'$ and $\gamma_j = (1,\xi_j^2,\xi_j)$. The factor structure can also account for clustering or transitivity of links due to latent factors. Assume that there is a cluster of individuals  such that there are more flows within the cluster. This would be captured by a factor structure with $R=1$, $\alpha_i = \xi_i I_i$ and $\gamma_j = \chi_j I_j$, where $\xi_i$ and $\chi_j$ are positive cluster effects on the sender and receiver, and $I_i$ is an indicator for cluster membership. The factor structure can also account for combinations of these network features. Indeed, one of its advantages is that the researcher has the flexibility of specifying some features and leaving other features unspecified. For example, in the trade application we use a specification that includes additive effects to account explicitly for degree heterogeneity and multiple interactive effects to account for the possibility of having homophily in latent factors and clustering without explicitly modelling any of them. 
 
\medskip 
 

We consider three running examples throughout the analysis:
\begin{example}[Linear model]\label{example: linear} Let $Y_{ij}$ be a continuous outcome. We can model the conditional distribution of $Y_{ij}$ using the Gaussian linear model
\begin{equation*}
f(y \mid z_{ij}) = \varphi(z_{ij}/\sigma)/\sigma, \ \ y \in \mathbb{R},
\end{equation*}
where $\varphi$ is the density function of the standard normal and $\sigma$ is a positive scale parameter.
\end{example}

\begin{example}[Binary response model]\label{example: probij} Let $Y_{ij}$ be a binary outcome and $F$ be a cumulative
distribution function of  the standard normal or logistic distribution. We can model the conditional distribution of $Y_{ij}$ using the probit or logit model
\begin{equation*}
f(y \mid z_{ij}) = F(z_{ij})^y[1 - F(z_{ij})]^{1-y}, \ \ y \in \{0,1\}.
\end{equation*}
\end{example}


\begin{example}[Count response model]\label{example: poisson} Let $Y_{ij}$ be a count or non-negative integer-valued outcome, and $\psi(\cdot; \lambda)$ be the
probability mass function of a Poisson random variable with parameter $\lambda > 0$. We can model the conditional distribution of $Y_{ij}$ using the Poisson model
\begin{equation*}
f(y \mid z_{ij}) = \psi(y;  \exp[z_{ij} ]), \ \ y \in \{0, 1,  2, .... \}.
\end{equation*}
\end{example}

\subsection{Average Partial Effects} 

In addition to the model parameter $\beta$, we might be interested in average partial effects (APEs).  These effects
are averages of the data, parameters and unobserved effects. They measure the effect of the covariates on moments of the distribution of the outcome conditional on the covariates and unobserved effects. The leading case is the conditional expectation,
$$
\mathbb{E} [Y_{ij} \mid X_{ij}, \alpha_i, \gamma_j, \beta] = \int y f(y \mid X_{ij}'\beta + \pi_{ij} ) dy,
$$
where the partial effects are  differences or derivatives of this expression with respect to the components of $X_{ij}$. 
We denote generically the partial effects by $\Delta(Y_{ij},X_{ij}, \beta, \alpha'_i \gamma_j) = \Delta_{ij}(\beta,  \alpha_i' \gamma_j)$, where the restriction that they depend on $\alpha_i$ and $\gamma_j$ through $\pi_{ij}$ is natural  given the model for the conditional density of $Y_{ij}$. We allow the partial effect to depend on $Y_{ij}$ to cover scale and other parameters not included in the single-index.
 The APE is
 \begin{equation}\label{eq:ape}
 \delta = \E\left[ \frac 1 {n} \sum_{(i,j) \in {\cal D}}  \Delta_{ij}(\beta,  \alpha_i^{\prime} \gamma_j) \right].
\end{equation} 

\noindent  \textbf{Example \ref{example: linear}} (Linear model).  \textit{The variance $\sigma^2$ in the linear model can be expressed as an APE with
\begin{equation}\label{example: probit: meff1}
\Delta_{ij}(\beta, \alpha'_i \gamma_j) =  (Y_{ij} - X_{ij}'\beta - \alpha'_i \gamma_j)^2.
\end{equation}
}

\noindent  \textbf{Example \ref{example: probij}} (Binary response model).  \textit{If $X_{ij,k}$, the $k$th element of $X_{ij}$, is
binary, its partial effect on the conditional probability of $Y_{ij}$ is
\begin{equation}\label{example: probit: meff1}
\Delta_{ij}(\beta, \alpha'_i \gamma_j) =  F(\beta_k +
X_{ij,-k}'\beta_{-k}  + \alpha'_i  \gamma_j) -
F(X_{ij,-k}'\beta_{-k} + \alpha'_i   \gamma_j),
\end{equation}
where $\beta_k$ is the $k$th element of $\beta$, and $X_{ij,-k}$ and $\beta_{-k}$ include all elements of $X_{ij}$ and $\beta$ except for the $k$th element. If $X_{ij,k}$ is continuous and $F$ is differentiable, the partial effect of $X_{ij,k}$
on the conditional probability of $Y_{ij}$ is
\begin{equation}\label{example: probit: meff2}
\Delta_{ij}(\beta, \alpha'_i \gamma_j) = \beta_k  \partial F(X_{ij}'\beta +
\alpha'_i \gamma_j), \ \ \partial F(u) := \partial F(u)/\partial u.
\end{equation}
}

\medskip

\noindent  \textbf{Example \ref{example: poisson}} (Count response model). \textit{If $X_{ij,k}$, the $k$th element of $X_{ij}$, is
binary, its partial effect on the conditional probability of $Y_{ij}$ in the Poisson model is
\begin{equation}\label{example: poisson: meff1}
\Delta_{ij}( \beta, \alpha'_i \gamma_j) =  \exp(\beta_k +
X_{ij,-k}'\beta_{-k}  + \alpha'_i  \gamma_j) -
\exp(X_{ij,-k}'\beta_{-k} + \alpha'_i   \gamma_j),
\end{equation}
where $\beta_k$ is the $k$th element of $\beta$, and $X_{ij,-k}$ and $\beta_{-k}$ include all elements of $X_{ij}$ and $\beta$ except for the $k$th element. If $X_{ij,k}$ is continuous, the partial effect of $X_{ij,k}$
on the conditional expectation of $Y_{ij}$ is
\begin{equation}\label{example: poisson: meff}
\Delta_{ij}(\beta, \alpha'_i \gamma_j) = \beta_k
 \exp(X_{ij}'\beta + \alpha'_i   \gamma_j).
\end{equation}
}

\subsection{Fixed effects estimator} 
We adopt a fixed effects approach and treat the unobserved effects $\phi_n$ as a vector of nuisance parameters to be estimated. 
Let
\begin{align*} 
     L(\beta,\phi_n) :=
              \sum_{(i,j) \in {\mathcal{D}}}  \log f(Y_{ij}  \mid X_{ij}' \beta + \pi_{ij}  )
\end{align*}
be the conditional log-likelihood function of the data  constructed from the parametric part of the model. The fixed effects estimator is  
\begin{equation}
  (\widehat \beta, \widehat \phi_n)
  \, \in \,
\argmax_{(\beta, \phi_n) \in  \mathbbm{R}^{d_x + R(I+J)}} \;  L(\beta,\phi_n).
   \label{LobjMAX}
\end{equation}
This problem has a unique solution with probability one for $\beta$ under the assumption that  $z \mapsto \log f(\cdot  \mid z)$ is concave. This assumption holds for all the cases that we consider including  logit, probit, ordered probit and Poisson models. The solution for $\phi_n$ is only unique up to normalization -- see Remark \ref{remark:norm} below. Obtaining the solution to \eqref{LobjMAX} can be computationally challenging because the objective function is not concave in  the parameter $\phi_n$ and the high-dimensionality of the parameter space. In Section \ref{sec:implementation} we provide an iterative method based on \cite{C14} to obtain the estimates. This method performs well in simulations. 



\medskip

%

%

Let $\widehat \phi_n = ({\rm vec}(\widehat \alpha)', {\rm vec}(\widehat \gamma)')' $, where $\widehat \alpha$ and $\widehat \gamma$ correspond to the components $\alpha$ and $\gamma$ such that $\widehat \alpha = (\widehat \alpha_1, \ldots, \widehat \alpha_I)'$ and  $\widehat \gamma = (\widehat \gamma_1, \ldots, \widehat \gamma_J)'$.
Plugging the estimator of $(\beta,\phi_n)$ in \eqref{eq:ape} yields the estimator of the APE,
\begin{equation}\label{eq:apehat}
\widehat \delta    =  \frac 1 {n} \sum_{(i,j) \in {\cal D}} \Delta_{ij}(\widehat \beta,  \widehat \alpha_i'  \widehat \gamma_j).
\end{equation}
 In Section \ref{sec:app_panel}, we show that $\widehat \beta$ and $\widehat \delta$ are consistent and normally distributed in large samples, but might have incidental parameter bias because the dimension of the nuisance parameter $\phi_n$ grows with the sample size \citep{NeymanScott1948}. 

\begin{remark}[Normalization]\label{remark:norm} As in linear factor models, the solution to the problem \eqref{LobjMAX} for 
 $\phi_n = ({\rm vec}(\alpha)', {\rm vec}(\gamma)')'$ is only unique up to normalization  because the log-likelihood function is invariant under the transformation
$\alpha \mapsto \alpha A'$ and $\gamma \mapsto \gamma A^{-1}$ for any non-singular $R \times R$ matrix $A$. The estimators $\widehat \beta$ and $\widehat \delta$ are  invariant to the normalization used to eliminate this indeterminancy.  Moreover, we can always reparametrize the model  in \eqref{eq:model} with respect to $\phi_n$ in a way that the true value of $\phi_n$  satisfies the adopted normalization. This invariance  allows us to choose different normalizations for different purposes. For example, we use a standard normalization for linear factor models in the computation of the estimators, whereas  we employ another normalization to derive the asymptotic distributions of the estimators in the Appendix.   We refer to  \cite{rs15} for a discussion on the effect of the normalization in the context of  linear factor models.  
\end{remark}

\section{A Simple Motivating Example} \label{sec:numeric} 
We  illustrate the statistical issues that arise in the estimation of factor models with a simple example. This example is analytically tractable  and might be of practical interest as it provides an estimator of the variance of a random variable in network and panel data allowing for flexible patterns of dependence. The analysis in this section is mainly heuristic leaving technical details such as the derivation of the  orders of some remainder terms in the asymptotic expansions for Section \ref{sec:app_panel}.  

\medskip

Consider a version of  Example \ref{example: linear} without covariates where $Y_{ij} \mid  \phi_n \sim \mathcal{N}(\alpha'_i  \gamma_j, \sigma^2)$. Assume that the observations $Y_{ij}$ are independent over $i$ and $j$, and that there is no missing data, i.e. ${\cal D}= \mD_0$. The quantity of interest is the scale parameter $\sigma^2$, which can be treated as an APE. This is a linear factor model where $\widehat \phi_n$ can be obtained  using  the principal component algorithm of \cite{Bai:2009p3321}. Then, the plug-in estimator of $\sigma^2$ is
\begin{equation}\label{eq:ns}
\widehat \sigma^2 = \frac{1}{IJ} \sum_{i=1}^I \sum_{j=1}^J \left( Y_{ij} - \widehat \alpha_i' \widehat \gamma_j \right)^2.
\end{equation}

\medskip

To analyze the properties of $\widehat \sigma^2$, it is useful to consider an asymptotic expansion of  $\widehat \alpha_i' \widehat \gamma_j$ around $\alpha_i' \gamma_j$ as $I,J \to \infty$. This yields
\begin{multline*}
\widehat \alpha_i' \widehat \gamma_j =  \alpha_i' \gamma_j + (\widehat \alpha_i - \alpha_i)'  \gamma_j  +   \alpha_i' (\widehat \gamma_j -\gamma_j) + (\widehat \alpha_i - \alpha_i)'  (\widehat \gamma_j - \gamma_j) \\
\approx \alpha_i' \gamma_j + (\widehat \alpha_i - \alpha_i)'  \gamma_j  +   \alpha_i' (\widehat \gamma_j -\gamma_j),
\end{multline*}
where $\approx$ means equal up to terms of lower order. Plugging this expansion in  \eqref{eq:ns} shows that $\widehat \sigma^2$ behaves asymptotically as a sample variance with $R(I+J)$ estimated fixed effects corresponding to the $\widehat \alpha_i$'s and $\widehat \gamma_t$'s. Then, standard degrees of freedom calculations give
\begin{equation}\label{eq:nsbias}
\E[\widehat \sigma^2] \approx \frac{(I-R)(J-R)}{IJ} \sigma^2 \approx \sigma^2 -  \frac{R(I+J)}{IJ} \sigma^2,
\end{equation}
which shows that $\widehat \sigma^2$ has an incidental parameter bias that grows proportionally to the number of factors $R$. The order of the bias corresponds to the number of estimated parameters, $R(I+J)$, divided by the number of observations, $IJ$, as predicted by the general formula in \cite{ARE} for fixed effects estimators. We show in numerical examples that this expression produces a very accurate approximation  to the bias even for small sample sizes. 

\medskip

We carry out 50,000 simulations with $\sigma^2 = 1$, and $\alpha_i$ and $\gamma_j$ drawn independently from multivariate normal distributions with mean zero and covariance function $\mathbb{I}_R$, the identity matrix of order $R$. Table \ref{table:ns1} compares the bias of $\widehat \sigma^2$ with the asymptotic approximation \eqref{eq:nsbias} in datasets with $I,J \in \{10, 25, 50\},$ and $R \in \{1,2,3\}$. We only report the results for $J \leq I$ since all the expressions are symmetric in $I$ and $J$.   Comparing the two rows in each panel of the table, we find that the asymptotic bias provides a very accurate approximation to the finite-sample bias of the estimator for all the sample sizes and numbers of factors.  
   
\begin{table}
\begin{center}\caption{\label{table:ns1} Asymptotic and Exact Bias of $\widehat \sigma^2$ }
\begin{tabular}{ccccccc} \hline\hline
&   $I = 10$ &  \multicolumn{2}{c}{$I=25$} & \multicolumn{3}{c}{$I=50$}  \\
  Bias  &   $J = 10$ & $J=10$  & $J=25$  & $J=10$  & $J=25$  & $J=50$  \\\hline
  & \multicolumn{6}{c}{$R=1$} \\
   \multicolumn{1}{l}{Asymptotic}   &  -.19 & -.14  & -.08  & -.12  & -.06  & -.04  \\
  \multicolumn{1}{l}{Exact}            &  -.20 & -.14  & -.08  & -.12  & -.06  & -.04  \\
 & \multicolumn{6}{c}{$R=2$} \\
  \multicolumn{1}{l}{Asymptotic}   &  -.36 & -.26  & -.15  & -.23  & -.12  & -.08  \\
  \multicolumn{1}{l}{Exact}            &  -.39 & -.27  & -.16  & -.24  & -.12  & -.08  \\
  & \multicolumn{6}{c}{$R=3$} \\
  \multicolumn{1}{l}{Asymptotic}   &  -.51 & -.38  & -.23  & -.34  & -.17  & -.12  \\
  \multicolumn{1}{l}{Exact}            &  -.55 & -.40  & -.23  & -.35  & -.18  & -.12  \\
\hline\hline
  \multicolumn{7}{l}{ \footnotesize{Notes: Results obtained by 50,000 simulations}}\\
 \multicolumn{7}{l}{\footnotesize{Design: $Y_{ij} \mid \phi_n \sim \mathcal{N}(\alpha_i' \gamma_t, \sigma^2)$, $\alpha_i \sim N(0,\mathbb{I}_{R})$, $\gamma_j \sim N(0,\mathbb{I}_{R})$, $\sigma^2=1$}}
  \end{tabular}
\end{center}
\end{table}

\medskip

The bias of $\widehat \sigma^2$ can be removed using analytical and split-sample methods. Thus, an analytical bias corrected estimator can be formed as
$$
\widetilde \sigma_{\rm ABC}^2 = \frac{IJ}{(I-R)(J-R)} \widehat \sigma^2.
$$
A split-sample bias corrected estimator can be formed as
$$
\widetilde \sigma_{\rm SBC}^2 =   3 \widehat{\sigma}^2 - \bar{\sigma}^2_{I,J/2} - \bar{\sigma}^2_{I/2,J},
$$
where $\bar{\sigma}^2_{I,J/2}$ is the average of the estimators in the half-panels $ \{(i,j) :  i = 1,\dots,I; j = 1, \dots, \lceil J/2 \rceil \}$ and $ \{(i,j) :  i = 1,\dots,I; j = \lfloor J/2 + 1\rfloor, \dots, J \}$, and $\bar{\sigma}^2_{I/2,J}$ is the average of the estimators in the half-panels $ \{(i,j) :  i = 1,\dots,\lceil I/2 \rceil; j = 1, \dots, J \}$ and $ \{(i,j) :  i = \lfloor I/2 + 1\rfloor,\dots,I; j = 1, \dots, J \}$, where $\lceil \cdot \rceil$ and $\lfloor \cdot \rfloor$ are the ceil and floor functions. As in nonlinear panel data, we expect these corrections to remove most of the bias of the estimator without increasing dispersion. Moreover, constructing confidence intervals around the corrected estimators should help bring coverage probabilities close to their nominal levels. We confirm these predictions in a numerical simulation.

\begin{table}
\begin{center}\caption{\label{table:ns2} Bias, SD, RMSE and Coverage Probabilities}
\begin{tabular}{lcccccccc} \hline\hline
  &   Bias & SD & RMSE & Cover &   Bias & SD & RMSE & Cover  \\\hline
  & \multicolumn{4}{c}{$I= 10, J=10$} &  \multicolumn{4}{c}{$I = 25, J=10$} \\
  $\widehat{\sigma}^2$                &-0.55 & 0.09 & 0.56 & 0.00 & -0.40 & 0.07 & 0.41 & 0.00 \\
  $\widetilde{\sigma}_{\rm ABC}^2$   &-0.08 & 0.19 & 0.20 & 0.75 &-0.02 & 0.11 & 0.11 & 0.85 \\
  $\widetilde{\sigma}_{\rm SBC}^2$   &-0.09 & 0.20 & 0.22 & 0.71 &-0.03 & 0.12 & 0.13 & 0.81 \\
  & \multicolumn{4}{c}{$I= 25, J=25$} &  \multicolumn{4}{c}{$I = 50, J=10$} \\
  $\widehat{\sigma}^2$                &-0.23 & 0.05 & 0.24 & 0.01 & -0.35 & 0.05 & 0.35 & 0.00 \\
  $\widetilde{\sigma}_{\rm ABC}^2$   &-0.01 & 0.06 & 0.06 & 0.91 &-0.01 & 0.08 & 0.08 & 0.88 \\
  $\widetilde{\sigma}_{\rm SBC}^2$   &-0.02 & 0.07 & 0.07 & 0.85 &-0.01 & 0.08 & 0.08 & 0.85 \\
  & \multicolumn{4}{c}{$I= 50, J=25$} &  \multicolumn{4}{c}{$I = 50, J=50$} \\
  $\widehat{\sigma}^2$                &-0.18 & 0.04 & 0.18 & 0.00 &-0.12 & 0.03 & 0.12 & 0.01 \\
  $\widetilde{\sigma}_{\rm ABC}^2$   &-0.00 & 0.04 & 0.04 & 0.92 &-0.00 & 0.03 & 0.03 & 0.93 \\
  $\widetilde{\sigma}_{\rm SBC}^2$   &-0.01 & 0.05 & 0.05 & 0.88 &-0.00 & 0.03 & 0.03 & 0.92 \\
\hline\hline
  \multicolumn{9}{l}{ \footnotesize{Notes: 50,000 simulations. Nominal level is $0.95$}}\\
 \multicolumn{9}{l}{\footnotesize{Design: $Y_{ij} \mid \phi_n \sim \mathcal{N}(\alpha_i' \gamma_t, \sigma^2)$, $\sigma^2=1$, $\alpha_i \sim N(0,\mathbb{I}_{R})$, $\gamma_j \sim N(0,\mathbb{I}_{R})$, $R=3$ }}
  \end{tabular}
\end{center}
\end{table}

\medskip

Table \ref{table:ns2} reports the bias, standard deviation and RMSE of the uncorrected and bias corrected estimators, together with coverage probabilities of 95\% confidence interval constructed around them. The results are based on 50,000 simulations of datasets generated as in Table \ref{table:ns1} with $I, J \in \{10, 25, 50\}$, and $R = 3$. The confidence intervals around the estimator $\widetilde \sigma^2 \in \{ \widehat \sigma^2, \widetilde \sigma_{\rm ABC}^2, \widetilde \sigma_{\rm SBC}^2\} $ are constructed as $\widetilde \sigma^2(1 \pm 1.96 \sqrt{2/(IJ)})$, where we use that the asymptotic variance of all the estimators is $2\sigma^4/(IJ)$. We find that the corrections offer huge improvements in terms of bias reduction and coverage of the confidence intervals. The corrections  increase the dispersion  for small sample sizes, but always reduce the RMSE. In this case the analytical correction slightly outperforms the split-sample correction.

\section{Asymptotic Theory}
\label{sec:app_panel}

We derive the asymptotic distribution of the estimators of the model parameter and APEs under sequences where $I$ and $J$ grow with the sample size at the same rate. We focus on these sequences because they are the only ones that deliver a non-degenerate limit distribution. Moreover, they are very natural choices for network data where $I = J$. Throughout this section, all the stochastic statements are conditional on the realization of the unobserved effects $\phi_n$ and should therefore  be qualified with almost surely. We shall omit this qualifier to lighten the notation.

\subsection{Model parameter}
We consider single-index models with strictly exogenous covariates and  unobserved effects that enter the density of the outcome through  $z_{ij} = X_{ij}'\beta + \pi_{ij}$, where $\pi_{ij} = \alpha_i'\gamma_j$.  These models cover  the
linear, probit and Poisson specifications of Examples~\ref{example: linear}--\ref{example: poisson}.   We focus on strictly exogenous covariates because for some data structures of interest such as network data  there is no natural ordering of the observations. The results can be extended to predetermined covariates when one of the dimensions is time, see the earlier version of the paper \citep{CFW14}. 
Let 
\begin{equation}\label{eq: index_model}
   \ell_{ij}(z_{ij} ) := \log f(Y_{ij} \mid X_{ij}, \beta, \alpha_i, \gamma_j)
\end{equation}
be the conditional log-likelihood coming from the parametric part of the model. 
We denote  the derivatives  of $z \mapsto \ell_{ij}(z)$ by $\partial_{z^q} \ell_{ij}(z) := \partial^q \ell_{ij}(z)/\partial z^q$, $q = 1,2, \ldots$. Let $\beta^0$, $\alpha_i^0$, $\gamma_j^0$,   and $\pi_{ij}^0 = \alpha_i^{0 \prime} \gamma_j^0$ denote the values of $\beta$, $\alpha_i$, $\gamma_j$, and $\pi_{ij}$ that generated the data. 
We drop the argument $z_{ij}$ when the derivatives are evaluated at the true value of the index  $z^0_{ij} := X_{ij}'\beta^0 + \pi_{ij}^0$, i.e.,
$\partial_{z^q} \ell_{ij} := \partial_{z^q}\ell_{ij}(z^0_{ij})$. 
Let $\boldsymbol{X}=\{X_{ij} : (i,j) \in {\cal D}\}$, $\alpha^0 = (\alpha_1^0,\ldots, \alpha_I^0)'$, and
 $\gamma^0 = (\gamma_1^0,\ldots, \gamma_J^0)'$ .

\medskip

We make the following assumptions:
\begin{assumption}[Nonlinear Factor Model]
   \label{ass:PanelA1}
   Let $\varepsilon>0$ and let 
     ${\cal B}^0_{\varepsilon}$ be a bounded subset of $\mathbbm{R}$
     that contains an $\varepsilon$-neighborhood of $z^0_{ij}$
     for all $i,j,I,J$.
          
     \begin{itemize}
         \item[(i)] Model:   $Y_{ij}$ is distributed as 
\begin{equation*}
Y_{ij}   \mid  \boldsymbol{X}, \beta^0, \alpha^0, \gamma^0  \sim \exp[\ell_{ij}( X_{ij}'\beta^0 + \pi^0_{ij})],
\end{equation*}
and conditional on $(\boldsymbol{X}, \beta^0, \alpha^0, \gamma^0)$, either (a) $Y_{ij}$ is independent across  $(i,j) \in {\mathcal{D}}$  or (b)  $(Y_{ij},Y_{ji})$ is independent across observations $(i,j) \in {\mathcal{D}}$ with $i \leq j$.
The number of factors $R$ is known. 

      \item[(ii)] Asymptotics: we consider limits of sequences where  $I_n/J_n \rightarrow \kappa^2$, $0<\kappa<\infty$, as $n = | \mD | \rightarrow \infty$. We shall drop the indexing by $n$ from $I_n$ and $J_n$ in the following.

      \item[(iii)]  Smoothness and moments:
            $z \mapsto \ell_{ij}(z)$ is
      four times continuously differentiable over ${\cal B}^0_{\varepsilon}$ a.s. 
      and  $\max_{i,j} \E[|\partial_{z^q} \ell_{ij}(z_{ij}^0)|^{8+\nu}]$, $q \leq 4$,  are  uniformly bounded over  $I,J$ for some $\nu > 0$.
     In addition, $X_{ij}$ is bounded uniformly  over $i,j,I,J$.

        \item[(iv)] Concavity: 
        for all $I,J,$ the function
    $z \mapsto  \ell_{ij} (z)$ is strictly concave over $z \in \mathbbm{R}$ a.s.
    Furthermore, there exist positive constants $b_{\min}$
    and $b_{\max}$ such that for all $z \in {\cal B}^0_{\varepsilon}$,
    $b_{\min} \leq -  \partial_{z^2} \ell_{ij}(z)    \leq b_{\max}$ a.s.  uniformly over $i,j,I,J$.  
    
       \item[(v)] Strong factors: 
       $I^{-1} \sum_{i=1}^I  \alpha^0_i    \alpha^{0 \, \prime}_i  \to_P \Sigma_1 >0,$
   and
   $J^{-1} \sum_j  \gamma^0_j   \gamma^{0 \, \prime}_j \to_P \Sigma_2 >0$.

     \item[(vi)] Generalized non-collinearity:
      for any matrix $A$, define the coprojection matrix as ${\cal M}_A := \mathbb{I} - A(A'A)^{\dag}A'$, where $ \mathbb{I}$ denotes the identity matrix of
       appropriate size and the superscript $^{\dag}$ denotes the Moore-Penrose generalized inverse. Let $\alpha^0 := (\alpha_1^0, \ldots, \alpha_I^0)'$ and $\mathbb{X}_k$ be a $I \times J$ matrix with elements $X_{ij,k}$, $i = 1,\ldots,I,$ $j = 1,\ldots,J$. The $d_x \times d_x$ matrix $D(\gamma)$ with elements 
        $$ D_{k_1k_2}(\gamma) =
        (IJ)^{-1}   
        {\rm Tr}( {\cal M}_{\alpha^0} \mathbb{X}_{k_1} {\cal M}_{\gamma}  \mathbb{X}_{k_2}'), \ \ k_1, k_2 \in \{1,...,d_x \},$$
        satisfies
        $D(\gamma) > c >0 $ for all $\gamma \in \mathbbm{R}^{J \times R}$, wpa1.

     \item[(vii)]  Missing data:   there is a finite number of missing observations for every $i$ and $j$, that is,
      $\max_{i} ( J-|\{(i',j') \in {\cal D} : i' = i\}| ) \leq C$
     and $\max_{j}( I - |\{(i',j') \in {\cal D} : j' = j\}| ) \leq C$
     for some constant $C<\infty$ that is independent of the sample size.

    \end{itemize}

\end{assumption}

The two cases considered in Assumption \ref{ass:PanelA1}$(i)$ are designed for different data structures.  Case (b) is more suitable for network data because it allows for reciprocity between the observations $(i,j)$ and $(j,i)$, whereas  case (a) is more suitable for panel data where there is no special relationship between these observations. Assumption \ref{ass:PanelA1}$(i)$ also imposes that the number of factors is known. We provide a practical method to choose the number of factors in Section \ref{sec:implementation}. We also recommend  checking the sensitivity to this number by reporting the maximum value of the average log-likelihood and the parameter estimates for multiple values of $R$. We provide an example in the empirical application of Section \ref{sec:empirics}.
Assumption \ref{ass:PanelA1}$(i)-(iii)$ are similar to \cite{FW16}, so we do not discuss them further here.   The concavity condition in Assumption \ref{ass:PanelA1}$(iv)$ holds for the logit, probit, ordered probit and Poisson models.  The strong factor and generalized noncollinearity conditions in  Assumption \ref{ass:PanelA1}$(v)-(vi)$ were previously imposed in \cite{Bai:2009p3321} and \cite{MoonWeidner2015,MoonWeidner2017}  for linear models with interactive effects. Generalized noncollinearity rules out covariates that do not display variation in the two dimensions of the dataset. \cite{BL17} and \cite{AB16} impose very similar conditions to Assumption \ref{ass:PanelA1}, so we refer to these papers for further discussion.

\medskip

We introduce more notation that is convenient to simplify the expressions in the asymptotic distribution. Let  $\Xi_{ij}$ be a $d_x$-dimensional vector defined by the following population weighted least squares projection
for each component of $\E( \partial_{z^2} \ell_{ij} X_{ij})$,
\begin{align*}
   \Xi_{ij,k} &=  \alpha^{\ast \, \prime}_{i,k}  \gamma^0_j + \alpha^{0 \, \prime}_i \gamma^{\ast}_{j,k} ,
   &
  \left(\alpha^{\ast}_{i,k} , \, \gamma^{\ast}_{j,k} \right)
   &\in \argmin_{\alpha_{i,k},\gamma_{j,k}} \sum_{i,j}
  \E( - \partial_{z^2} \ell_{ij} )
        \left(  \frac{\E( \partial_{z^2} \ell_{ij} X_{ij,k})}
                        {\E( \partial_{z^2} \ell_{ij} )}   
         - \alpha'_{i,k}  \gamma^0_j - \alpha^{0 \, \prime}_i \gamma_{j,k} \right)^2.
\end{align*}  
Also define the residual of the projection
\begin{align*}
     \tilde X_{ij}:= X_{ij} - \Xi_{ij} .
\end{align*}    
Finally, let $\EE  := \plim_{I,J \to \infty}$,
${\cal D}_i := \{ j \, : \, (i,j) \in {\cal D} \}$ and ${\cal D}_j := \{ i \, : \, (i,j) \in {\cal D} \}$.

\medskip

The following theorem establishes the asymptotic distribution of  $\widehat \beta$ defined in \eqref{LobjMAX}.


\begin{theorem}[Asymptotic distribution of $\widehat \beta$]
  \label{th:BothEffects}
Suppose that Assumption \ref{ass:PanelA1} holds, that the following limits exist
\begin{align*}
   \overline B_{\infty} &=
      -  \EE \left\{  \frac {1} {I}  \sum_{(i,j) \in {\cal D}}
            \gamma^{0 \, \prime}_j  
            \left[  \sum_{h \in {\cal D}_i}   \gamma^{0}_h \gamma^{0 \prime}_h \, \E\left(  \partial_{z^2} \ell_{ih} \right)
            \right]^{-1}
             \gamma^{0}_j \;
        \E\left(
                \partial_{z} \ell_{ij}  \partial_{z^2} \ell_{ij} \tilde X_{ij}
                + \frac 1 2  \partial_{z^3} \ell_{ij} \tilde X_{ij}
                 \right)  
    \right\}  ,
        \nonumber \\
      \overline D_{\infty} &=   -  \EE \left\{  \frac {1} {J}  \sum_{(i,j) \in {\cal D}}
            \alpha^{0\, \prime}_i  
            \left[ \sum_{h \in {\cal D}_j}   \alpha^{0}_h \alpha^{0 \prime}_h \, \E\left(  \partial_{z^2} \ell_{hj} \right)
            \right]^{-1}
             \alpha^{0}_i \;
        \E\left(
                \partial_{z} \ell_{ij}  \partial_{z^2} \ell_{ij} \tilde X_{ij}
                + \frac 1 2  \partial_{z^3} \ell_{ij} \tilde X_{ij}
                 \right)  
    \right\} , \\
          \overline W_{\infty}  &=  - \EE \left[  \frac {1} {n}  \sum_{(i,j) \in {\cal D}}  \E \left(
              \partial_{z^2}  \ell_{ij} \tilde X_{ij} \tilde X'_{ij} \right) \right], 
               \\
          \overline \Sigma_{\infty}  &=   \EE \left[  \frac {1} {n}  \sum_{(i,j) \in {\cal D}}  \E \left\{ \left(
              \partial_{z}  \ell_{ij} \tilde X_{ij}  +  \partial_{z}  \ell_{ji} \tilde X_{ji} \right)    \partial_{z}  \ell_{ij}  \tilde X'_{ij} \right\} \right],
   \end{align*}
and that $\overline W_{\infty}>0$.
Then,
   \begin{align*}
      \sqrt{n} \left( \widehat \beta - \beta^0 
      - \frac{I} n \, \overline{W}_{\infty}^{-1} \overline B_{\infty} 
       - \frac{J} n \, \overline{W}_{\infty}^{-1} \overline D_{\infty}  \right)
          \; \to_d \;
      {\cal N}( 0 ,
           \;\overline W_{\infty}^{-1}  \overline \Sigma_{\infty} \overline W_{\infty}^{-1}).
    \end{align*}
\end{theorem}

\medskip

\begin{remark}[Panel Data] In case (a) of Assumption \ref{ass:PanelA1}$(i)$, the asymptotic variance of $\widehat \beta$ simplifies to
$$
\overline W_{\infty}^{-1}  \overline \Sigma_{\infty} \overline W_{\infty}^{-1} = - \overline W_{\infty}^{-1},
$$
by the fact that the scores $ \partial_{z}  \ell_{ij} \tilde X_{ij}$ and $ \partial_{z}  \ell_{ji} \tilde X_{ji}$ are uncorrelated and the information equality.
\end{remark}

Theorem  \ref{th:BothEffects} shows that $\widehat \beta$ is consistent and normally distributed, but can have bias of the same order as its standard deviation. The scaling factors in the expressions for $\overline B_{\infty}$ and $\overline D_{\infty}$ are such that those expressions
are of order one, for example, we can express $\overline B_{\infty}$ equivalently as 
$$
    -  \EE \left\{  \frac {1} {I}  \sum_{i=1}^I \frac 1 {|{\cal D}_i|} \sum_{j \in {\cal D}_i} 
            \gamma^{0 \, \prime}_j  
            \left[  \frac 1 {|{\cal D}_i|} \sum_{h \in {\cal D}_i}   \gamma^{0}_h \gamma^{0 \prime}_h \, \E\left(  \partial_{z^2} \ell_{ih} \right)
            \right]^{-1}
             \gamma^{0}_j \;
        \E\left(
                \partial_{z} \ell_{ij}  \partial_{z^2} \ell_{ij} \tilde X_{ij}
                + \frac 1 2  \partial_{z^3} \ell_{ij} \tilde X_{ij}
                 \right)  
    \right\} ,
$$
where all sums explicitly appear as part of a  sample average.
 We verify the presence of  bias  in our running examples.

\noindent  \textbf{Example \ref{example: linear}} (Linear model).\textit{ In this case
$$
\ell_{ij}(z) = - \frac{1}{2} \log (2 \pi \sigma^2) - \frac{(Y_{ij} -z_{ij})^2}{2\sigma^2},
$$
so that $\partial_{z} \ell_{ij} = (Y_{ij} -z_{ij}^0)/\sigma^2$,  $ \partial_{z^2} \ell_{ij} =  -1/\sigma^2$, and $ \partial_{z^3} \ell_{ij} = 0$. Substituting these values in the expressions of the bias of Theorem \ref{th:BothEffects} yields
$
 \overline B_{\infty} =  \overline D_{\infty} = 0,
$
which agrees  with the result in \cite{Bai:2009p3321} of no asymptotic bias for $\beta$ in homoskedastic linear models with interactive effects and strictly exogenous covariates.
}

\medskip
 
\noindent  \textbf{Example \ref{example: probij}} (Binary response model). \textit{In this case 
$$\ell_{ij}(z) = Y_{ij} \log F(z) + (1 - Y_{ij}) \log [1 -  F(z)],$$ so that $\partial_{z} \ell_{ij} = H_{ij} (Y_{ij} - F_{ij}),$   $ \partial_{z^2} \ell_{ij} =   -  H_{ij} \partial F_{ij} + \partial H_{ij} (Y_{ij} - F_{ij})$,  and $\partial_{z^3} \ell_{ij} =  -  H_{ij} \partial^2 F_{ij} - 2 \partial H_{ij} \partial F_{ij} + \partial^2 H_{ij} (Y_{ij} - F_{ij})$, where $H_{ij} = \partial F_{ij} / [F_{ij}(1-F_{ij})], and $   $\partial^{j}  G_{ij} := \partial^{j} G(Z)|_{Z =z_{ij}^0}$ for any function $G$ and $j = 0,1,2$. Substituting these values in the expressions of the bias of Theorem \ref{th:BothEffects} for the probit model yields  
\begin{align*}
   \overline B_{\infty} &=
        \EE \left\{  \frac {1} {2\, I}  \sum_{(i,j) \in {\cal D}}
            \gamma^{0 \, \prime}_j  
            \left[  \sum_{h \in {\cal D}_i}   \gamma^{0}_h \gamma^{0 \prime}_h \, \E\left(  \partial_{z^2} \ell_{ih} \right)
            \right]^{-1}
             \gamma^{0}_j \;
        \E\left(\partial_{z^2} \ell_{ij} \tilde{X}_{ij} \tilde{X}_{ij}'
                 \right)  
    \right\}  \beta^0  ,
        \nonumber \\
      \overline D_{\infty} &=     \EE \left\{  \frac {1} {2\, J}  \sum_{(i,j) \in {\cal D}}
            \alpha^{0\, \prime}_i  
            \left[  \sum_{h \in {\cal D}_j}   \alpha^{0}_h \alpha^{0 \prime}_h \, \E\left(  \partial_{z^2} \ell_{hj} \right)
            \right]^{-1}
             \alpha^{0}_i \;
        \E\left(\partial_{z^2} \ell_{ij} \tilde{X}_{ij} \tilde{X}_{ij}'
                 \right)  
    \right\}  \beta^0.
   \end{align*}   
 The asymptotic bias  is therefore a positive definite matrix weighted average of the true parameter value as in the case of the probit model with additive  individual and time effects in \cite{FW16}. The bias grows linearly with the number of factors because
\begin{equation}\label{eq:ex3simplification}
  \sum_{j \in {\cal D}_i} 
            \gamma^{0 \, \prime}_j  
            \left[  \sum_{h \in {\cal D}_i}   \gamma^{0}_h \gamma^{0 \prime}_h 
            \right]^{-1}
             \gamma^{0}_j =  \sum_{i \in {\cal D}_j}  \alpha^{0\, \prime}_i  
            \left[  \sum_{h \in {\cal D}_j}   \alpha^{0}_h \alpha^{0 \prime}_h  \right]^{-1}
             \alpha^{0}_i = R,
\end{equation}
 and $\E\left(  \partial_{z^2} \ell_{i j} \right)$ and $\E\left(\partial_{z^2} \ell_{ij} \tilde{X}_{ij} \tilde{X}_{ij}'
                 \right)  $ are bounded uniformly in $i,j$. 
 }

\medskip

\noindent  \textbf{Example \ref{example: poisson} } (Count response model). \textit{In this case 
$$\ell_{ij}(z) = z Y_{ij} - \exp(z) - \log Y_{ij}!,$$
where the symbol $!$ denotes the factorial function, so that $\partial_{z} \ell_{ij} = Y_{ij} - \lambda_{ij}$ and  $\partial_{z^2} \ell_{ij} = \partial_{z^3} \ell_{ij}= - \lambda_{ij} $, where $\lambda_{ij} = \exp(z_{ij}^0)$. Substituting these values in the expressions of the bias of Theorem \ref{th:BothEffects} yields  
$$\overline B_{\infty} = \overline D_{\infty} = 0,$$ which generalizes the result in   \cite{FW16} of no asymptotic bias in the Poisson model with strictly exogenous covariates and  additive individual and time effects to the Poisson model with strictly exogenous covariates and factor structure. }

\subsection{Average Partial Effects}\label{subsec:apes}

We use additional assumptions to derive the asymptotic distribution of the estimator of the APEs. They involve smoothness conditions on the partial effect function $(\beta,\pi) \mapsto \Delta_{ij}(\beta,\pi)$ needed to obtain the limit distribution of $\widehat \delta$ from the limit distribution of $(\widehat \beta, \widehat \phi_n)$ via  delta method. For a vector of nonnegative integer numbers $v = (v_1, \ldots, v_{d_x})$, let $\partial_{\beta^v}:= \partial^{|v|} / \partial \beta_1^{v_1} \cdots \partial \beta_{d_x}^{v_{d_x}}$ and $|v| = v_1 + \ldots + v_{d_x}$. 

\begin{assumption} [Partial effects]   \label{ass:PanelA2} ~
  Let $\epsilon>0$, and let
     ${\cal B}^0_{\varepsilon}$ be a subset of $\mathbbm{R}^{d_x+1}$
     that contains an $\varepsilon$-neighborhood of $(\beta^0,\pi^0_{ij})$
     for all $i,j,I,J$.
  
  \begin{itemize}
  \item[(i)] Model: for all $i,j,I,J,$
  the partial effects depend on $\alpha_i$ and $\gamma_j$ through  $\pi_{ij} = \alpha_i'  \gamma_j$:  
\begin{equation*}
\Delta(Y_{ij},X_{ij}, \beta, \alpha_i, \gamma_j) = \Delta_{ij}(\beta, \pi_{ij}),
\end{equation*}
where $(\beta,\pi) \mapsto \Delta_{ij}(\beta,\pi)$ is a known real-valued function. 
The realizations of the partial effects are denoted by $ \Delta_{ij} := \Delta_{ij}(\beta^0, \pi_{ij}^0).$

  \item[(ii)] Smoothness and moments: The function
       $(\beta,\pi) \mapsto \Delta_{ij}(\beta,\pi)$
        is
      four times continuously differentiable over ${\cal B}^0_{\varepsilon}$ a.s., 
      and $\max_{i,j} \E[|\partial_{\beta^v \pi^q} \ell_{ij}(\beta^0,z_{ij}^0)|^{8+\nu}]$, $|v| + q \leq 4$, are uniformly bounded over  $I,J$ for some $\nu >0$. 
     
 \end{itemize}
\end{assumption}

It is convenient again to introduce notation to simplify the expressions in the asymptotic distribution. Let $\Psi_{ij}$ be the weighted least squares population projection
\begin{align*}
   \Psi_{ij} &=  \alpha^{\ast \, \prime}_{i}  \gamma^0_j + \alpha^{0 \, \prime}_i \gamma^{\ast}_{j} ,
   &
  \left(\alpha^{\ast}_i , \, \gamma^{\ast}_j \right)
   & \in \argmin_{\alpha_{i},\gamma_{j}} \sum_{i,j}
  \E( - \partial_{z^2} \ell_{ij} )
        \left(  \frac{\E(  \partial_{\pi} \Delta_{ij} )}
                        {\E( \partial_{z^2} \ell_{ij} )}   
         - \alpha'_{i}  \gamma^0_j - \alpha^{0 \, \prime}_i \gamma_{j} \right)^2.
\end{align*}
We denote  the partial derivatives  of  $(\beta,\pi) \mapsto \Delta_{ij}(\beta,\pi)$ by $\partial_\beta \Delta_{ij}(\beta, \pi) := \partial  \Delta_{ij}(\beta,\pi)/\partial \beta$,
$\partial_{\beta \beta'} \Delta_{ij}(\beta, \pi) := \partial^2  \Delta_{ij}(\beta,\pi)/(\partial \beta \partial \beta')$,
$\partial_{\pi^q} \Delta_{ij}(\beta, \pi) := \partial^q \Delta_{ij}(\beta,\pi)/\partial \pi^q$, $q = 1,2,3,\ldots$.
We drop the arguments $\beta$ and $\pi$ when the derivatives are evaluated at the true
values $\beta^0$ and $\pi^0_{ij}$, e.g.
$\partial_{\pi^q} \Delta_{ij} := \partial_{\pi^q}\Delta_{ij}(\beta^0,\pi^0_{ij})$. 
We also define
$D_{\pi}  \Delta_{ij}  := \partial_{\pi}  \Delta_{ij} - \partial_{z^2} \ell_{ij} \Psi_{ij}$
and 
$D_{\pi^2}  \Delta_{ij}  := \partial_{\pi^2}  \Delta_{ij} -\partial_{z^3} \ell_{ij} \Psi_{ij}$.

\medskip

We are now ready to present the asymptotic distribution of  $\widehat \delta$ defined in \eqref{eq:apehat}.

\begin{theorem}[Asymptotic distribution of $\widehat \delta$]
  \label{th:DeltaLimit}
   Suppose that the assumptions of Theorem~\ref{th:BothEffects}
   and Assumption~\ref{ass:PanelA2} hold, and that the following limits exist:
\begin{align*}
   \overline {(D_{\beta} \Delta)}_{\infty} &= \EE \left[
        \frac {1} {n}  \sum_{(i,j) \in {\cal D}}
         \E(  \partial_{\beta} \Delta_{ij} -  \Xi_{ij} \partial_{\pi} \Delta_{ij}  ) \right]',
       \nonumber \\
   \overline B^\delta_{\infty} &=
      -  \EE \left\{  \frac {1} {I}  \sum_{(i,j) \in {\cal D}}
            \gamma^{0 \, \prime}_j  
            \left[ \sum_{h \in {\cal D}_i}   \gamma^{0}_h \gamma^{0 \prime}_h \, \E\left(  \partial_{z^2} \ell_{ih} \right)
            \right]^{-1}
             \gamma^{0}_j \;
        \E\left[
                \partial_{z} \ell_{ij}  
                \, D_{\pi}  \Delta_{ij}  
                + \frac 1 2  D_{\pi^2}  \Delta_{ij} 
                 \right]  
    \right\}  ,
        \nonumber \\
      \overline D^\delta_{\infty} &=   -  \EE \left\{  \frac {1} {J}  \sum_{(i,j) \in {\cal D}}
            \alpha^{0\, \prime}_i  
            \left[ \sum_{h \in {\cal D}_j}   \alpha^{0}_h \alpha^{0 \prime}_h \, \E\left(  \partial_{z^2} \ell_{hj} \right)
            \right]^{-1}
             \alpha^{0}_i \;
      \E\left[
                \partial_{z} \ell_{ij}  
                \, D_{\pi}  \Delta_{ij} 
                + \frac 1 2  D_{\pi^2}  \Delta_{ij} 
                 \right]  
    \right\} , \\
        \overline{V}_{\infty}^{\delta}  &=  - \EE  \left\{ \frac {1} {n}  \sum_{(i,j) \in {\cal D}}  
        \E \left( \Gamma_{ij} \Gamma_{ij}'  +  \Gamma_{ji}  \Gamma_{ij}' \right) \right\},
   \end{align*}
where $\Gamma_{ij} =   \overline {(D_{\beta} \Delta)}_{\infty} \overline W_\infty^{-1} \partial_{z} \ell_{ij} \tilde X_{ij}
      -    \Psi_{ij}  
          \partial_{z} \ell_{ij}$.
    Then,
    \begin{equation*}
    \sqrt{n} \left[\widehat \delta - \delta^0 
      - \frac{I} n \,  \overline {(D_{\beta} \Delta)}_{\infty}  \overline{W}_{\infty}^{-1} \overline B_{\infty} 
       - \frac{J} n \, \overline {(D_{\beta} \Delta)}_{\infty}   \overline{W}_{\infty}^{-1} \overline D_{\infty} 
      - \frac{I} n \,   \overline B^\delta_{\infty} 
       - \frac{J} n \,  \overline D^\delta_{\infty} 
       \right] \to_d \mathcal{N}(0 ,
           \;\overline V_{\infty}^{\delta}) .
    \end{equation*}
\end{theorem}

\medskip

\begin{remark}[Panel Data] In case (a) of Assumption \ref{ass:PanelA1}$(i)$, the term involving the cross products  $\Gamma_{ji}  \Gamma_{ij}'$ drops out from the asymptotic variance $\overline V_{\infty}^{\delta}$.
\end{remark}

\medskip

Theorem  \ref{th:DeltaLimit} shows that $\widehat \delta$ is consistent and normally distributed, but can have bias of the same order as its standard deviation.  The first two terms of the bias come from the bias of $\widehat \beta$. They drop out when  either $\widehat \beta$ does not have bias or the APE is estimated from a bias corrected estimator of $\beta$. We verify the presence of  bias in two of the running examples. 

\medskip

\noindent  \textbf{Example \ref{example: linear}} (Linear model). \textit{In this case $\overline B_{\infty}  = \overline D_{\infty}  = 0$ and
$$
\Delta_{ij}(\beta,\pi) = (Y_{ij} - X_{ij}'\beta -  \pi)^2, 
$$
so that $\partial_z \Delta_{ij} = -2(Y_{ij} - X_{ij}'\beta^0  - \pi_{ij}^0)$ and $\partial_{z^2} \Delta_{ij} = 2$. Substituting these values in the expressions of the bias of Theorem \ref{th:DeltaLimit}  yields  
$$
 \overline B^\delta_{\infty} =  \overline D^\delta_{\infty} =  - R \sigma^2,
$$
where we use \eqref{eq:ex3simplification}. This result formalizes the analysis in Section \ref{sec:numeric}
}

\medskip

\noindent  \textbf{Example  \ref{example: probij} } (Binary response model). \textit{Let $\Delta_{ij}(\beta,\pi)$ be as defined in either (\ref{example: probit: meff1}) or (\ref{example: probit: meff2}).
Using the notation previously introduced for this example, the expressions of $\overline B_{\infty}^{\delta}$ and $\overline D_{\infty}^{\delta}$ in Theorem \ref{th:DeltaLimit} yield
\begin{align*}
   \overline B_{\infty}^{\delta} &=
        \EE \left\{  \frac {1} {2\, I}  \sum_{(i,j) \in {\cal D}}
            \gamma^{0 \, \prime}_j  
            \left[  \sum_{h \in {\cal D}_i}   \gamma^{0}_h \gamma^{0 \prime}_h \, \E\left(  \partial_{z^2} \ell_{ih} \right)
            \right]^{-1}
             \gamma^{0}_j \;
        \E\left(\partial_{\pi^2} \Delta_{ij} - \Psi_{ij} H_{ij} \partial^2 F_{ij}                 \right)  
    \right\},
        \nonumber \\
      \overline D_{\infty}^{\delta} &=     \EE \left\{  \frac {1} {2\, J}  \sum_{(i,j) \in {\cal D}}
            \alpha^{0\, \prime}_i  
            \left[  \sum_{h \in {\cal D}_j}   \alpha^{0}_h \alpha^{0 \prime}_h \, \E\left(  \partial_{z^2} \ell_{hj} \right)
            \right]^{-1}
             \alpha^{0}_i \;
        \E\left(\partial_{\pi^2} \Delta_{ij} - \Psi_{ij} H_{ij} \partial^2 F_{ij} 
                 \right)  
    \right\}.
   \end{align*}   
As for the model parameter, these bias terms grow linearly with the number of factors $R$.}

\medskip

\noindent  \textbf{Example \ref{example: poisson}} (Count response model). \textit{Let $\Delta_{ij}(\beta,\pi)$ be as defined in either \eqref{example: poisson: meff1} or \eqref{example: poisson: meff}. In this case $\overline B_{\infty}  = \overline D_{\infty}  = 0$, and $\partial_z \Delta_{ij} = \partial_{z^2} \Delta_{ij} = \Delta_{ij}$. Substituting these values in the expressions of the bias of Theorem \ref{th:DeltaLimit}  yields  
$$
 \overline B^\delta_{\infty} =  \overline D^\delta_{\infty} =  0,
$$
which generalizes the result in   \cite{FW16} of no asymptotic bias for the estimators of the APEs in the Poisson model with strictly exogenous covariates and  additive individual and time effects to the Poisson model with strictly exogenous covariates and  factor structure.
}


%
%

\subsection{Bias correction and Inference}
Theorems  \ref{th:BothEffects}  and  \ref{th:DeltaLimit} establish that the estimators of the model parameter and APEs have a bias of the same order as their standard deviations in some models.  In this section, we  describe how to apply recent developments in nonlinear panel data  to correct the bias from the estimators. To simplify the notation we assume that there is no missing data.\footnote{ We refer to \cite{ARE} for a discussion on how to modify the corrections to deal with missing data.} We consider a generic estimator $\widehat \theta$ of the parameter $\theta$, which may correspond to the model parameter or an APE. In this notation, Theorems  \ref{th:BothEffects}  and  \ref{th:DeltaLimit}  show that $\widehat \theta$ can have a bias $\mB_{\infty} = \EE[ \mB(\beta^0,\phi_n^0)]$ with structure
$$
\mB(\beta,\phi_n) = \frac{B(\beta,\phi_n)}{J} + \frac{D(\beta,\phi_n)}{I}.
$$
The intuition behind this structure is that there are $J$ observations that are informative to estimate each $\alpha_i$ and $I$ observations that are informative to estimate each $\gamma_j$. 

\medskip

An analytical correction based on \cite{Hahn:2004p882} and \cite{FW16} can be formed as
$$
\widetilde \theta_{\rm ABC} = \widehat \theta - \widehat  \mB, \ \  \ \  \widehat  \mB =   \mB(\widehat \beta, \widehat \phi_n).
$$
A split-sample correction based on \cite{DJ2015} and \cite{FW16} can be formed as
$$
\widetilde \theta_{\rm SBC} = 3 \widehat \theta  - \bar \theta_{I,J/2}  - \bar \theta_{I/2,J}, 
$$
where $\bar{\theta}_{I,J/2}$ is the average of the estimators in the haft-panels $ \{(i,j) :  i = 1,\dots,I; j = 1, \dots, \lceil J/2 \rceil \}$ and $ \{(i,j) :  i = 1,\dots,I; j = \lfloor J/2 + 1\rfloor, \dots, J \}$, and $\bar{\theta}_{I/2,J}$ is the average of the estimators in the haft-panels $ \{(i,j) :  i = 1,\dots,\lceil I/2 \rceil; j = 1, \dots, J \}$ and $ \{(i,j) :  i = \lfloor I/2 + 1\rfloor,\dots,I; j = 1, \dots, J \}$, where $\lceil \cdot \rceil$ and $\lfloor \cdot \rfloor$ are the ceil and floor functions.  For network data where $I=J$ and the two dimensions of the data index the same entities,  \cite{Stata2017} proposed the leave-one-out correction
$$
\widetilde \theta_{\rm NBC} =  I \widehat \theta  - (I-1) \bar \theta_{I-1}, \ \ \bar \theta_{I-1} = I^{-1} \sum_{i=1}^I \widehat \theta_{-i}, 
$$
where $\widehat \theta_{-i}$ is the estimator in the subpanel $ \{(k,j) :  k = 1,\dots,I; j = 1, \dots, I, k \neq i, j \neq i  \}$, that is, the original panel leaving out the observations corresponding to the entity $i$ as either sender or receiver.

The discussion of bias correction so far is applicable very generally to network and panel models with two-way fixed effects.
We now specialize it to our nonlinear models with interactive fixed effects. 
%
%
%
%
%
For the analytic bias correction and for variance estimation we require consistent estimators for the quantities 
$\overline B_{\infty} $, $ \overline D_{\infty}$, $\overline{W}_{\infty}$, and $ \overline \Sigma_{\infty}$
defined in Theorem~\ref{th:BothEffects}. Let $\widehat B$, $ \widehat D$, $\widehat W$ and $\widehat \Sigma$
be the corresponding sample analogs, obtained by simply dropping expectations and plugging in the fixed
effect estimators for the true value of the parameters. For example, 
\begin{align*}
  \widehat W   &=  -  \frac {1} {n}  \sum_{(i,j) \in {\cal D}} 
              \partial_{z^2}  \widehat \ell_{ij}  \left( X_{ij} - \widehat \Xi_{ij}  \right) \left( X_{ij} - \widehat \Xi_{ij}  \right)' ,
\end{align*}
where $  \partial_{z^2}  \widehat \ell_{ij} =   \partial_{z^2}    \ell_{ij} \left( X_{ij}' \widehat \beta +  \widehat \alpha_i' \, \widehat \gamma_j \right)$,
and $\widehat \Xi_{ij} $ is the $d_x$-vector with 
elements   $ \widehat \Xi_{it,k}  =  \alpha^{\# \, \prime}_{i,k}  \widehat \gamma_j + \widehat \alpha^{ \prime}_i \gamma^{\#}_{t,k} $,
with $ \alpha^{\# \, \prime}_{i,k} $ and $\gamma^{\#}_{t,k} $ obtained as the solution to
\begin{align*}
  \left(\alpha^{\#}_{k} , \, \gamma^{\#}_{k} \right)
   &\in \argmin_{\alpha_{i,k},\gamma_{t,k}} \sum_{i,j}
  ( - \partial_{z^2} \widehat \ell_{ij} )
        \left(  \frac{ \partial_{z^2} \widehat \ell_{ij} X_{ij,k}}
                        {\partial_{z^2} \widehat \ell_{ij} }   
         - \alpha'_{i,k}  \widehat \gamma_j - \widehat \alpha^{ \prime}_i \gamma_{t,k} \right)^2.
\end{align*}  
Once those sample analogs are constructed, then the analytic bias correction of $\widehat \beta$
reads
$$\widetilde \beta_{\rm ABC} 
= \widehat \beta  
      - \frac{I} n \, \widehat W^{-1} \widehat B 
       - \frac{J} n \,  \widehat W^{-1} \widehat D.
$$
Analogously, we can construct sample analogs 
for $ \overline B^\delta_{\infty} $, $\overline D^\delta_{\infty}  $, $ \overline {(D_{\beta} \Delta)}_{\infty}$,
defined in Theorem~\ref{th:DeltaLimit}, in order to construct $  \widetilde \delta_{\rm ABC}$.
Also, let $  \widehat V^{\delta}$ be the sample analog of $\overline V_{\infty}^{\delta}$.

\begin{theorem} [Asymptotic Distribution of $\widetilde \beta_{\rm ABC} $ and $\widetilde \delta_{\rm ABC}$] \label{th:bc}
Under the conditions of Theorem \ref{th:BothEffects},
$$
      \sqrt{n} \left( \widetilde \beta_{\rm ABC} - \beta^0 
         \right)
          \;  \to_d \;
      {\cal N}( 0 ,
           \;\overline W_{\infty}^{-1}  \overline \Sigma_{\infty} \overline W_{\infty}^{-1}),
$$
$  \widehat W \,  \rightarrow_P \, \overline W_{\infty}$
and
$    \widehat \Sigma  \,  \rightarrow_P \,   \Sigma_{\infty}$.
If, in addition, the conditions of Theorem~\ref{th:DeltaLimit} hold, then
$$
      \sqrt{n} \left( \widetilde \delta_{\rm ABC} - \delta^0 
         \right)
          \; \to_d \;
      {\cal N}( 0 ,
           \;\overline V_{\infty}^{\delta}),
$$
and $
     \widehat V^{\delta} \, \rightarrow_P \,  \overline V_{\infty}^{\delta} .
$ 
 
\end{theorem}
Theorem~\ref{th:bc} shows that analytic bias correction can be used to obtain
estimators of $\beta^0$ and $\delta^0$ that are asymptotically unbiased. It also shows that the simple plug-in estimators
of the asymptotic variances are consistent, thus allowing to perform asymptotically valid hypothesis tests
and to construct asymptotically valid confidence intervals for $\beta^0$ and $\delta^0$.

Showing that the Jackknife corrected estimators $\widetilde \beta_{\rm JBC}$ and $\widetilde \delta_{\rm JBC}$ have the same asymptotic distribution as 
$\widetilde \beta_{\rm ABC}$ and $\widetilde \delta_{\rm ABC}$  requires
an additional homogeneity assumption, which guarantees that the
unconditional distribution of the data is stationary across $i$ and $j$. This assumption ensures that the terms $B$ and $D$ in the bias expansion of
$ \widehat \theta$ are the same as in the bias expansions of the half-panel estimates $\bar \theta_{I,J/2}$  and $\bar \theta_{I/2,J}$,
so that forming the Jackknife linear combination $\widetilde \theta_{\rm SBC}$  indeed cancels those bias terms.
In other words, the data distribution should not systematically differ across the subsamples used for the Jackknife correction  \citep{DJ2015, FW16}.
 
The derivation of the asymptotic distribution of the leave-one-out correction $\widetilde \theta_{\rm NBC}$  furthermore requires a third-order bias expansion (i.e., up to terms of order $1/I^2$), because in the expression of $\widetilde \theta_{\rm NBC}$  the  estimators $\widehat \theta$ and $\bar \theta_{I-1}$
are multiplied by the factors $I$ and $(I-1)$ that grow with the sample size.
We have not worked out those higher-order expansion here,
but we refer to \cite{SunDhaene2017} for an example of higher-order expansions in nonlinear panel models.

\section{Implementation Details}\label{sec:implementation}

\subsection{Computation of the Estimator}

We apply the following EM-type algorithm based on \cite{C14} to find the solution to the program \eqref{LobjMAX}: 

\begin{algorithm}[Likelihood Maximization]\label{alg:estimation} (i) Initialization: provide the initial values $\beta^{(0)}$, $\alpha^{(0)}$ and $\gamma^{(0)}$ for $\beta$, $\alpha$ and $\gamma$ (e.g., set all these initial values equal to zero). (ii) Iteration $k \geq 1$: given  $\beta^{(k-1)}$, $\alpha^{(k-1)}$ and $\gamma^{(k-1)}$, (a) compute the $I \times J$ matrix $\mu^{(k)}$ with elements
$$
\mu_{ij}^{(k)} = z_{ij}^{(k)} - \frac{\partial_z \ell_{ij}(z_{ij}^{(k)})}{\partial_{z^2} \ell_{ij}(z_{ij}^{(k)})}, \ \ z_{ij}^{(k)} = X_{ij}'\beta^{(k-1)} + \alpha_i^{(k-1)'}\gamma_{j}^{(k-1)} ;
$$ 
(b) update $\alpha$ and $\gamma$: solve the principal components program
$$
(\alpha^{(k)}, \gamma^{(k)}) \in \argmin_{{\rm vec}(a) \in \mathbb{R}^{I\times R}, {\rm vec}(g) \in  \mathbb{R}^{J\times R}} {\rm Tr}(\mu^{(k)} - a'g)(\mu^{(k)} - a'g)';
$$
and (c) update $\beta$:
$$
\beta^{(k)} = \left[\tilde X^{(k)'} \tilde X^{(k)} \right]^{-1} \tilde X^{(k)'} {\rm vec}(\tilde \mu^{(k)}), 
$$
where $\tilde \mu^{(k)} = \mathcal{M}_{\alpha^{(k)}} \mu^{(k)}  \mathcal{M}_{\gamma^{(k)}}$,  $\tilde X^{(k)}$ is an $IJ \times d_x$ matrix with typical column $\tilde X_c^{(k)} = {\rm vec}( \mathcal{M}_{\alpha^{(k)}} \mathbb{X}_c  \mathcal{M}_{\gamma^{(k)}})$, ${\cal M}_{\alpha^{(k)}} := \mathbb{I} - \alpha^{(k)}(\alpha^{(k)'}\alpha^{(k)})^{\dag}\alpha^{(k)'}$, ${\cal M}_{\gamma^{(k)}} := \mathbb{I} - \gamma^{(k)}(\gamma^{(k)'}\gamma^{(k)})^{\dag}\gamma^{(k)'}$ and $\mathbb{X}_c$ is an $I \times J$ matrix with elements $X_{ij,c}$.
(iii) Convergence: repeat step (ii) until $\|\beta^{(k)} - \beta^{(k-1)} \|_{\infty}  \leq \epsilon$, where $\epsilon$ is a tolerance parameter (e.g., $\epsilon = 10^{-5}$).
\end{algorithm}

\cite{C14} analyzed the convergence guarantees for this algorithm. She showed that the algorithm  converges to a local maximum of the log-likelihood. Since the log-likelihood can have multiple local maxima, we recommend to run the algorithm for several initial values and choose the solution that yields the highest value of the log-likelihood.

\begin{remark}[Additive Effects] Separate additive effects in both dimensions can be treated as one known factor of ones with unknown loading and one known loading of ones with unknown factor. They can therefore be included by imposing the constraints that the second column of $\alpha^{(k)}$ and the first column of $\gamma^{(k)}$ are equal to vectors of ones in  part (b) of step (ii). Other known factors with unknown loadings or known loadings with unknown factors can be incorporated similarly by imposing constraints in part (b) of step (ii).
\end{remark}

\subsection{Estimating the Number of Factors}

The problem of estimating the number of factors $R$ has been extensively discussed for linear factor models without covariates,
see for example, \cite{BaiNg2002,hallin2007generalized,Onatski2010,alessi2010improved,AhnHorenstein2013}. 
These methods can be extended to linear models with covariates, provided that an appropriate preliminary 
estimator $\widetilde \beta$ of the regression parameters $\beta$ is available that does not require knowing $R$. In this case the existing methods are applied to the residuals $Y_{ij} - X_{ij}'  \widetilde \beta$. If there exists an upper bound for the number of factors, $R_{\max} \geq R$, then the preliminary estimator $\widetilde \beta$ is given by the least squares estimator with
$R_{\max}$ factors, see \cite{MoonWeidner2015}. These methods can also be extended to the nonlinear factor models that we consider. For example, the various information criteria in \cite{BaiNg2002} are all based on minimizing the sum of squared residuals plus a penalty function, and can be adapted to the likelihood problem in the spirit of classic model selection criteria (AIC, BIC, etc), see \cite{AB16} for an example of this approach.\footnote{\cite{kss12} proposed an alternative estimator of the number of factors in linear models specially adapted to i.i.d. errors.} It is less obvious, however, how to extend the eigenvalue ratio (ER) test of \cite{AhnHorenstein2013} to nonlinear models.  This method is  attractive because it does not depend on somewhat arbitrary functional form assumptions or tuning parameters. It only requires to specify $R_{\max}$, but there is no penalty function or any other tuning parameter. Assuming that there exists an upper bound $R_{\max} > R$, we propose adapting this method to nonlinear factor single-index models using the following algorithm:

\begin{algorithm}[Estimation of $R$]\label{alg:nfactors}
(1) Obtain preliminary estimates $\widetilde \beta$, $\widetilde \alpha$ and $\widetilde \gamma$ using Algorithm \ref{alg:estimation} with $R = R_{\max}$.  
(2) Compute preliminary estimates of the factor structure as the $I \times J$
matrix $\widetilde \pi$ with elements $\widetilde \pi_{ij} := \widetilde \alpha_i' \widetilde \gamma_j$. By construction, ${\rm rank}( \widetilde \pi) \leq R_{\max}$. (3)  Apply the eigenvalue ratio criterion of \cite{AhnHorenstein2013} to $\widetilde \pi$ in order to estimate $R$, that is,
     \begin{align*}
           \widehat R &= \max_{r \in \{1,\ldots, R_{\max} - 1\}}  {\rm EV}(r) ,
           &     
           {\rm EV}(r) &= \frac{ \lambda_r\left( \widetilde \pi \widetilde \pi' \right) } { \lambda_{r+1}\left( \widetilde \pi \widetilde \pi' \right) } ,
     \end{align*}
     where $\lambda_r\left( \widetilde \pi \widetilde \pi' \right)$ denotes the $r$'th largest eigenvalue of $\widetilde \pi \widetilde \pi' $.
\end{algorithm}

\begin{remark}[Additive Effects] When the specification includes factors with known loadings and/or loadings with known factors, $\widetilde \pi_{ij}$ is the estimator of the part of the factor structure that does not contain known factors and known loadings and $R_{\max}$ refers to the number of factors in this part. 
\end{remark}

This algorithm can be seen as a natural generalization of the \cite{AhnHorenstein2013} to single-index models. Indeed, if we applied it to the linear model $Y_{ij} = X_{ij}' \beta +  \alpha_i' \, \gamma_j + \varepsilon_{ij}$, 
with $\log f(Y_{ij}  \mid X_{ij}' \beta + \alpha_i' \gamma_j )$ replaced by $-(Y_{ij} - X_{ij}' \beta -  \alpha_i' \, \gamma_j)^2$, 
then $$\lambda_r\left( \widetilde \pi \widetilde \pi' \right) = \lambda_r\left[ \left(Y_{ij}  - X_{ij}' \widetilde \beta   \right) \left(Y_{ij}  - X_{ij}' \widetilde \beta   \right) ' \right],$$ which corresponds to the eigenvalue ratio criterion of \cite{AhnHorenstein2013}
applied to the residuals $Y_{ij}  - X_{ij}' \widetilde \beta $. Based on this coverage of the linear model, we conjecture that $\widehat R$ is a consistent estimator of $R$ under suitable conditions. To formalize this argument, a key step  is to establish the consistency of the preliminary estimator $\widetilde \beta $, extending the results of  \cite{MoonWeidner2015} from linear to nonlinear models, and the properties of the estimator of the factor structure $\widetilde \pi$. The main technical challenge is to characterize  $\widetilde \pi$, which is not even available for the linear model with covariates and $R>R_0$.
 We leave this analysis to future research. In the rest of the section we show that the method performs well in numerical simulations.


 To show how $\widehat R$ performs in small samples, we generate samples from the Poisson model of Example \ref{example: poisson}  with additive effects where $z_{ij} = X_{ij} \beta + \alpha_{1i} + \gamma_{1j} + \alpha_{2i}' \gamma_{2j}$,   $X_{ij} \sim N(1,1/3)$, $\beta = 0$, $\alpha_{1i} \sim U(0,1)$, $\gamma_{1i}  \sim U(0,1)$,  $\alpha_{2i}$ is an $R_2$-dimensional standard normal vector with  independent components, $\gamma_{2i}$ is an $R_2$-dimensional standard normal vector with  independent components,  and $X_{ij}$, $\alpha_{1i'}$, $\gamma_{1j'}$, $\alpha_{2i''}$ and $\gamma_{2j''}$ are mutually independent for all $i,i',i'' = 1,\ldots,I$ and $j,j',j'' = 1,\ldots,J$. We generate $1,000$ datasets with $I=J\in\{50,75, 100,150\}$ and $R_2 \in \{1,2,3\}$, and apply Algorithm \ref{alg:nfactors} with $R_{\max} \in \{4, 5,6 \}$. Table \ref{table:nfactors} reports the average of $\widehat R_2$ across simulations and the proportion of simulations where $\widehat R_2 = R_2$.  Here, we find that $\widehat R_2$  has little bias and often yields the true $R_2$, specially for the larger sample sizes with $I \geq 75$. Interestingly, the performance of  $\widehat R_2$  improves as $R_{\max}$ gets closer to  $R_2$.  Given this sensitivity, we recommend computing  $\widehat R_2$ for several values of  $R_{\max}$.
  
\begin{table}\caption{Simulation Results for $\widehat R_2$ in Poisson Model}\label{table:nfactors}\vspace{.25cm}
\centering
\begin{tabular}{cccccccc}
\hline\hline 
$I=J$ & $R_{\max}$ & $\E[\widehat R_2]$ &  $\Pr(\widehat R_2 = R_2)$  & $\E[\widehat R_2]$ &  $\Pr(\widehat R_2 = R_2)$  & $\E[\widehat R_2]$ &  $\Pr(\widehat R_2 = R_2)$  \tabularnewline
\hline 
& & \multicolumn{2}{c}{ $R_2=1$} &  \multicolumn{2}{c}{ $R_2=2$} & \multicolumn{2}{c}{ $R_2=3$} \tabularnewline
50 &  4 & 1.05 & 0.96   & 1.94 & 0.84   & 2.80 & 0.88 \tabularnewline 
& 5 & 1.16 & 0.88  & 1.92 & 0.71  & 2.84 & 0.67 \tabularnewline
& 6 & 1.34 & 0.75  & 1.90 & 0.57  & 2.83 & 0.50 \medskip \tabularnewline
75 &  4 & 1.01 & 0.99   & 1.99 & 0.96   & 2.78 & 0.83 \tabularnewline 
& 5 & 1.01 & 0.99  & 1.97 & 0.91  & 2.99 & 0.83 \tabularnewline
& 6 & 1.03 & 0.97  & 1.93 & 0.83  & 3.11 & 0.72 \medskip \tabularnewline
100 &  4 & 1.06 & 0.96   & 2.01 & 0.98  & 2.98 & 0.99 \tabularnewline 
& 5 & 1.13 & 0.92  & 2.06 & 0.95  & 3.00 & 0.99 \tabularnewline
& 6 & 1.28 & 0.87  & 2.11 & 0.92  & 3.01 & 0.98 \medskip \tabularnewline
150 &  4 & 1.01 & 0.99   & 2.01 & 0.97   & 2.99 & 0.99 \tabularnewline  
& 5 & 1.04 & 0.98  & 2.09 & 0.90  & 2.98 & 0.96 \tabularnewline
& 6 & 1.09 & 0.96  & 2.15 & 0.91  & 2.99 & 0.94 \tabularnewline
\hline\hline 
\multicolumn{8}{l}{\footnotesize{Notes: 1,000 simulations. The design includes one covariate and additive effects. }}
\end{tabular}
\end{table}

\section{Application to Gravity Equation}\label{sec:empirics}

\subsection{Gravity Equation with Multiple Latent Factors}

The gravity equation is a fundamental empirical relationship in international economics. We estimate a gravity equation of trade between countries using data from \cite{Helpman01052008}  on bilateral trade flows and other trade-related variables for 157 countries in 1986.\footnote{The original data set includes 158 countries. We exclude Congo because it did not export to any other country in 1986.} The data set contains a network of trade data where both $i$ and $j$ index countries as senders (exporters) and receivers (importers), such that  $I = J = 157$. The outcome $Y_{ij}$ is the volume of trade in thousands of constant 2000 US dollars from country $i$ to country $j$, and the covariates $X_{ij}$ include determinants of bilateral trade flows such as the logarithm of the distance in kilometers between country  $i$'s capital and country $j$'s capital and  indicators for common colonial ties, currency union, regional free trade area (FTA),  border,  legal system,  language, and religion.  Table \ref{table:ds} reports descriptive statistics of the variables used in the analysis. There are $157 \times 156 = 24,492$ observations corresponding to different pairs of countries. The observations with $i = j$ are missing because we do not observe trade flows from a country to itself.  The trade variable in the first row is an indicator of positive volume of trade. There are no  trade flows for  55\% of the country pairs.  


\begin{table}[ht]\caption{Descriptive Statistics}\label{table:ds}
\centering
\begin{tabular}{lcc}
  \hline\hline
 & \text{Mean} & \text{Std. Dev.} \\ 
  \hline
  Trade & 0.45 & 0.50 \\ 
  Trade Volume & 84,542 & 1,082,219 \\ 
  Log Distance & 4.18 & 0.78 \\ 
  Legal & 0.37 & 0.48 \\ 
  Language & 0.29 & 0.45 \\ 
  Religion & 0.17 & 0.25 \\ 
  Border & 0.02 & 0.13 \\ 
  Currency & 0.01 & 0.09 \\ 
  FTA & 0.01 & 0.08 \\ 
  Colony & 0.01 & 0.10 \\ \hline
  \multicolumn{1}{l}{Country Pairs} &  \multicolumn{2}{c}{24,492} \\ 
   \hline\hline
   \multicolumn{3}{l}{\footnotesize{Source: \cite{Helpman01052008}}}
\end{tabular}
\end{table}

\medskip

We estimate a Poisson model with the following specification of the intensity
$$
\E[Y_{ij} \mid X_{ij}, \alpha_{1i}, \gamma_{1j}, \alpha_{2i}, \gamma_{2j} ] = \exp(X_{ij}'\beta + \alpha_{1i} + \gamma_{1j} + \alpha_{2i}'\gamma_{2j}), 
$$
where $\alpha_{2i}$ and $\gamma_{2i}$ are $R_2$-dimensional vectors of factors and factor loadings. This model is a special case of Example \ref{example: poisson} with $\alpha_i = (\alpha_{1i}, 1, \alpha_{2i}')'$,  $\gamma_j = (1,\gamma_{1j}, \gamma_{2j}')'$, and $R = 2 + R_2$.  We explicitly include additive importer and exporter effects to account for scale and multilateral resistance effects following \cite{EatonKortum2001} and \cite{AndersonWincoop2003}. Moreover, we also include interactive country effects to capture possible clustering and homophily induced by latent factors such as country trade partnerships, presence of multinationals or immigrant communities, or differences in natural resources or industrial composition.

 \medskip

\begin{table}
\caption{Parameters of Gravity Equation}\label{tab:Coefficients-of-Static}
\medskip
\begin{centering}
\begin{tabular}{lccccccc}
\hline\hline
 & $R_2=0$ & $R_2=1$ & $R_2=2$ & $R_2=3^*$ & $R_2=4$ & $R_2=5$ & $R_2=6$\tabularnewline
\hline
Log Distance & -0.64 & -0.63 & -0.71 & -0.69 & -0.77 & -0.90 & -1.01\tabularnewline
 & (0.05) & (0.05) & (0.05) & (0.06) & (0.07) & (0.09) & (0.21)\tabularnewline
 & [0.07] & [0.05] & [0.06] & [0.06] & [0.08] & [0.09] & [0.22]\tabularnewline
Border & 0.71 & 0.41 & 0.32 & 0.36 & 0.38 & 0.36 & 0.27\tabularnewline
 & (0.12) & (0.06) & (0.05) & (0.05) & (0.06) & (0.12) & (0.11)\tabularnewline
 & [0.16] & [0.07] & [0.06] & [0.06] & [0.06] & [0.12] & [0.11]\tabularnewline
Legal & 0.30 & 0.14 & 0.26 & 0.22 & 0.13 & 0.16 & 0.27\tabularnewline
 & (0.04) & (0.04) & (0.04) & (0.04) & (0.04) & (0.06) & (0.11)\tabularnewline
 & [0.06] & [0.04] & [0.04] & [0.04] & [0.04] & [0.06] & [0.11]\tabularnewline
Language & -0.17 & -0.19 & -0.02 & 0.03 & -0.09 & -0.03 & 0.09\tabularnewline
 & (0.07) & (0.07) & (0.06) & (0.06) & (0.07) & (0.11) & (0.22)\tabularnewline
 & [0.10] & [0.07] & [0.06] & [0.06] & [0.08] & [0.12] & [0.21]\tabularnewline
Colony & 0.36 & 0.58 & 0.39 & 0.45 & 0.63 & 0.61 & 0.55\tabularnewline
 & (0.08) & (0.11) & (0.09) & (0.09) & (0.12) & (0.28) & (0.46)\tabularnewline
 & [0.12] & [0.14] & [0.12] & [0.12] & [0.14] & [0.28] & [0.46]\tabularnewline
Currency & 0.60 & 0.29 & 1.37 & 1.38 & 0.65 & 0.63 & 0.77\tabularnewline
 & (0.27) & (0.31) & (0.39) & (0.33) & (1.08) & (1.93) & (2.05)\tabularnewline
 & [0.30] & [0.38] & [0.41] & [0.36] & [1.16] & [1.92] & [2.13]\tabularnewline
FTA & 0.25 & 0.15 & 0.17 & 0.13 & 0.25 & 0.31 & 0.26\tabularnewline
 & (0.07) & (0.06) & (0.06) & (0.06) & (0.09) & (0.14) & (0.25)\tabularnewline
 & [0.09] & [0.07] & [0.07] & [0.07] & [0.09] & [0.14] & [0.26]\tabularnewline
Religion & -0.25 & 0.18 & 0.24 & 0.34 & 0.44 & 0.30 & 0.35\tabularnewline
 & (0.12) & (0.11) & (0.14) & (0.13) & (0.13) & (0.27) & (0.34)\tabularnewline
 & [0.13] & [0.11] & [0.13] & [0.13] & [0.13] & [0.26] & [0.34]\tabularnewline
 \hline
Log-likelihood & -0.44 & 0.31 & 0.67 & 0.84 & 0.96 & 1.04 & 1.11 \\ 
\hline\hline
\multicolumn{8}{l}{\footnotesize{Notes: all the columns include importer and exporter additive effects. }}\\
\multicolumn{8}{l}{\footnotesize{Standard errors in parenthesis.  Standard errors robust to reciprocity in brackets.}}\\
\multicolumn{8}{l}{\footnotesize{$^*$ Number of factors selected with $R_{\max} =5$. Log-likelihood is multiplied by 100.}}
\end{tabular}
\par\end{centering}
\end{table}

Table \ref{tab:Coefficients-of-Static} reports the estimates and standard errors of the parameter $\beta$.\footnote{We do not report estimates of APEs because in the specification of the Poisson model that we use  the parameters can be interpreted as elasticities. } 
We consider specifications with different number of interactive effects, $R_2$, in addition to the additive effects .  The last row of the table reports the maximum value of the average log-likelihood, 
$L(\widehat \beta, \widehat \phi_n)/n $. We report two sets of standard errors corresponding to the dependence structures of cases (a) and (b) of Assumption \ref{ass:PanelA1}(i). The standard errors in brackets account for possible reciprocity in the data.  
In this case, the method of Section \ref{sec:implementation}  selects $R_2 = 3$ factors when $R_{\max} = 4$ and $R_{\max} = 5$. We take  $R_2 = 3$ as our preferred specification, but we also note that, relative to the standard errors, the estimates are not very sensitive to the $R_2$ in the range of values  that we consider. One possible concern with the use of the Poisson model in the trade application is the excess zeros, i.e. the high probability of zero trade.\footnote{We thank an anonymous referee for raising this issue.} In this case, however, it does not seem to be a problem because the estimated model with $R_2=3$ predicts a probability of zero trade of $0.61$, which is higher than the observed probability of $0.55$.   

\medskip

We find that the sign of most of the effects is robust to the inclusion of latent factors.  The only exceptions are the effects of common religion and language, which in the specification with only additive effects have counterintuitive negative signs  that turn positive in our preferred specification. 
Comparing across columns, we observe that the model without factors seems to exaggerate the role of common border, whereas it downplays the effect of distance and colonial links. For example, increasing by 10\% the distance  reduces by 6.9\%  the volume of trade and  sharing border  increases it by 36\% according to our preferred specification with  $R_2 = 3$, whereas the same effects are  6\% and 71\% according to the specification with $R_2=0$. Except for language, all the coefficients are individually significant at the 5\% level.  Overall, increasing the number of factors makes the estimates less precise due to the loss of degrees of  freedom. This observation showcases a trade-off in estimation between efficiency and robustness to richer dependence structures in the unobservables. Finally, accounting for reciprocity slightly increases the standard errors, but does not change the statistical significance of the estimates. 


%

\subsection{Calibrated Monte Carlo Simulation}

We evaluate the finite-sample properties of our estimation and inference methods in a Monte Carlo simulation that mimics the trade application. The design is calibrated to the Poisson model with additive importer and exporter country effects and one factor. We analyze the performance of the estimator of $\beta$  in terms of bias, dispersion and inference accuracy.  To speed up computation,
we include only one covariate: the log distance. More specifically, we generate $Y_{ij}$  from a Poisson distribution with intensity $\exp(X_{ij} \widehat{\beta} +\widehat{\alpha}_{1i}+\widehat{\gamma}_{1j}+\widehat{\alpha}_{2i}\widehat{\gamma}_{2j})$ independently across $i$ and $j$, where $X_{ij}$ takes the values of log-distance in the trade data set, and $\widehat{\beta}$ and $\{\widehat{\alpha}_{1i}, \widehat{\alpha}_{2i},\widehat{\gamma}_{1i}, \widehat{\gamma}_{2i}\}_{j=1}^{157}$,
are equal to the estimates of the parameter, importer effects,
exporter effects, factors and factor loadings. We repeat
this procedure in $1,000$ simulations for four different sample sizes:
$I=50$, $I=75$, $I=100$ and $I=157$ (full sample in the
application). For each sample size and simulation, we draw a random
sample of $I$ countries both as importers and exporters without replacement,
so that the number of observations is $I\times(I-1)$. For each simulated
sample, we reestimate the model parameter and standard errors, and construct 95\% confidence interval for the model parameter. 

\medskip

\begin{table}
\caption{Results of Calibrated Simulations}\label{tab:Calibrated-Monte-Carlo} 
\medskip
\begin{centering}
\begin{tabular}{crrrrrrrrrr}
\hline\hline 
\multicolumn{1}{c}{$I$} &  \multicolumn{1}{c}{Bias}  &  \multicolumn{1}{c}{SD} &  \multicolumn{1}{c}{RMSE}  & \multicolumn{1}{c}{SE/SD} & \multicolumn{1}{c}{p;95} &  \multicolumn{1}{c}{Bias}  &  \multicolumn{1}{c}{SD} &  \multicolumn{1}{c}{RMSE}  & \multicolumn{1}{c}{SE/SD} & \multicolumn{1}{c}{p;95}  \tabularnewline
\hline
& \multicolumn{5}{c}{$R_2=1$} & \multicolumn{5}{c}{$R_2=R_2^*$}\tabularnewline
50 & 6.08 & 14.90 & 16.08 & 1.13 & 96 &   6.67 & 16.99 & 18.24 & 1.06 & 95 \tabularnewline
75 & 4.93 & 8.04 & 9.42 & 1.15 & 95 &  6.62 & 8.79 & 11.00 & 1.12 & 93 \tabularnewline
100 & 1.38 & 6.09 & 6.24 & 1.14 & 97 &  3.88 & 6.45 & 7.52 & 1.12 & 94 \tabularnewline
157 & 0.59 & 3.51 & 3.56 & 1.15 & 97 &  1.82 & 3.88 & 4.27 & 1.07 & 95 \tabularnewline
 & \multicolumn{5}{c}{$R_2=2$} & \multicolumn{5}{c}{$R_2 = 3$} \tabularnewline
50 & 6.76 & 15.71 & 17.09 & 1.12 & 97 & 8.61 & 16.63 & 18.71 & 1.11 & 95 \tabularnewline
75 & 5.97 & 8.70 & 10.55 & 1.11 & 94 & 6.68 & 9.37 & 11.50& 1.07 & 91 \tabularnewline
 100 & 3.27 & 6.37 & 7.16 & 1.12 & 95 & 4.81 & 6.80 & 8.33 & 1.08 & 93 \tabularnewline
157  & 2.24 & 3.61 & 4.24 & 1.14 & 94 & 1.99 & 3.89 & 4.37 & 1.08 & 94 \tabularnewline
\hline\hline
\multicolumn{11}{l}{\footnotesize{Notes: 1,000 simulations calibrated to trade data with additive  effects and $1$ factor.}} \\
\multicolumn{11}{l}{\footnotesize{$R_2=R_2^*$ estimates the number of factors with $R_{\max} = 4$.}} 
\end{tabular}
\par\end{centering}
\end{table}

Table \ref{tab:Calibrated-Monte-Carlo} reports the bias (Bias), standard deviation (SD), and root mean squared error (RMSE) of the estimator of the parameter $\beta$, together with the ratio of average standard error to the simulation standard
deviation (SE/SD), and the empirical coverage in percentage of a confidence interval
with 95\% nominal value (p;95). We estimate models with four different numbers of factors in addition to the additive effects, $R_2 \in \{1,2,3, R_2^*\}$, where $R_2^*$ is the number of factors selected by the method of Section \ref{sec:implementation} with $R_{\max} = 4$, which can vary across simulations. The results for the bias, SD and RMSE are reported in percentage of the true parameter value. We find that  the  bias is smaller than the standard deviation for every sample size. When we use the true number of factors $R_2=1$, the confidence intervals cover the parameter in more  than 95\% of the simulations. The excess coverage is due to the  overestimation of the dispersion of the estimators by the standard errors. 
Selecting the number of factors does not introduce bias, but increases the dispersion of the estimator of the parameter. The additional variability yields slight undercoverage of the  confidence intervals for small sample sizes.  On the other hand, adding unnecessary factors to the specification increases the bias and dispersion of the estimator, but the confidence intervals continue having good coverage properties. This robustness to the inclusion of too many factors is  consistent with the  theoretical results of \cite{MoonWeidner2015} for linear factor models. Overall, the simulations show that the asymptotic theory of Section \ref{sec:app_panel} provides a good approximation to the finite-sample behavior of the estimator.

\bibliographystyle{apa}
\bibliography{refs}

\newpage
\appendix

\section{Proofs}

\subsection{Notation and Normalization}

Remember the log-likelihood defined in the main text, and also define the rescaled version,
\begin{align*} 
     L(\beta,\phi) 
     & :=
              \sum_{(i,j) \in {\mathcal{D}}}  \log f(Y_{ij}  \mid X_{ij}' \beta + \pi_{ij}  ) ,
 &  {\cal L}^*(\beta,\phi)     &  :=   n^{-1/2} \,  L(\beta,\phi).
\end{align*}
For the true value of the fixed effect parameters $\phi^0 = ({\rm vec}(\alpha)^{0 \prime}, {\rm vec}(\gamma^{0})')'$ we impose the normalization
$
     \sum_{i=1}^I  \alpha_i^0  \, \alpha_i^{0 \prime}
       = \sum_{j=1}^J  \gamma^0_j  \, \gamma_j^{0  \prime} ,
$
and define the restricted parameter set
\begin{align*}
     \Phi
     &:= \left\{ 
       \phi \in \mathbbm{R}^{d_\phi}
       \; : \;
      \sum_{i=1}^I  \alpha_i^0  \, \alpha_i'
       = \sum_{j=1}^J  \gamma_j  \, \gamma_j^{0  \prime}
     \right\} ,
\end{align*}
where $d_v := \dim v$ for any vector $v$.  
Notice  that $\phi^0 \in \Phi$.
The  maximum likelihood estimator that imposes the normalization $\phi \in \Phi$ reads
\begin{align}
 (\widehat \beta, \widehat \phi)
  \, &= \,
\argmax_{(\beta, \phi)   \in \mathbbm{R}^{d_\beta} \times \Phi} \;   L(\beta,\phi).
   \label{EstimatorsConstrained}
\end{align}

Imposing $\widehat \phi \in \Phi$  is an infeasible normalization, because 
the true value of the parameters appears in the definition of $\Phi$.
However,
all our final asymptotic results are on the estimators $\widehat \beta$ and $\widehat \delta$,
which are invariant to the chosen normalization for $\widehat \phi$, that is, those results on $\widehat \beta$ and $\widehat \delta$
also hold unchanged for any other normalization, and imposing $\widehat \phi \in \Phi$
is simply a matter of convenience for the following proofs. 
There is always a need for a normalization choice when estimating the factor loadings and factors in
interactive fixed effect models, because the model only depends on the product $ \alpha_i'  \gamma_j$,
which is unchanged under the transformation
\begin{align}
    \alpha_i &\mapsto  A' \alpha_i  
    &
     \gamma_j &\mapsto A^{-1}  \gamma_j ,
     \label{TRAFO}
\end{align} 
for some invertible $R \times R$ matrix $A$.
Notice that in the definition of $\Phi$  there are $R^2$ normalization constraints,
which is exactly enough to uniquely determine the $R^2$ continuous degrees of freedom 
of the matrix $A$. In addition, there is still a discrete sign change possible ($\alpha_i \mapsto - \alpha_i$
and $ \gamma_j \mapsto -  \gamma_j$), and we assume in the following that this discrete choice is specified somehow
(e.g. by imposing $\alpha_{11}>0$) to make the estimator $\widehat \phi$ unique. The details of this final discrete choice
do not matter, as long as the same sign normalization is imposed on $\widehat \phi$ and $\phi^0$.

Our normalization constraints in the definition of $\Phi$ are linear in $\phi$. It is this linearity which makes this particular normalization  attractive for our purposes.
In particular, instead 
of imposing this normalization directly we can also impose it via a quadratic penalty function
by defining the penalized  objective function
\begin{align}
      \mathcal{L}(\beta,\phi) & = n^{-1/2} 
     \left[  L(\beta,\phi)
     - \frac b 2  \phi' \, V \, V' \, \phi   \right],
     \label{DefLcal}
\end{align}
where $b>0$ is some constant, and $V$ is a $d_\phi \times R^2$ matrix, which depends
on $\alpha^0$ and $\gamma^0$, and is implicitly defined by
\begin{align*}
     V' \, \phi   &= {\rm vec}\left[  
        \sum_{i=1}^I   \alpha_i^0    \alpha_i'
       - \sum_{j=1}^J \gamma_j   \gamma_j^{0 \, \prime}
          \right] .
\end{align*}
Thus, the above penalty term can also be expressed as
\begin{align*}
     \phi' \, V \, V' \, \phi  
     &= \left\|    \sum_{i=1}^I   \alpha_i^0    \alpha_i'
       - \sum_{j=1}^J \gamma_j   \gamma_j^{0 \, \prime}  \right\|_F^2 ,
\end{align*}
where $\|.\|_F$ denotes the Frobenius norm.
The constrained estimator in  \eqref{EstimatorsConstrained} can then equivalently be obtained by solving
the unconstrained problem
\begin{align*}
 (\widehat \beta, \widehat \phi)
  \, &= \,
\argmax_{(\beta, \phi)   \in \mathbbm{R}^{d_\beta + d_\phi}} \;   \mathcal{L}(\beta,\phi),
\end{align*}
and we also define 
\begin{align*}
  \widehat \phi(\beta)
  \, &= \,
\argmax_{ \phi    \in \mathbbm{R}^{  d_\phi}} \;   \mathcal{L}(\beta,\phi),
&
 \widehat \phi(\beta)  &= ( {\rm vec}(\widehat \alpha(\beta))' , {\rm vec}(\widehat \gamma(\beta))')'.
\end{align*}
Finally, we introduce 
the index sets ${\bf I} := \{1,\ldots,I\}$ and ${\bf J}:=\{1,\ldots,J\}$.

\subsection{Consistency}

\begin{lemma}
   \label{lemma:consistency}
   Let Assumption~\ref{ass:PanelA1} be satisfied. Then,
   $\| \widehat \beta - \beta^0 \| =  {\cal O}_P(I^{-3/8})$
   and  
   \begin{align*}
    \frac 1 {\sqrt{n}} \left\|  \widehat \alpha(\beta)  \widehat \gamma(\beta)'
                                   - \alpha^0 \gamma^{0 \prime} \right\|_F  
              &= {\cal O}_P(I^{-3/8}  + \| \beta - \beta^0 \| ), 
          &
           \frac 1 {\sqrt{I}} \| \widehat \phi(\beta) - \phi^0 \|  
              &= {\cal O}_P(I^{-3/8} + \| \beta - \beta^0 \| )    ,
   \end{align*}                         
   uniformly over $\beta$ in a $\epsilon$-neighborhood around $\beta^0$, for some $\epsilon>0$. 
\end{lemma}

\begin{proof}[\bf Proof of Lemma~\ref{lemma:consistency}]
For all $z_1, z_2 \in {\cal B}^0_{\varepsilon}$ a second order Taylor expansion of $ \ell_{ij}(z_1)$ around $z_2$ gives 
\begin{align}
   \ell_{ij}(z_1) - \ell_{ij}(z_2)
    &=  [\partial_z \ell_{ij}(z_1)] (z_1-z_2)
                                     - \ft 1 2 [\partial_{z^2} \ell_{ij}(\tilde z)] \,  (z_1-z_2)^2
    \nonumber \\
    &\geq [\partial_z \ell_{ij}(z_1)] (z_1-z_2)
                                     + \frac{b_{\min}} 2 \,  (z_1-z_2)^2
    \nonumber \\
    & =   \frac{b_{\min}} 2
         \left(z_1-z_2 + \frac 1 {b_{\min}} [\partial_z \ell_{ij}(z_1)] \right)^2
          -      \frac{1} {2 b_{\min}}
         [\partial_z \ell_{ij}(z_1)]^2,    
   \label{QuadraticBound}              
\end{align}
where $\tilde z \in [\min(z_1,z_2), \max(z_1,z_2)]$.
Let $e_{ij}  :=   \partial_z \ell_{ij}/b_{\min}$.
Using \eqref{QuadraticBound} we find that
\begin{align*}
    0 &\geq  \frac{1 }{\sqrt{IJ}} \left[ {\cal L}(\beta^0,\phi^0) -  {\cal L}(\widehat \beta,\widehat \phi) \right]
     =  \frac 1 {IJ}   \sum_{i,j \in {\mathcal{D}}} \left[ \ell_{ij}( z^0_{ij} ) - \ell_{ij}( \widehat z_{ij} ) \right]
    \nonumber \\
       &\geq   \frac {b_{\min}} {2 \, IJ} \sum_{i,j \in {\mathcal{D}}}
       \left[   (z^0_{ij} -  \widehat z_{ij} + e_{ij})^2 - e_{ij}^2 \right]
    \nonumber 
       =   \frac {b_{\min}} {2 \, IJ} \sum_{i=1}^I \sum_{j=1}^J
       \left[   (z^0_{ij} -  \widehat z_{ij} + e_{ij})^2 - e_{ij}^2 \right]
        + {\cal O}_P\left(  \frac{IJ - n} {IJ}  \right)
    \nonumber \\
      &=  \frac {b_{\min}} {2 \, IJ} \sum_{i=1}^I \sum_{j=1}^J
       \left\{   \left[ X'_{ij}  (\widehat \beta - \beta^0)
                          + \widehat \alpha_i^{\, \prime} \, \widehat \gamma_j
                          - \alpha^{0 \, \prime}_i \, \gamma^0_j
                          - e_{ij}  \right]^2 - e_{ij}^2 \right\}   + {\cal O}_P\left(  \frac{1} {IJ}  \right) .
\end{align*}
Note that the penalty term of the objective function does not enter here,
because it is zero when evaluated both at the estimator and at the true values of the parameters.
Let $e$ be the $I \times J$ matrix with entries $e_{ij}$.
Let $X_k$ be the $I \times J$ matrix with entries $X_{k,ij}$, $k=1,\ldots,d_\beta$.
Let $\beta \cdot X = \sum_k \beta_k X_k$.
In matrix notation, the above inequality reads
\begin{align*}
  &   \frac 1 {IJ}   {\rm Tr}(e' e)
   \\
             & \; \; \geq
          \frac 1 {IJ}     {\rm Tr} \left[   \left(  (\widehat \beta - \beta^0) \cdot X + \widehat \alpha \widehat \gamma'
                                   - \alpha^0 \gamma^{0 \prime}
                                   -e \right)  
                                   \left(  (\widehat \beta - \beta^0) \cdot X + \widehat \alpha \widehat \gamma'
                                   - \alpha^0 \gamma^{0 \prime}
                                   -e \right)'                       
                 \right]  + {\cal O}_P\left(  \frac{1} {IJ}  \right) .
\end{align*}
Analogous to the consistency proof for linear regression models with interactive fixed effects
in \cite{Bai:2009p3321} and \cite{MoonWeidner2017} we can conclude that
\begin{align}       
     \frac 1 {IJ}  {\rm Tr}(e' e)          
         &\geq  
          \frac 1 {IJ}  {\rm Tr} \left[ 
                  {\cal M}_{\alpha^0}
                                    \left(  (\widehat \beta - \beta^0) \cdot X -e \right)  
                     {\cal M}_{\widehat \gamma}                 
                                    \left(    (\widehat \beta - \beta^0) \cdot X -e \right)'  \right]     + {\cal O}_P\left(  \frac{1} {IJ}  \right)     
  \nonumber   \\                               
        &=  \frac 1 {IJ} \bigg[  
        {\rm Tr}(e' e) 
      +    {\rm Tr} \left[      {\cal M}_{\alpha^0}
                                    \left(  (\widehat \beta - \beta^0) \cdot X  \right)
                                     {\cal M}_{\widehat \gamma}   
                                      \left(  (\widehat \beta - \beta^0) \cdot X  \right)'   \right]   
       + 2 {\rm Tr} \left[   \left(  (\widehat \beta - \beta^0) \cdot X  \right)'
                                    e   \right]                                    
  \nonumber   \\    & \qquad      \qquad                                  
      + {\cal O}_P(  \| e\|^2 ) + {\cal O}_P(  \|\widehat \beta - \beta^0 \| \| e\| \max_k \| X_k \| )   \bigg] 
        + {\cal O}_P\left(  \frac{1} {IJ}  \right),                                                                 
    \label{ConsInequ}
\end{align} 
where we used that e.g.
\begin{align*}
     \left| {\rm Tr} \left(  X_k'
                  {\cal P}_{\alpha^0}
                                  e  \right) \right|
            &\leq    {\rm rank}\left(         X_k'  {\cal P}_{\alpha^0}  e \right) 
                  \left\|     X_k'     {\cal P}_{\alpha^0}   e  \right\|        
               \leq  \| X_k \| \| e\|  ,
    \\
      \left| {\rm Tr} \left(  e'
                  {\cal P}_{\alpha^0}
                                  e  \right) \right|
            &\leq    {\rm rank}\left(         e'  {\cal P}_{\alpha^0}  e \right) 
                  \left\|     e'     {\cal P}_{\alpha^0}   e  \right\|        
               \leq     \| e\|^2   .
\end{align*}
Lemma~D.6 in \cite{FW16} shows that under Assumption \ref{ass:PanelA1},
    $\| \partial_z \ell \| =  {\cal O}_P(I^{5/8} )$,
where $\partial_z \ell$ is the $I\times J$ matrix 
         with entries $\partial_z \ell_{ij}$. 
We thus also have $\|e\| = {\cal O}_P(I^{5/8} )$.
We furthermore have $\|X_k\|^2 \leq  \|X_k\|^2_{F} = \sum_{ij} X_{k,ij}^2 = {\cal O}_P(IJ)$,
so that $\|X_k\| = {\cal O}_P(\sqrt{IJ})$.
Hence, 
$ \| X_k \| \| e\|   =  {\cal O}_P(I^{13/8})$,
$  \| e\|^2 =    {\cal O}_P(I^{5/4})$, and
\begin{align*}
      {\rm Tr} \left(  X_k'  e   \right)  &= \frac 1 {b_{\min}} \sum_{ij} X_{ij}  \partial_z \ell_{ij} 
        =  {\cal O}_P( \sqrt{IJ} )  .
\end{align*}
Applying these results and the generalized collinearity assumption to \eqref{ConsInequ} gives
\begin{align*}
    0 &\geq  c \| \widehat \beta - \beta^0 \| +  {\cal O}_P( I^{-3/8}  \| \widehat \beta - \beta^0 \| ) 
       +  {\cal O}_P(I^{-3/4}  ) .
\end{align*}
This implies that $\| \widehat \beta - \beta^0 \| =  {\cal O}_P(I^{-3/8})$.

Define $e_{ij}(\beta) =  \partial_z \ell_{ij}(X'_{ij}  \beta + \alpha^0_i \gamma^{0 \, \prime}_j)/b_{\min}$.
Analogous to the above argument we find from
$ {\cal L}(\beta,\widehat \phi(\beta)) \geq  {\cal L}(\beta,\phi^0)$
that
\begin{align*}
    0 &\geq \sqrt{IJ} \left[ {\cal L}(\beta,\phi^0) -  {\cal L}(\beta,\widehat \phi(\beta)) \right]
=    \sum_{i,j} \left[ \ell_{ij}( X'_{ij}  \beta + \alpha^0_i \gamma^{0 \, \prime}_j ) - \ell_{ij}( 
     X'_{ij}  \beta + \widehat \alpha_i(\beta) \widehat \gamma'_j(\beta) ) \right]
    \nonumber \\
      &=  \frac {b_{\min}} {2 } \sum_{i,j}
       \left\{   \left[   \widehat \alpha_i(\beta) \widehat \gamma'_j(\beta)
                          - \alpha^0_i \gamma^{0 \, \prime}_j
                          - e_{ij}(\beta)   \right]^2 - [e_{ij}(\beta)]^2 \right\} .
\end{align*}
This implies that
\begin{align*}
    & {\rm Tr}(e(\beta)' e(\beta))
\geq
            {\rm Tr} \left[   \left(   \widehat \alpha(\beta) \widehat \gamma(\beta)'
                                   - \alpha^0 \gamma^{0 \prime}
                                   -e(\beta) \right)  
                                   \left(  \widehat \alpha(\beta) \widehat \gamma(\beta)'
                                   - \alpha^0 \gamma^{0 \prime}
                                   -e(\beta) \right)'                       
                 \right] 
      \\
         &=    {\rm Tr}(e(\beta)' e(\beta))                                              
           + \underbrace{ {\rm Tr} \left[   \left(   \widehat \alpha(\beta) \widehat \gamma(\beta)'
                                   - \alpha^0 \gamma^{0 \prime}
                                    \right)  
                                   \left(   \widehat \alpha(\beta) \widehat \gamma(\beta)'
                                   - \alpha^0 \gamma^{0 \prime}
                                   \right)'                       
                 \right] }_{=  \left\|  \widehat \alpha(\beta) \widehat \gamma(\beta)'
                                   - \alpha^0 \gamma^{0 \prime} \right\|^2_F}
                        + {\cal O}_P\left( \left\|  \widehat \alpha(\beta) \widehat \gamma(\beta)'
                                   - \alpha^0 \gamma^{0 \prime} \right\|_F
                        \|e(\beta)\|   
                                    \right)            .                      
\end{align*}
Since
$ \widehat \alpha(\beta) \widehat \gamma(\beta)'
                                   - \alpha^0 \gamma^{0 \prime}$ is at most of rank $2R$,  
 $\frac 1 {\sqrt{2R}} \left\|  \widehat \alpha(\beta) \widehat \gamma(\beta)'
                                   - \alpha^0 \gamma^{0 \prime} \right\|_F
                                   \leq \left\|  \widehat \alpha(\beta) \widehat \gamma(\beta)'
                                   - \alpha^0 \gamma^{0 \prime} \right\|
                                   \leq \left\|  \widehat \alpha(\beta) \widehat \gamma(\beta)'
                                   - \alpha^0 \gamma^{0 \prime} \right\|_F$,
                                i.e. the Frobenius and the    spectral norm are equivalent.                   
Since $e_{ij}(\beta) = e_{ij} + [X'_{ij}   (\beta-\beta^0)] \partial_{z^2} \ell_{ij} (X'_{ij}  \tilde \beta + \alpha^0_i \gamma^{0 \, \prime}_j)/b_{\min}$,
where $\tilde \beta$ lies between $\beta$ and $\beta^0$, 
we have
$\| e(\beta) \| \leq \|e\| 
   + {\cal O}_P(\sqrt{IJ} \| \beta - \beta^0 \| )$.                                                   
We thus find
\begin{align*}
     0 &\geq  \frac 1 {IJ} \left\|  \widehat \alpha(\beta) \widehat \gamma(\beta)'
                                   - \alpha^0 \gamma^{0 \prime} \right\|^2_F   
     +  {\cal O}_P\left[ 
     (I^{-3/8} + \| \beta - \beta^0 \| ) 
     \left\|  \widehat \alpha(\beta) \widehat \gamma(\beta)'
                                   - \alpha^0 \gamma^{0 \prime} \right\|_F
                                   /\sqrt{IJ} 
                                   \right]  .    
\end{align*}
From this we conclude that 
\begin{align}
     \frac 1 {\sqrt{IJ}} \left\|  \widehat \alpha(\beta) \widehat \gamma(\beta)'
                                   - \alpha^0 \gamma^{0 \prime} \right\|_F  &=  {\cal O}_P(I^{-3/8}  + \| \beta - \beta^0 \| ) .
                         \label{BoundAG1}          
\end{align} 
Next, using our normalization $\phi^0 \in \Phi$ and $\widehat \phi \in \Phi$,
\begin{align*}
     \alpha^{0 \prime}  \left[ \widehat \alpha(\beta) \widehat \gamma(\beta)'   - \alpha^0 \gamma^{0 \prime} \right] \gamma^0
     &=   \left[  \alpha^{0 \prime} \widehat \alpha(\beta)  \right]^2
        - \left[  \alpha^{0 \prime} \alpha^0  \right]^2 ,
\end{align*}
and therefore
\begin{align*}
  \left\|
     \left[  \frac 1 I \alpha^{0 \prime} \widehat \alpha(\beta)  \right]^2
        - \left[  \frac 1 I \alpha^{0 \prime} \alpha^0  \right]^2
    \right\|_F    
    &=  \frac 1 {I^2}  \left\|
    \alpha^{0 \prime}  \left[ \widehat \alpha(\beta) \widehat \gamma(\beta)'   - \alpha^0 \gamma^{0 \prime} \right] \gamma^0
    \right\|_F 
    \leq   \frac 1 {I^2}   \left\|  \alpha^{0} \right\|_F
    \left\| \widehat \alpha(\beta) \widehat \gamma(\beta)'   - \alpha^0 \gamma^{0 \prime} \right\|_F
      \left\| \gamma^0 \right\|
   \\   
    &=   \frac 1 {I^2}   {\cal O}(I^{1/2}) \, \sqrt{IJ} \,  {\cal O}_P(I^{-3/8}  + \| \beta - \beta^0 \| )  {\cal O}(J^{1/2}) 
     =  {\cal O}_P(I^{-3/8}  + \| \beta - \beta^0 \| ).
\end{align*}
Using the strong-factor assumption 
 $I^{-1}  \alpha^{0 \prime} \alpha^0    \to_P \Sigma_1 >0$ we thus have
\begin{align}
    \left[ I^{-1} \alpha^{0 \prime} \widehat \alpha(\beta) \right]^{-1}
     &=
     \left[  I^{-1}  \alpha^{0 \prime} \alpha^0  \right]^{-1}
     +  {\cal O}_P(I^{-3/8}  + \| \beta - \beta^0 \| ).
                           \label{BoundAG2}        
\end{align}
Again by the normalization $\widehat \phi \in \Phi$ we also have
\begin{align*}
     \left[ \widehat \alpha(\beta) \widehat \gamma(\beta)'   - \alpha^0 \gamma^{0 \prime} \right] \gamma^0
     &=     \widehat \alpha(\beta)     \alpha^{0 \prime} \widehat \alpha(\beta)  
        -    \alpha^0    \alpha^{0 \prime} \alpha^0   ,
\end{align*}
and therefore
\begin{align*}
      \widehat \alpha(\beta)  
      &=   \alpha^0      \left[  I^{-1}  \alpha^{0 \prime} \alpha^0  \right]  \left[ I^{-1}  \alpha^{0 \prime} \widehat \alpha(\beta)   \right]^{-1}
         +  I^{-1}  \left[ \widehat \alpha(\beta) \widehat \gamma(\beta)'   - \alpha^0 \gamma^{0 \prime} \right] \gamma^0
          \left[ I^{-1}   \alpha^{0 \prime} \widehat \alpha(\beta)   \right]^{-1} .
\end{align*}
Applying \eqref{BoundAG1} and \eqref{BoundAG2} thus gives
\begin{align*}
     I^{-1/2}  \left\| \widehat \alpha(\beta)  -  \alpha^0  \right\|_F
      &\leq    I^{-1/2}  \left\|  \alpha^0  \right\|_F
    \left\|  \mathbb{I}_R - \left[  I^{-1}  \alpha^{0 \prime} \alpha^0  \right]  \left[ I^{-1}  \alpha^{0 \prime} \widehat \alpha(\beta)   \right]^{-1} \right\|_F
    \\ & \quad
    + I^{-3/2}  \left\| \widehat \alpha(\beta) \widehat \gamma(\beta)'   - \alpha^0 \gamma^{0 \prime} \right\|_F
    \left\| \gamma^0 \right\|_F
        \left\|  \left[ I^{-1}   \alpha^{0 \prime} \widehat \alpha(\beta)   \right]^{-1}  \right\|_F
     \\
     &=   I^{-1/2} {\cal O}(I^{1/2})     {\cal O}_P(I^{-3/8}  + \| \beta - \beta^0 \| )
     + I^{-3/2}  \, \sqrt{IJ} \,   {\cal O}_P(I^{-3/8}  + \| \beta - \beta^0 \| )  {\cal O}(J^{1/2})  {\cal O}(1)
     \\
     &=  {\cal O}_P(I^{-3/8}  + \| \beta - \beta^0 \| ) .
\end{align*}
Analogously we conclude that 
$ J^{-1/2}   \|\widehat \gamma(\beta) - \gamma^0\| = {\cal O}_P(I^{-3/8}  + \| \beta - \beta^0 \| )$,
and therefore
$   \frac 1 {\sqrt{I}} \| \widehat \phi(\beta) - \phi^0 \|  
              = {\cal O}_P(I^{-3/8} + \| \beta - \beta^0 \| ) $.
\end{proof}

\subsection{Inverse Expected Incidental Parameter Hessian}

We define the expected incidental parameter Hessian for the log-likelihood with and without the penalty term as
\begin{align*}
     \overline{\cal H} &:= \E[ - \partial_{\phi \phi'} {\cal L}] 
     =  \overline{\cal H}^* +  \frac{b} {\sqrt{n}} \, VV',      &
   \overline{\cal H}^* &:= \E[ - \partial_{\phi \phi'} {\cal L}^*].
\end{align*}
Our definition of ${\cal L}^*(\beta,\phi)  =  n^{-1/2} \,  L(\beta,\phi)$
includes the factor $n^{-1/2}$, which makes sure that the eigenvalues of $\overline{\cal H}^*$ remain of order
one asymptotically as $I,J \rightarrow \infty$ at the same rate. Similarly, the factor $1/\sqrt{n}$ in the second term 
of  $\overline{\cal H}$ makes sure that the eigenvalues of $  \frac{b} {\sqrt{n}} \, VV'$ remain of order one asymptotically.
The Hessian matrix $\overline{\cal H}^* $ has $R^2$ zero eigenvalues 
corresponding to the $R^2$ flat directions in the log-likelihood described by the
transformations \eqref{TRAFO} that leave the likelihood unchanged. Correspondingly, the matrix $VV'$ is exactly of rank $R^2$,
making sure that $  \overline{\cal H} $ has no more zero eigenvalues and is invertible, as formally shown by Lemma~\ref{lemma:HessianAdditive}
below. Those considerations explain why we have chosen the penalty term
$\frac b 2  \phi' \, V \, V' \, \phi $ and the pre-factor $n^{-1/2}$ in our definition of  ${\cal L}(\beta,\phi) $
in \eqref{DefLcal} above.

Let $a={\rm vec}(\alpha)$ and $c={\rm vec}(\gamma)$, so that  $\phi = (a', c')'$.
Correspondingly we can decompose the Hessian matrix,
\begin{align*}
\overline{\cal H}^* =  
\left(\begin{array}{cc}   \E[ - \partial_{aa'} {\cal L}^*]  &  \E[ - \partial_{ac'} {\cal L}^*]   \\  \E[ - \partial_{ca'} {\cal L}^*]  &  \E[ - \partial_{cc'} {\cal L}^*] 
\end{array}\right) 
=:
\left(\begin{array}{cc}  \overline{\mathcal{H}}_{(\alpha\alpha)}^* & \overline{\mathcal{H}}_{(\alpha\gamma)}^*  \\ {[\overline{\mathcal{H}}_{(\alpha\gamma)}^*]}' & \overline{\mathcal{H}}_{(\gamma\gamma)}^* 
\end{array}\right) .
\end{align*}
Here, $\overline{\mathcal{H}}_{(\alpha\alpha)}^*$ is a block-diagonal $IR \times IR$ matrix
with $R \times R$ diagonal blocks,
and $ \overline{\mathcal{H}}_{(\gamma\gamma)}^* $ is a block-diagonal $JR \times JR$ matrix
with $R \times R$ diagonal blocks, that is
\begin{align*}
     \overline{\mathcal{H}}_{(\alpha\alpha)}^*
     &= \text{diag}\left( 
     \left[ \frac 1 {\sqrt{n}}  \sum_{j \in {\cal D}_i} \E(- \partial_{z^2} \ell_{ij})    \gamma^0_j \gamma^{0 \prime}_j \right]_{i \in {\bf I}} \right),
     &
          \overline{\mathcal{H}}_{(\gamma \gamma)}^*
     &= \text{diag}\left(
         \left[  \frac 1 {\sqrt{n}}  \sum_{i \in {\cal D}_j} \E(- \partial_{z^2} \ell_{ij})    \alpha^0_j \alpha^{0 \prime}_j
         \right]_{j \in {\bf J}}   \right) .
\end{align*}
For any matrix $A$ with elements $A_{kl}$,   let $\|A\|_{\max} = \max_{k,l} |A_{kl}|$. Notice that $\|.\|_{\max}$ is not
sub-multiplicative, so it is not a matrix norm.

\begin{lemma}
   \label{lemma:HessianAdditive}
   Under Assumption~\ref{ass:PanelA1},
   \begin{equation*}
   \left\| \overline {\cal H}^{-1} -
   \diag \left(  \overline {\cal H}_{(\alpha \alpha)}^*, \overline {\cal H}_{(\gamma \gamma)}^*
    \right)^{-1}
    \right\|_{\max}
        =  \mathcal{O}\left( n^{-1/2} \right)  .
   \end{equation*}
\end{lemma}

\begin{proof}[\bf Proof]
    We consider the case ${\cal D}={\cal D}_0$ in the following.
    We decompose
    \begin{align*}
\overline{\cal H}^* =   \underbrace{
\left(\begin{array}{cc}  \overline{\mathcal{H}}_{(\alpha\alpha)}^* & 0  \\0 & \overline{\mathcal{H}}_{(\gamma\gamma)}^* 
\end{array}\right) 
}_{=: \overline {\cal D}}
+
\underbrace{
\left(\begin{array}{cc} 0 & \overline{\mathcal{H}}_{(\alpha\gamma)}^*  \\ {[\overline{\mathcal{H}}_{(\alpha\gamma)}^*]}' & 0
\end{array}\right) 
}_{=:  \overline {\cal A}^*} ,
\end{align*}
and let $\overline {\cal A}  := \overline {\cal A}^* +     \frac{b} {\sqrt{n}} \, VV'$. Then, 
$\overline{\cal H}  = \overline {\cal D} + \overline {\cal A} $. 
The $IR \times JR$ matrix $ \overline{\mathcal{H}}_{(\alpha\gamma)}^*$ is
composed of $I \times J$ blocks of size $R \times R$ as follows
\begin{align*}
 \overline{\mathcal{H}}_{(\alpha\gamma)}^* & =
  \left[  \frac 1 {\sqrt{n}}   \E(- \partial_{z^2} \ell_{ij})    \gamma^0_j \alpha^{0 \prime}_i \right]_{i \in {\bf I}, j \in {\bf J}} ,
\end{align*}
and similarly we have blocks for the   $(I+J)R \times (I+J)R$ matrix $VV'$
\begin{align*}
VV' &=\left(\begin{array}{cc}
   \left[        \alpha^0_i \alpha^{0 \prime}_{i^*} \right]_{i,i^* \in {\bf I} } 
   &
  \left[  -      \gamma^0_j \alpha^{0 \prime}_i \right]_{i \in {\bf I}, j \in {\bf J}}
  \\ 
  \left[     -   \alpha^0_i \gamma^{0 \prime}_j \right]_{j \in {\bf J}, i \in {\bf I}}  &
     \left[        \gamma^0_j \gamma^{0 \prime}_{j^*} \right]_{j,j^* \in {\bf J} } 
\end{array}\right) 
 =: \left(\begin{array}{cc}
   \left[   VV' \right]_{(\alpha \alpha)} 
   &
  \left[   VV' \right]_{(\alpha \gamma)} 
  \\ 
    \left[   VV' \right]_{(\gamma \alpha)}   &
     \left[   VV' \right]_{(\gamma \gamma)}  
\end{array}\right) .
\end{align*}
Let $b^* :=\min\{ b_{\min}, b\}$. 
For symmetric matrices $A$ and $B$ we write $A \geq B$ if $A-B$ is positive semi-definite. 
We have
\begin{align*}
     \overline {\cal A} -   \frac{b-b^*} {\sqrt{n}} \, VV'
     -    \frac{b^*} {\sqrt{n}}
     \left(\begin{array}{cc}
   \left[   VV' \right]_{(\alpha \alpha)} 
   &
0
  \\ 
   0   &
     \left[   VV' \right]_{(\gamma \gamma)}  
\end{array}\right)
     &=   \left(\begin{array}{cc}
   0
   &
 \overline{\mathcal{H}}_{(\alpha\gamma)}^*   - \frac{b^*} {\sqrt{n}} \left[   VV' \right]_{(\alpha \gamma)} 
  \\ 
 \left[ \overline{\mathcal{H}}_{(\alpha\gamma)}^* \right]'    - \frac{b^*} {\sqrt{n}}   \left[   VV' \right]_{(\gamma \alpha)}   &
   0
\end{array}\right) ,
\end{align*}
and since $V'V\geq 0$ (implying also $  \left[   VV' \right]_{(\alpha \alpha)} \geq 0$ and $  \left[   VV' \right]_{(\gamma \gamma)}  \geq 0$)
we thus have
\begin{align*}
     \overline {\cal A}  
     &\geq   \left(\begin{array}{cc}
   0
   &
 \overline{\mathcal{H}}_{(\alpha\gamma)}^*   - \frac{b^*} {\sqrt{n}} \left[   VV' \right]_{(\alpha \gamma)} 
  \\ 
 \left[ \overline{\mathcal{H}}_{(\alpha\gamma)}^* \right]'    - \frac{b^*} {\sqrt{n}}   \left[   VV' \right]_{(\gamma \alpha)}   &
   0
\end{array}\right) .
\end{align*}
Using this and $  \E[ - \partial_{\phi \phi'} \ell_{ij}] \geq 0$ we obtain
\begin{align*}
     \overline{\cal H}
     &=  \overline {\cal D} + \overline {\cal A} 
   \\  
     &\geq  
      \overline {\cal D}
      +  \left(\begin{array}{cc}
   0
   &
 \overline{\mathcal{H}}_{(\alpha\gamma)}^*   - \frac{b^*} {\sqrt{n}} \left[   VV' \right]_{(\alpha \gamma)} 
  \\ 
 \left[ \overline{\mathcal{H}}_{(\alpha\gamma)}^* \right]'    - \frac{b^*} {\sqrt{n}}   \left[   VV' \right]_{(\gamma \alpha)}   &
   0
\end{array}\right) 
 - 
 \underbrace{
 n^{-1} \sum_{i=1}^I \sum_{j=1}^J  \E[ - \partial_{\phi \phi'} \ell_{ij}]   \frac{\E(- \partial_{z^2} \ell_{ij}) - b^*  }  { \E(- \partial_{z^2} \ell_{ij})}
 }_{\geq 0}
 \\
 &=
 b^*
 \left(\begin{array}{cc} \text{diag}\left(
         \left[  \frac 1 {\sqrt{n}}  \sum_{i=1}^I     \gamma^0_i \gamma^{0 \prime}_i
         \right]_{j \in {\bf J}}   \right)
          & 0  \\0 & \left[  \frac 1 {\sqrt{n}}  \sum_{j=1}^J     \alpha^0_j \alpha^{0 \prime}_j
         \right]_{j \in {\bf J}} 
\end{array}\right) 
\\
 &= b^*  \left(\begin{array}{cc}  
          n^{-1/2} \, \mathbb{I}_I \otimes \gamma^{0 \prime} \gamma^0
          & 0  \\0 &  n^{-1/2} \, \mathbb{I}_J \otimes \alpha^{0 \prime} \alpha^0
\end{array}\right) 
   \geq c \, \mathbb{I}_{(I+J)R} ,
\end{align*}
wpa1 (with probability approaching one), where existence of $c>0$ is guaranteed by our strong factor Assumptions~\ref{ass:PanelA1}(v).
The result of the last display implies  that
\begin{align}
   \overline{\cal H}^{-1} \leq c^{-1}  \, \mathbb{I}_{(I+J)R}.
     \label{UpperBoundHinverseCrude}
\end{align}
We have thus obtained a spectral bound for $   \overline{\cal H}^{-1}$. This turns out to be the key step in the proof. 
The remainder of the proof is just a relatively straightforward expansion of $   \overline{\cal H}^{-1}$.
Namely, using $\overline{\cal H}  = \overline {\cal D} + \overline {\cal A} $ we find that
\begin{align*}
     \overline{\cal H}^{-1}
     &= \overline {\cal D}^{-1}
      -  \overline {\cal D}^{-1} \, \overline {\cal A} \,  \overline {\cal D}^{-1}
      +  \left[ \overline {\cal D}^{-1}  \,  \overline{\cal H}      \,  \overline {\cal D}^{-1}
      -   2  \overline {\cal D}^{-1}
      +      \overline{\cal H}^{-1}   \right]
     \\   
     &= \overline {\cal D}^{-1}
      -  \overline {\cal D}^{-1} \, \overline {\cal A} \,  \overline {\cal D}^{-1}
      +  \overline {\cal D}^{-1} \left( \overline{\cal H} - \overline {\cal D} \right) \overline{\cal H}^{-1} \,  \left( \overline{\cal H} - \overline {\cal D} \right) \,  \overline {\cal D}^{-1}
    \\  
     &= \overline {\cal D}^{-1}
      -  \overline {\cal D}^{-1} \, \overline {\cal A} \,  \overline {\cal D}^{-1}
      +  \overline {\cal D}^{-1} \, \overline {\cal A}  \,  \overline{\cal H}^{-1} \, \overline {\cal A} \,  \overline {\cal D}^{-1}
    \\
    &\leq   \overline {\cal D}^{-1}
      -  \overline {\cal D}^{-1} \, \overline {\cal A} \,  \overline {\cal D}^{-1}
      +  c^{-1} \overline {\cal D}^{-1} \, \overline {\cal A}^2  \overline {\cal D}^{-1} ,
\end{align*}
and therefore
\begin{align*}
     \left\|  \overline{\cal H}^{-1} -  \overline {\cal D}^{-1} \right\|_{\max}
     \leq 
      \left\|   \overline {\cal D}^{-1} \, \overline {\cal A} \,  \overline {\cal D}^{-1} \right\|_{\max}
      +     c^{-1}    \left\| \overline {\cal D}^{-1} \, \overline {\cal A}^2  \overline {\cal D}^{-1}  \right\|_{\max} .
\end{align*}
From the expressions for $\overline {\cal D}$ and $\overline {\cal A} $ above 
one finds that $\overline {\cal D}$ is block-diagonal with entries of order one,
and $ \left\|   \overline {\cal A}  \right\|_{\max} = {\cal O}(n^{-1/2})$, which implies 
$ \left\|   \overline {\cal A}^2  \right\|_{\max} = {\cal O}((I+J) n^{-1})  = {\cal O}(n^{-1/2})$.
The RHS of the last display
is therefore indeed of order $n^{-1/2}$.
\end{proof}

\subsection{Local Concavity of the Objective Function}

The consistency results for $\widehat \beta$ and $\widehat \phi(\beta)$ in 
Lemma~\ref{lemma:consistency} 
 provide   initial convergence rates,
 implying that we only need to consider a shrinking neighborhood around $\beta^0$ and $\phi^0$
 for the remaining asymptotic analysis. 
The following lemma shows that the objective function
 ${\cal L}(\beta,\phi)$ is strictly concave in such a local neighborhood.
 Later in the proof
this strict concavity will allow 
us to apply the general expansion results in
\cite{FW16}.

   Analogously to the expected incidental parameter Hessian  $\overline {\cal H}$ at the true parameters
that was discussed above, we now introduce the following notation for incidental parameter Hessian (without expectations, and not necessarily at the true parameters),
\begin{align*}
 {\cal H}(\beta,\phi) := - \partial_{\phi \phi'} {\cal L}(\beta,\phi) =  
\left(\begin{array}{cc}  \mathcal{H}_{(\alpha\alpha)}^*(\beta,\phi) & \mathcal{H}_{(\alpha\gamma)}^*(\beta,\phi)  \\ {[\mathcal{H}_{(\alpha\gamma)}^*(\beta,\phi)]}' & \mathcal{H}_{(\gamma\gamma)}^*(\beta,\phi)
\end{array}\right) 
+  \frac{b} {\sqrt{n}} \, V  V' .
\end{align*}

\begin{lemma}
   \label{lemma:concavity}
   Let Assumption~\ref{ass:PanelA1} be satisfied,
   and let $r_\beta = r_{\beta,n} = o_P(1)$
   and $r_\phi = r_{\phi,n} = o_P(n^{1/4})$.    
   Then, ${\cal H}(\beta,\phi)$ is positive definite 
   for all $\beta \in {\cal B}(r_\beta,\beta^0)$ and 
   $\phi \in {\cal B}(r_\phi,\phi^0)$, wpa1, where
   ${\cal B}(r_\beta,\beta^0)$  is an $r_\beta$-ball around $\beta^0$ and ${\cal B}(r_\phi,\phi^0)$
is $r_\phi$-ball around $\phi^0$, both under the  Euclidian norm.    
   This implies that
 ${\cal L}(\beta,\phi)$ is strictly concave
 in $\phi \in {\cal B}(r_\phi,\phi^0)$ wpa1,
 for all $\beta \in {\cal B}(r_\beta,\beta^0)$.
\end{lemma} 

\begin{proof}[\bf Proof]
 Let   $\ell_{ij}(\beta,\pi_{ij}) :=\ell_{ij}(z_{ij})$, where $\pi_{ij}=\alpha'_i \gamma_j$ and  $z_{ij}=X_{ij}' \beta + \alpha'_i \gamma_j$.
Then, 
\begin{align*}
     {\mathcal{H}}_{(\alpha\alpha)}^*(\beta,\phi)
     &= \text{diag}\left( 
     \left[ \frac 1 {\sqrt{n}}  \sum_{j \in {\cal D}_i}  [- \partial_{z^2} \ell_{ij}(\beta,\pi_{ij}) ]    \gamma^0_j \gamma^{0 \prime}_j \right]_{i \in {\bf I}} \right),
     \\
          {\mathcal{H}}_{(\gamma \gamma)}^*(\beta,\phi)
     &= \text{diag}\left(
         \left[  \frac 1 {\sqrt{n}}  \sum_{i \in {\cal D}_j}  [- \partial_{z^2} \ell_{ij}(\beta,\pi_{ij}) ]    \alpha^0_j \alpha^{0 \prime}_j
         \right]_{j \in {\bf J}}   \right) ,
    \\
 {\mathcal{H}}_{(\alpha\gamma)}^*(\beta,\phi) & =
  \left\{  \frac 1 {\sqrt{n}}   [- \partial_{z^2} \ell_{ij}(\beta,\pi_{ij}) ]     \gamma^0_j \alpha^{0 \prime}_i 
  +  \frac 1 {\sqrt{n}} [-\partial_{z} \ell_{ij}(z_{ij})] \, \mathbb{I}_R \right\}_{i \in {\bf I}, j \in {\bf J}} ,
\end{align*}
We  decompose the Hessian into the contribution from the first
and from the second derivative of the log-likelihood,  namely  ${\cal H}(\beta,\phi) = H(\beta,\phi) + F(\beta,\phi)$,
where 
\begin{align*}
   F(\beta,\phi)  &= \left(\begin{array}{cc}
                                         0_{N \times N}    &   F_{(\alpha \gamma)}(\beta,\phi) 
                                         \\   
                                         {[  F_{(\alpha \gamma)}(\beta,\phi) ]'} &  0_{T \times T}
                               \end{array}\right) ,
                               &
     F_{(\alpha \gamma)}(\beta,\phi)  & =
  \left\{     \frac 1 {\sqrt{n}} [-\partial_{z} \ell_{ij}(z_{ij})] \, \mathbb{I}_R \right\}_{i \in {\bf I}, j \in {\bf J}}             .               
\end{align*}
Notice that $H(\beta,\phi) $ has the same structure as $\overline{\mathcal H}$.
Analogously  to the bound \eqref{UpperBoundHinverseCrude} derived in the 
proof of Lemma~\ref{lemma:HessianAdditive} we can thus show that
there exists a constant $c>0$ such that wpa1 we have,
for $\phi \in {\cal B}(r_\phi,\phi^0)$ 
 and $\beta \in {\cal B}(r_\beta,\beta^0)$,
\begin{align*}
   H(\beta,\phi)  \geq c  \, \mathbb{I}_{(I+J)R}.
\end{align*}
The new terms that need to be accounted for here are the first derivative terms $ F(\beta,\phi) $, which are zero in expectation
at the true parameter and therefore did not show up in our discussion of $\overline {\mathcal{H}}$ above.
The goal in the following is to show that $\|  F(\beta,\phi) \| = o_P(1)$,
or equivalently $\|  F_{(\alpha \gamma)}(\beta,\phi) \| = o_P(1)$,
 within the shrinking neighborhood 
of the true parameters. Here, $\|.\|$ refers to the spectral norm.
 
For ease of notation we consider $R=1$ in the remainder of this proof. 
Then,     $F_{(\alpha \gamma)ij}(\beta,\phi) = - \frac 1 {\sqrt{n}} \partial_{\pi} \ell_{ij}(\beta,\alpha'_i \gamma_j) $.
    A Taylor expansion gives
    \begin{align*}
       \partial_{\pi} \ell_{ij}(\beta,\alpha'_i \gamma_j)
       &=  \partial_{\pi} \ell_{ij}(\beta^0,\alpha^0_i \gamma^{0 \, \prime}_j)
          + (\beta-\beta^0)' 
          \partial_{\beta \pi} \ell_{ij}(\tilde \beta_{ij},\tilde \pi_{ij})
         +  (\alpha'_i \gamma_j - \alpha^0_i \gamma^{0 \, \prime}_j)
         \partial_{\pi^2} \ell_{ij}(\tilde \beta_{ij},\tilde \pi_{ij}) .
    \end{align*}
    The spectral norm of the $I\times J$ matrix with
    entries $ \partial_{\beta_k \pi} \ell_{ij}(\tilde \beta_{ij},\tilde \pi_{ij})$ is bounded
    by the Frobenius norm of this matrix, which is of order $\sqrt{n}$, since we assume uniformly bounded moments
    for $ \partial_{\beta_k \pi} \ell_{ij}(\tilde \beta_{ij},\tilde \pi_{ij})$.
     The spectral norm of the $I\times J$ matrix with
    entries $(\alpha'_i \gamma_j - \alpha^0_i \gamma^{0 \, \prime}_j)
         \partial_{\pi^2} \ell_{ij}(\tilde \beta_{ij},\tilde \pi_{ij})$
    is also bounded by the Frobenius   norm of this matrix,
    which equals $\sqrt{\sum_{ij} (\alpha'_i \gamma_j - \alpha^0_i \gamma^{0 \, \prime}_j)^2
        [ \partial_{\pi^2} \ell_{ij}(\tilde \beta_{ij},\tilde \pi_{ij})]^2}$  
    and thus bounded by     $b_{\max} \sqrt{\sum_{ij} (\alpha'_i \gamma_j - \alpha^0_i \gamma^{0 \, \prime}_j)^2 }
     = b_{\max} \| \alpha  \gamma' - \alpha^0 \gamma^{0 \prime} \|_F$.    We thus find
    \begin{align*}
       \left\| F_{(\alpha \gamma)ij}(\beta,\phi) \right\|
       &\leq
        \frac 1 {\sqrt{n}} \left(  \|   \partial_{\pi} \ell_{ij}  \|
               + {\cal O}_P(\sqrt{n})   \| \beta-\beta^0 \|  
               + b_{\max} \| \alpha  \gamma' - \alpha^0 \gamma^{0 \prime} \|_F
            \right)
       \\
        &=    {\cal O}_P( \frac 1 {\sqrt{n}}   I^{5/8} )
           +     {\cal O}_P(r_\beta) 
           +    {\cal O}_P( r_\phi /\sqrt{I})   
       \\
         &= o_P(1),    
    \end{align*}
    for $\phi \in {\cal B}(r_\phi,\phi^0)$ 
 and $\beta \in {\cal B}(r_\beta,\beta^0)$,    
    where we also used that 
    $ \| \alpha  \gamma' - \alpha^0 \gamma^{0 \prime} \|_F
     = {\cal O}_P(\sqrt{I}) \| \phi - \phi^0 \|  $.
    
    Combining the result in the last display with \eqref{UpperBoundHinverseCrude}
    we find that there exists a constant $c>0$ such that wpa1 we have,
for $\phi \in {\cal B}(r_\phi,\phi^0)$ 
 and $\beta \in {\cal B}(r_\beta,\beta^0)$,
\begin{align*}
   {\cal H}(\beta,\phi)  \geq c   \, \mathbb{I}_{(I+J)R}.
\end{align*}
We have thus shown that   ${\cal L}(\beta,\phi)$ is indeed strictly concave
(or that  $-{\cal L}(\beta,\phi)$ is strictly convex)  within this shrinking neighborhood.
\end{proof}

\subsection{Stochastic Expansion}

Once we have the consistency result of Lemma~\ref{lemma:consistency}
and the local strict concavity result of Lemma~\ref{lemma:concavity}, then the derivation of the stochastic expansion
of the fixed effect estimators $\widehat \beta$ and $\widehat \delta$ does not rely on
the specific single index and interactive fixed effect structure of our model.
Some of the conceptual issues indeed become more transparent when ignoring that structure.
Therefore,
in this subsection, let $\ell_{ij}(\beta,\alpha_i,\gamma_j) :=  \ell_{ij}(X_{ij}' \beta + \alpha'_i \gamma_j)$
and   $\Delta_{ij}(\beta, \alpha_i , \gamma_j) := \Delta_{ij}(\beta,\pi_{ij})$. 
Remember that our fixed effect estimators $\widehat \beta$
and $\widehat \gamma$ maximize
the objective function
$$   \mathcal{L}(\beta,\phi)
      =  n^{-1/2}
        \left[
            \sum_{(i,j) \in {\mathcal{D}}}  \ell_{ij}(\beta,\alpha_i,\gamma_j)
         +   \frac b 2  \phi' V V' \phi \right],$$
where $\phi=[(\alpha_i')_{i \in {\bf I} }, (\gamma_j')_{j \in {\bf J}} ]'$.
The APE is $\delta^0 =  \Delta(\beta^0,\phi^0) = \frac 1 n \sum_{(i,j) \in {\mathcal{D}}}    \Delta_{ij}(\beta^0, \alpha^0_i , \gamma^0_j)$,
and the corresponding plug-in estimator reads
$\widehat \delta =   \Delta (\widehat \beta, \widehat \phi)$.
For partial derivatives of $\ell_{ij}(\beta,\alpha_i,\gamma_j) $
and $\Delta (\widehat \beta, \widehat \phi)$ we use superscripts in the following,
 expectations are always conditional on $\phi$ and are indicated by a bar,
and arguments are omitted when evaluated at the true parameters. 
 For example,
$  \overline  \ell_{ij}^{\, \alpha_i \alpha_i}  $ is the $d_\alpha \times d_\alpha$ expected Hessian  matrix of $\ell_{ij}(\beta,\alpha_i,\gamma_j)$ with respect to $\alpha_i$ 
evaluated at the true parameters. This is the notation also used in \cite{ARE}, but here the $\alpha_i$ and $\gamma_j$ are vectors
of length $d_\alpha$ and $d_\gamma$, respectively. For our interactive fixed effect model we have $d_\alpha = d_\gamma =R$, but this is
not used in the rest of this subsection. The advantage of this generality is that, for example, the following formulas are also applicable 
to models where in addition to the interactive effects we include separate additive effects  in the single index.

It is convenient to make the log-likelihood   information-orthogonal between $\beta$ and
 the incidental parameters. This can be achieved by the transformation\footnote{
 This  transformation corresponds to the reparameterization
 $\alpha_i^* = \alpha_i -  \xi^{(\alpha)}_i  \beta$ 
   and $\gamma_j^* = \gamma_j  -  \xi^{(\gamma)}_j \beta$.
 The log-likelihood with respect to these parameters is
  $ \ell_{ij}(\beta,    \alpha_i^* +   \xi^{(\alpha)}_i  \beta ,     \gamma_j^* +  \xi^{(\gamma)}_j \beta ) =: \ell^*_{ij}(\beta,   \alpha_i^*  , \gamma_j^* )$,
 which gives  our definition of $\ell^*_{ij}$  after
 renaming $(\alpha_i^* ,  \gamma_j^*)$
  as $(\alpha_i, \gamma_j)$ again.
 }
\begin{align*}
    \ell^*_{ij}(\beta,\alpha_i,\gamma_j)  &:=
      \ell_{ij}(\beta,  \alpha_i   +  \xi^{(\alpha)}_i  \beta , \gamma_j  +  \xi^{(\gamma)}_j \beta) ,
\\
    \Delta^*_{ij}(\beta,\alpha_i,\gamma_j)  &:=
      \Delta_{ij}(\beta,  \alpha_i   +  \xi^{(\alpha)}_i  \beta , \gamma_j  +  \xi^{(\gamma)}_j \beta) ,
\end{align*} 
where the $d_\alpha \times d_{\beta}$ matrices $\xi^{(\alpha)}_i$, and
the $d_\gamma \times d_{\beta} $ matrices $\xi^{(\gamma)}_j $ are a solution to the system of equations
\begin{align*}
    \sum_{j \in {\cal D}_i} \left[ \overline \ell_{ij}^{\,  \alpha_i \beta}  
                      +   \overline  \ell_{ij}^{\, \alpha_i \alpha_i}    \xi^{(\alpha)}_i 
                      +   \overline  \ell_{ij}^{\, \alpha'_i \gamma_j}   \xi^{(\gamma)}_j \right] &= 0 ,
          \qquad i=1,\ldots,I ,\notag
     \\          
   \sum_{i \in {\cal D}_j } \left[ \overline \ell_{ij}^{\,  \gamma_j \beta}  
                      +     \overline  \ell_{ij}^{\, \gamma_j \alpha_i}   \xi^{(\alpha)}_i  
                      +     \overline  \ell_{ij}^{\, \gamma_j \gamma_j}  \xi^{(\gamma)}_j\right] &= 0 ,
          \qquad j=1,\ldots,J .
\end{align*}
Analogously, let the $d_\alpha$-vectors $\psi^{(\alpha)}_i$ and the $d_\gamma$-vectors $\psi^{(\gamma)}_j$ be solutions
to the system of equations
\begin{align*}
    \sum_{j \in {\cal D}_i} \left[ \overline \Delta_{ij}^{\, \alpha_i}  
                      +    \overline  \ell_{ij}^{\, \alpha_i \alpha_i}   \psi^{(\alpha)}_i 
                      +     \overline  \ell_{ij}^{\, \alpha'_i \gamma_j}   \psi^{(\gamma)}_j   \right] &= 0 ,
          \qquad i=1,\ldots,I ,\notag
     \\          
   \sum_{i \in {\cal D}_j } \left[ \overline \Delta_{ij}^{\, \gamma_j}  
                      +   \overline  \ell_{ij}^{\, \gamma_j \alpha_i}    \psi^{(\alpha)}_i
                      +   \overline  \ell_{ij}^{\, \gamma_j \gamma_j}  \psi^{(\gamma)}_j   \right] &= 0 ,
          \qquad j=1,\ldots,J .
\end{align*}
Finally, let
\begin{align*}
 \overline W &= - \, \frac 1 {\sqrt{n}} \,
                \left( \overline {\cal L}^{\beta \beta}
                     +  \overline {\cal L}^{\beta \phi} \; \overline {\cal H}^{-1} \;
              \overline {\cal L}^{\phi \beta} \right) 
            =   - \, \frac 1 {\sqrt{n}} \,
            \overline {\cal L}^{\, * \, \beta \beta}
            = \frac 1 n    \sum_{(i,j) \in {\mathcal{D}}}      \overline \ell_{ij}^{\, * \, \beta \beta} .
\end{align*}
The $d_\beta \times d_\beta$ matrix $\overline W_{\infty}$ defined in Assumption \eqref{th:BothEffects}
is simply the probability limit of $ \overline W$, that is, 
 $ \overline W_{\infty}   =  \EE \;   \overline W$ in main text notation.

\begin{theorem}[Stochastic Expansion for $\widehat \beta$ and $\widehat \delta$]
     \label{th:GeneralExpansion}
     Let Assumption~\ref{ass:PanelA1} be satisfied.
    We then have
    \begin{align*}
        \sqrt{n} \left( \widehat \beta - \beta^0 \right)  &= \overline W_{\infty}^{-1}
       \,  U + o_P(1) ,
    \end{align*}
    where the $d_\beta$-vector  $U$ has elements
    \begin{align*}
         U_k & :=   \frac 1 {\sqrt{n}}    \sum_{(i,j) \in {\cal D}}
         \left\{   \ell_{ij}^{\, * \, \beta_k} 
            -  \Ep \left[ \left(   \ell_{ij}^{\, * \, \beta_k \alpha_i}  \right)'
             \left( \sum_{h \in {\cal D}_i} \overline \ell_{ih}^{\, \alpha_i \alpha_i} \right)^{-1} \ell_{ij}^{\alpha_i}  \right]
             -  \Ep  \left[   \left(   \ell_{ij}^{\, * \, \beta_k \gamma_j}  \right)'
             \left( \sum_{h \in {\cal D}_j} \overline \ell_{hj}^{\, \gamma_j \gamma_j} \right)^{-1} \ell_{ij}^{\gamma_j}  \right]
             \right.
        \\
    & \qquad \qquad \qquad 
    + \frac 1 2 \,   \Ep \left[ 
   \left( \ell_{ij}^{\alpha_i} \right)'
       \left( \sum_{h \in {\cal D}_i} \overline \ell_{ih}^{\, \alpha_i \alpha_i} \right)^{-1} 
 \left(  \sum_{h \in {\cal D}_i}  \overline \ell_{ih}^{\, * \, \beta_k \alpha_i \alpha_i}  \right)
             \left( \sum_{h \in {\cal D}_i} \overline \ell_{ih}^{\, \alpha_i \alpha_i} \right)^{-1} \ell_{ij}^{\alpha_i}  \right] 
          \\
    &\qquad \qquad \qquad   \left.
    + \frac 1 2 \,   \Ep \left[ 
   \left( \ell_{ij}^{\gamma_j} \right)'
       \left( \sum_{h \in {\cal D}_j} \overline \ell_{hj}^{\, \gamma_j \gamma_j} \right)^{-1} 
 \left(  \sum_{h \in {\cal D}_j}  \overline \ell_{hj}^{\, * \, \beta_k \gamma_j \gamma_j}  \right)
             \left( \sum_{h \in {\cal D}_j} \overline \ell_{hj}^{\, \gamma_j \gamma_j} \right)^{-1} \ell_{ij}^{\gamma_j}  \right]
             \right\}   .           
    \end{align*}
   Furthermore, if also Assumption~\ref{ass:PanelA2} holds, then
    \begin{align*}
    \widehat \delta - \delta^0 
      &= \left( \overline {\Delta}^{\,*\,\beta} \right)'
      (\widehat \beta - \beta^0) 
     + \frac 1 n    \sum_{(i,j) \in {\cal D}}
        \left\{
         \psi^{(\alpha) \prime}_i
          \ell_{ij}^{\, * \, \alpha_i} 
       +  \psi^{(\gamma) \prime}_j
          \ell_{ij}^{\, * \, \gamma_j}     
          \phantom{  \left( \sum_{h \in {\cal D}_i} \overline \ell_{ih}^{\, \alpha_i \alpha_i} \right)^{-1} }
          \right.
   \\ & \quad       
  - 
        \Ep \left[  
             \left(  \Delta_{ij}^{\alpha_i}  +  \ell_{ij}^{\, \alpha_i \alpha_i}  \psi^{(\alpha)}_i
                                   +  \ell_{ij}^{\, \alpha'_i \gamma_j}  \psi^{(\gamma)}_j    \right)'
             \left( \sum_{h \in {\cal D}_i} \overline \ell_{ih}^{\, \alpha_i \alpha_i} \right)^{-1} 
              \ell_{ij}^{\alpha_i}  
         \right]
    \\ & \quad                                     
             -  \Ep  \left[ 
               \left(  \Delta_{ij}^{\gamma_j}  
               +  \ell_{ij}^{\, \gamma_j \alpha_i}  \psi^{(\alpha)}_i
             +  \ell_{ij}^{\, \gamma_j \gamma_j}  \psi^{(\gamma)}_j 
                               \right)' 
             \left( \sum_{h \in {\cal D}_j} \overline \ell_{hj}^{\, \gamma_j \gamma_j} \right)^{-1} 
              \ell_{ij}^{\gamma_j}  
             \right]
    \\ & \quad   
    + \frac 1 2 \,   \Ep \left[ 
   \left( \ell_{ij}^{\alpha_i} \right)'
       \left( \sum_{h \in {\cal D}_i} \overline \ell_{ih}^{\, \alpha_i \alpha_i} \right)^{-1} 
 \left(  \sum_{h \in {\cal D}_i}  \overline \Delta_{ih}^{\, \# \,\alpha_i \alpha_i}  \right)
             \left( \sum_{h \in {\cal D}_i} \overline \ell_{ih}^{\, \alpha_i \alpha_i} \right)^{-1} \ell_{ij}^{\alpha_i}  \right] 
    \\ & \quad     \left.
    + \frac 1 2 \,   \Ep \left[ 
   \left( \ell_{ij}^{\gamma_j} \right)'
       \left( \sum_{h \in {\cal D}_j} \overline \ell_{hj}^{\, \gamma_j \gamma_j} \right)^{-1} 
 \left(  \sum_{h \in {\cal D}_j}  \overline \Delta_{hj}^{\,\# \,   \gamma_j \gamma_j}  \right)
             \left( \sum_{h \in {\cal D}_j} \overline \ell_{hj}^{\, \gamma_j \gamma_j} \right)^{-1} \ell_{ij}^{\gamma_j}  \right]
             \right\}         
               + o_P\left( 1/ \sqrt{n} \right) ,
    \end{align*}
    where the $d_\alpha \times d_\alpha$
    matrices $\overline \Delta_{ij}^{\, \# \,\alpha_i \alpha_i}$ 
     and the $d_\gamma \times d_\gamma$
    matrices $\overline \Delta_{ij}^{\,\# \,   \gamma_j \gamma_j} $ are given by
    \begin{align*}
           \overline \Delta_{ij}^{\, \# \,\alpha_i \alpha_i}
           &= \overline \Delta_{ij}^{\, \alpha_i \alpha_i}
          + \sum_{g=1}^{d_\alpha} \overline \ell_{ij}^{\, \alpha_i \alpha_i \alpha_{ig}}  \psi^{(\alpha)}_{ig}
           + \sum_{g=1}^{d_\gamma} \overline \ell_{ij}^{\, \alpha_i \alpha_i \gamma_{jg}}  \psi^{(\gamma)}_{jg}
           ,
         \\  
           \overline \Delta_{ij}^{\, \# \,\gamma_j \gamma_j}
           &= \overline \Delta_{ij}^{\, \gamma_j \gamma_j}
          + \sum_{g=1}^{d_\alpha} \overline \ell_{ij}^{\, \gamma_j \gamma_j \alpha_{ig}}  \psi^{(\alpha)}_{ig}
           + \sum_{g=1}^{d_\gamma} \overline \ell_{ij}^{\, \gamma_j \gamma_j \gamma_{jg}}  \psi^{(\gamma)}_{jg} .
    \end{align*}

\end{theorem}

\begin{proof}[\bf Proof]
  \# \underline{Expansion of $\widehat \beta$.}
    Our assumptions
    together with   results of
   Lemma~\ref{lemma:consistency}, \ref{lemma:HessianAdditive} and Lemma~\ref{lemma:concavity}
    guarantee that the conditions of
    Theorem~B.1 and Corollary~B.2 in \cite{FW16} 
    are satisfied, so that by applying that corollary we have
    \begin{align*}
         \sqrt{n} (\widehat \beta - \beta^0) = \overline W_{\infty}^{-1} U+ o_P(1) ,
    \end{align*}
    where $U= U^{(0)} + U^{(1)}$, with
   \begin{align*}
      U^{(0)} &=   
                     {\cal L}^\beta
                   +  \overline {\cal L}^{\beta \phi} \,  \overline {\cal H}^{-1} {\cal L}^\phi  
                   = {\cal L}^{* \, \beta} 
                   = \frac 1 {n^{1/2}}    \sum_{(i,j) \in {\mathcal{D}}}      \overline \ell_{ij}^{\, * \, \beta}  ,
   \nonumber \\
      U^{(1)} &=  
     \widetilde{\cal L}^{\beta \phi} \, \overline {\cal H}^{-1} {\cal L}^\phi 
        -   \overline {\cal L}^{\beta \phi} \,
                      \overline {\cal H}^{-1} \, \widetilde {\cal H} \,
                      \overline {\cal H}^{-1} \, {\cal L}^\phi  
    +  
             \frac 1 2 \,    \sum_{g=1}^{d_\phi}
           \left( \overline {\cal L}^{\beta \phi \phi_g}
               + \overline {\cal L}^{\beta \phi } \, \overline {\cal H}^{-1}
            \overline {\cal L}^{\phi \phi \phi_g} \right)
              [\overline {\cal H}^{-1} {\cal L}^\phi ]_g
              \overline {\cal H}^{-1} {\cal L}^\phi  
          \\
          &=     
           \widetilde{\cal L}^{* \, \beta \phi} \; \overline {\cal H}^{-1} {\cal L}^\phi 
    +  
             \frac 1 2 \,    \sum_{g=1}^{d_\phi}
          \overline {\cal L}^{\, * \, \beta \phi \phi_g} \,
              [\overline {\cal H}^{-1} {\cal L}^\phi ]_g \,
              \overline {\cal H}^{-1} {\cal L}^\phi   .
   \end{align*}
   Here, tilde symbols indicate deviations from expectation, for example, 
   $  \widetilde{\cal L}^{\beta \phi} =  {\cal L}^{\beta \phi} -   \overline{\cal L}^{\beta \phi}$,
   with $ \overline{\cal L}^{\beta \phi} = \Ep {\cal L}^{\beta \phi}$.
   Analogous to the proof of Theorem C.1 in \cite{FW16}, and also using the above Lemma~\ref{lemma:HessianAdditive}   again,  one can then show that the terms in $U^{(1)}$ only contribute asymptotic bias, namely
   \begin{align*}
      \widetilde{\cal L}^{* \, \beta \phi} \; \overline {\cal H}^{-1} {\cal L}^\phi 
      &=    \Ep \left[ \widetilde{\cal L}^{* \, \beta \phi} \; \overline {\cal H}^{-1} {\cal L}^\phi  \right] + o_P(1) 
   \\   
        &=  \Ep \left[ \widetilde{\cal L}^{* \, \beta \alpha} \; \left(   \overline {\cal H}_{(\alpha \alpha)}^* \right)^{-1} {\cal L}^\alpha  \right]
       + \Ep \left[ \widetilde{\cal L}^{* \, \beta \gamma} \; \left(   \overline {\cal H}_{(\gamma \gamma)}^* \right)^{-1} {\cal L}^\gamma  \right] + o_P(1) ,  
    \\  
      \frac 1 2 \,    \sum_{g=1}^{d_\phi}
          \overline {\cal L}^{\, * \, \beta \phi \phi_g} \,
              [\overline {\cal H}^{-1} {\cal L}^\phi ]_g \,
              \overline {\cal H}^{-1} {\cal L}^\phi 
       &=      \Ep \left[
         \frac 1 2 \,    \sum_{g=1}^{d_\phi}
          \overline {\cal L}^{\, * \, \beta \phi \phi_g} \,
              [\overline {\cal H}^{-1} {\cal L}^\phi ]_g \,
              \overline {\cal H}^{-1} {\cal L}^\phi 
              \right] + o_P(1)
          \\
          &=
             \Ep \left[
         \frac 1 2 \,    \sum_{g=1}^{I d_\alpha}
          \overline {\cal L}^{\, * \, \beta \alpha \alpha_g} \,
           \left[  \left(   \overline {\cal H}_{(\alpha \alpha)}^* \right)^{-1}   {\cal L}^\alpha \right]_g \,
            \left(   \overline {\cal H}_{(\alpha \alpha)}^* \right)^{-1}   {\cal L}^\alpha 
              \right] 
              \\
           & \quad
           +    \Ep \left[
         \frac 1 2 \,    \sum_{g=1}^{J d_\gamma}
          \overline {\cal L}^{\, * \, \beta \gamma \gamma_g} \,
           \left[  \left(   \overline {\cal H}_{(\gamma \gamma)}^* \right)^{-1}   {\cal L}^\gamma \right]_g \,
            \left(   \overline {\cal H}_{(\gamma \gamma)}^* \right)^{-1}   {\cal L}^\gamma 
              \right]   
              + o_P(1) . 
   \end{align*}
   In component notation we can now rewrite the above terms as follows (remember that we define the Hessian
   matrix $ \overline {\cal H}$ with a negative sign)
   \begin{align*}
            {\cal L}^\beta &  =    \frac 1 {\sqrt{n}}    \sum_{(i,j) \in {\cal D}}  \ell_{ij}^{\, * \, \beta_k} 
         \\   
       \Ep \left[ \widetilde{\cal L}^{* \, \beta \alpha} \; \left(   \overline {\cal H}_{(\alpha \alpha)}^* \right)^{-1} {\cal L}^\alpha  \right]
        &=
         -  \frac 1 {\sqrt{n}}    \sum_{(i,j) \in {\cal D}}
             \Ep \left[ \left(  \ell_{ij}^{\, * \, \beta_k \alpha_i}  \right)'
             \left( \sum_{h \in {\cal D}_i} \overline \ell_{ih}^{\, \alpha_i \alpha_i} \right)^{-1} \ell_{ij}^{\alpha_i}  \right] ,
      \\ 
       \Ep \left[ \widetilde{\cal L}^{* \, \beta \gamma} \; \left(   \overline {\cal H}_{(\gamma \gamma)}^* \right)^{-1} {\cal L}^\gamma  \right] 
     &= -  \frac 1 {\sqrt{n}}    \sum_{(i,j) \in {\cal D}}
      \Ep  \left[   \left(  \ell_{ij}^{\, * \, \beta_k \gamma_j}  \right)'
             \left( \sum_{h \in {\cal D}_j} \overline \ell_{hj}^{\, \gamma_j \gamma_j} \right)^{-1} \ell_{ij}^{\gamma_j}  \right] ,
 \end{align*}
 and            
   \begin{align*}
 & \Ep \left[
         \frac 1 2 \,    \sum_{g=1}^{I d_\alpha}
          \overline {\cal L}^{\, * \, \beta \alpha \alpha_g} \,
           \left[  \left(   \overline {\cal H}_{(\alpha \alpha)}^* \right)^{-1}   {\cal L}^\alpha \right]_g \,
            \left(   \overline {\cal H}_{(\alpha \alpha)}^* \right)^{-1}   {\cal L}^\alpha 
              \right] 
  \\ & \qquad
   =  \frac 1 2 \,   \frac 1 {\sqrt{n}}    \sum_{(i,j) \in {\cal D}}   \Ep \left[ 
   \left( \ell_{ij}^{\alpha_i} \right)'
       \left( \sum_{h \in {\cal D}_i} \overline \ell_{ih}^{\, \alpha_i \alpha_i} \right)^{-1} 
 \left(  \sum_{h \in {\cal D}_i}  \overline \ell_{ih}^{\, * \, \beta_k \alpha_i \alpha_i}  \right)
             \left( \sum_{h \in {\cal D}_i} \overline \ell_{ih}^{\, \alpha_i \alpha_i} \right)^{-1} \ell_{ij}^{\alpha_i}  \right] 
          \\
 &  \Ep \left[
         \frac 1 2 \,    \sum_{g=1}^{J d_\gamma}
          \overline {\cal L}^{\, * \, \beta \gamma \gamma_g} \,
           \left[  \left(   \overline {\cal H}_{(\gamma \gamma)}^* \right)^{-1}   {\cal L}^\gamma \right]_g \,
            \left(   \overline {\cal H}_{(\gamma \gamma)}^* \right)^{-1}   {\cal L}^\gamma 
              \right] 
              \\ & \qquad
 =  \frac 1 2 \,   \frac 1 {\sqrt{n}}    \sum_{(i,j) \in {\cal D}}     
       \Ep \left[ 
   \left( \ell_{ij}^{\gamma_j} \right)'
       \left( \sum_{h \in {\cal D}_j} \overline \ell_{hj}^{\, \gamma_j \gamma_j} \right)^{-1} 
 \left(  \sum_{h \in {\cal D}_j}  \overline \ell_{hj}^{\, * \, \beta_k \gamma_j \gamma_j}  \right)
             \left( \sum_{h \in {\cal D}_j} \overline \ell_{hj}^{\, \gamma_j \gamma_j} \right)^{-1} \ell_{ij}^{\gamma_j}  \right] .      
    \end{align*}
    Combining the above gives the expansion for $\widehat \beta - \beta^0$ in the theorem.
    
    \bigskip

   \# \underline{Expansion of $\widehat \delta$.}
   Again, our assumptions and lemmas guarantee that the conditions of
    Theorem~B.4 in \cite{FW16} are satisfied, so that by applying that theorem we have
\begin{align*}
    \widehat \delta - \delta 
   &=  \left(  \overline {\Delta}^\beta
      +
       \overline {\cal L}^{\beta \phi}
     \;   \overline {\cal H}^{-1} \;   \overline \Delta^\phi 
    \right)' (\widehat \beta - \beta^0) 
     + U^{(0)}_\Delta
     + U^{(1)}_\Delta  + o_P\left( 1/ \sqrt{n} \right)
     \\
     &= \left( \overline {\Delta}^{\,*\,\beta} \right)'
      (\widehat \beta - \beta^0) 
     + U^{(0)}_\Delta
     + U^{(1)}_\Delta  + o_P\left( 1/ \sqrt{n} \right) ,
   \end{align*}
   where
   \begin{align*}   
       U^{(0)}_\Delta &=  {\cal L}^{\phi \, \prime} \,    \overline {\cal H}^{-1} \, \overline \Delta^\phi ,
     \\ 
       U^{(1)}_\Delta
        &=   {\cal L}^{\phi \, \prime} \,    \overline {\cal H}^{-1} \, \widetilde \Delta^\phi
           -  {\cal L}^{\phi \, \prime} \,
            \overline {\cal H}^{-1}  \, \widetilde {\cal H} \,
              \overline {\cal H}^{-1}  \, \overline \Delta^\phi 
    + \ft 1 2 \,   {\cal L}^{\phi \, \prime} \overline {\cal H}^{-1}
  \left[ \overline \Delta^{\phi \phi} + 
         \sum_{g=1}^{d_\phi}
                   \overline {\cal L}^{{\phi \phi \phi_g}}
              \left(   \overline {\cal H}^{-1}   \,  \overline \Delta^\phi  \right)_g \right]
          \overline {\cal H}^{-1}  {\cal L}^\phi.
\end{align*}
Again, following the logic in the proof of  Theorem C.1 in \cite{FW16} one finds that
$  U^{(1)}_\Delta$ only contributes asymptotic bias, namely
\begin{align*}
       {\cal L}^{\phi \, \prime} \,    \overline {\cal H}^{-1} \, \widetilde \Delta^\phi
           -  {\cal L}^{\phi \, \prime} \,
            \overline {\cal H}^{-1}  \, \widetilde {\cal H} \,
              \overline {\cal H}^{-1}  \, \overline \Delta^\phi 
        &= \E\left[       
        {\cal L}^{\phi \, \prime} \,    \overline {\cal H}^{-1}
         \left(  \widetilde \Delta^\phi
           -   \widetilde {\cal H} \,
              \overline {\cal H}^{-1}  \, \overline \Delta^\phi  \right)
              \right] +  o_P\left( 1/ \sqrt{n} \right)
     \\
      &=   \E\left\{       
        {\cal L}^{\alpha \, \prime} \,      \left(   \overline {\cal H}_{(\alpha \alpha)}^* \right)^{-1} 
         \left[  \widetilde \Delta^\alpha
           -  \left(  \widetilde {\cal H} \,
              \overline {\cal H}^{-1}  \, \overline \Delta^\phi \right)_{(\alpha)} \right]   
              \right\} 
     \\ & \quad
     +  \E\left\{       
        {\cal L}^{\gamma \, \prime} \,      \left(   \overline {\cal H}_{(\gamma \gamma)}^* \right)^{-1} 
         \left[  \widetilde \Delta^\gamma
           -  \left(  \widetilde {\cal H} \,
              \overline {\cal H}^{-1}  \, \overline \Delta^\phi \right)_{(\gamma)} \right]   
              \right\}          
              +  o_P\left( 1/ \sqrt{n} \right)     ,
\end{align*}
and
\begin{align*}
   &  \ft 1 2 \,   {\cal L}^{\phi \, \prime} \overline {\cal H}^{-1}
  \left[ \overline \Delta^{\phi \phi} + 
         \sum_{g=1}^{d_\phi}
                   \overline {\cal L}^{{\phi \phi \phi_g}}
              \left(   \overline {\cal H}^{-1}   \,  \overline \Delta^\phi  \right)_g \right]
          \overline {\cal H}^{-1}  {\cal L}^\phi
   \\ 
     &=    \E\left\{     \ft 1 2 \,   {\cal L}^{\phi \, \prime} \overline {\cal H}^{-1}
  \left[ \overline \Delta^{\phi \phi} + 
         \sum_{g=1}^{d_\phi}
                   \overline {\cal L}^{{\phi \phi \phi_g}}
              \left(   \overline {\cal H}^{-1}   \,  \overline \Delta^\phi  \right)_g \right]
          \overline {\cal H}^{-1}  {\cal L}^\phi
            \right\}      +  o_P\left( 1/ \sqrt{n} \right)
   \\
    &=     \E\left\{     \ft 1 2 \,   {\cal L}^{\alpha \, \prime}    \left(   \overline {\cal H}_{(\alpha \alpha)}^* \right)^{-1}
  \left[ \overline \Delta^{\alpha \alpha} + 
         \sum_{g=1}^{d_\phi}
                   \overline {\cal L}^{{\alpha \alpha \phi_g}}
              \left(   \overline {\cal H}^{-1}   \,  \overline \Delta^\phi  \right)_g \right]
          \left(   \overline {\cal H}_{(\alpha \alpha)}^* \right)^{-1}  {\cal L}^\alpha
            \right\}     
      \\ & \quad
      +   \E\left\{     \ft 1 2 \,   {\cal L}^{\gamma \, \prime}    \left(   \overline {\cal H}_{(\gamma \gamma)}^* \right)^{-1}
  \left[ \overline \Delta^{\gamma \gamma} + 
         \sum_{g=1}^{d_\phi}
                   \overline {\cal L}^{{\gamma \gamma \phi_g}}
              \left(   \overline {\cal H}^{-1}   \,  \overline \Delta^\phi  \right)_g \right]
          \left(   \overline {\cal H}_{(\gamma \gamma)}^* \right)^{-1}  {\cal L}^\gamma
            \right\}        
       +  o_P\left( 1/ \sqrt{n} \right)     .
\end{align*}
   In component notation we can now rewrite the above terms as follows (again, remember that we define the Hessian
   matrix $ \overline {\cal H}$ with a negative sign)
 \begin{align*}
      & \E\left\{       
        {\cal L}^{\alpha \, \prime} \,      \left(   \overline {\cal H}_{(\alpha \alpha)}^* \right)^{-1} 
         \left[  \widetilde \Delta^\alpha
           -  \left(  \widetilde {\cal H} \,
              \overline {\cal H}^{-1}  \, \overline \Delta^\phi \right)_{(\alpha)} \right]   
              \right\} 
   \\ & \qquad
    =   - 
        \Ep \left[  
             \left(  \Delta_{ij}^{\alpha_i}  +  \ell_{ij}^{\, \alpha_i \alpha_i}  \psi^{(\alpha)}_i
                                   +  \ell_{ij}^{\, \alpha'_i \gamma_j}  \psi^{(\gamma)}_j    \right)'
             \left( \sum_{h \in {\cal D}_i} \overline \ell_{ih}^{\, \alpha_i \alpha_i} \right)^{-1} 
              \ell_{ij}^{\alpha_i}  
         \right] ,
    \\
     & \E\left\{       
        {\cal L}^{\gamma \, \prime} \,      \left(   \overline {\cal H}_{(\gamma \gamma)}^* \right)^{-1} 
         \left[   \Delta^\gamma
           -  \left(  \widetilde {\cal H} \,
              \overline {\cal H}^{-1}  \, \overline \Delta^\phi \right)_{(\gamma)} \right]   
              \right\}  
       \\ & \qquad          
      =   -  \Ep  \left[ 
               \left(  \Delta_{ij}^{\gamma_j}  
               +  \ell_{ij}^{\, \gamma_j \alpha_i}  \psi^{(\alpha)}_i
             +  \ell_{ij}^{\, \gamma_j \gamma_j}  \psi^{(\gamma)}_j 
                               \right)' 
             \left( \sum_{h \in {\cal D}_j} \overline \ell_{hj}^{\, \gamma_j \gamma_j} \right)^{-1} 
              \ell_{ij}^{\gamma_j}  
             \right] ,
 \end{align*}
  \begin{align*}
      &  \E\left\{     \ft 1 2 \,   {\cal L}^{\alpha \, \prime}    \left(   \overline {\cal H}_{(\alpha \alpha)}^* \right)^{-1}
  \left[ \overline \Delta^{\alpha \alpha} + 
         \sum_{g=1}^{d_\phi}
                   \overline {\cal L}^{{\alpha \alpha \phi_g}}
              \left(   \overline {\cal H}^{-1}   \,  \overline \Delta^\phi  \right)_g \right]
          \left(   \overline {\cal H}_{(\alpha \alpha)}^* \right)^{-1}  {\cal L}^\alpha
            \right\}     
   \\ & \qquad
   =  \frac 1 2 \,   \Ep \left[ 
   \left( \ell_{ij}^{\alpha_i} \right)'
       \left( \sum_{h \in {\cal D}_i} \overline \ell_{ih}^{\, \alpha_i \alpha_i} \right)^{-1} 
 \left(  \sum_{h \in {\cal D}_i}  \overline \Delta_{ih}^{\, \# \,\alpha_i \alpha_i}  \right)
             \left( \sum_{h \in {\cal D}_i} \overline \ell_{ih}^{\, \alpha_i \alpha_i} \right)^{-1} \ell_{ij}^{\alpha_i}  \right] ,
     \\
    &   \E\left\{     \ft 1 2 \,   {\cal L}^{\gamma \, \prime}    \left(   \overline {\cal H}_{(\gamma \gamma)}^* \right)^{-1}
  \left[ \overline \Delta^{\gamma \gamma} + 
         \sum_{g=1}^{d_\phi}
                   \overline {\cal L}^{{\gamma \gamma \phi_g}}
              \left(   \overline {\cal H}^{-1}   \,  \overline \Delta^\phi  \right)_g \right]
          \left(   \overline {\cal H}_{(\gamma \gamma)}^* \right)^{-1}  {\cal L}^\gamma
            \right\}             
    \\ & \qquad
      =    \frac 1 2 \,   \Ep \left[ 
   \left( \ell_{ij}^{\gamma_j} \right)'
       \left( \sum_{h \in {\cal D}_j} \overline \ell_{hj}^{\, \gamma_j \gamma_j} \right)^{-1} 
 \left(  \sum_{h \in {\cal D}_j}  \overline \Delta_{hj}^{\,\# \,   \gamma_j \gamma_j}  \right)
             \left( \sum_{h \in {\cal D}_j} \overline \ell_{hj}^{\, \gamma_j \gamma_j} \right)^{-1} \ell_{ij}^{\gamma_j}  \right]
                .
 \end{align*}  
  Combining the above gives the expansion for $\widehat \delta - \delta^0$ in the theorem.          
\end{proof}

\subsection{Proof of Main Text Theorems} 

\begin{proof}[\bf Proof of Theorem~\ref{th:BothEffects}]
According to Theorem~\ref{th:GeneralExpansion} we have
$\sqrt{n} \left( \widehat \beta - \beta^0 \right)   = \overline W_{\infty}^{-1} \,  U + o_P(1)$.
The first term in $U$ is   $\frac 1 {\sqrt{n}}    \sum_{(i,j) \in {\cal D}}     \ell_{ij}^{\, * \, \beta}$,
where in main text notation
we have $\ell_{ij}^{\, * \, \beta}  = \partial_z \ell_{ij} \tilde X_{ij}$.
Assumption~\ref{ass:PanelA1}(i) guarantees that $ \ell_{ij}^{\, * \, \beta}$ has mean zero
(a linear combination of scores evaluated at the true parameters) and is either independent across all $(i,j)$,
or only correlated within pairs $(i,j)$ and $(j,i)$. 
This term  therefore only contributes  variance, no  bias, to the limiting distribution of $\widehat \beta$.
Applying the Lindeberg-Levy CLT and the Cramer-Wold device we find
\begin{align*}
     \frac 1 {\sqrt{n}}    \sum_{(i,j) \in {\cal D}}     \ell_{ij}^{\, * \, \beta}
     &\rightarrow_d
     {\cal N}\left( 0,   \overline \Sigma_{\infty}   \right) ,
\end{align*}
where for the fully independent case (a) in Assumption~\ref{ass:PanelA1}(i),\footnote{
Here, we also used the Bartlett identity
$ \Ep \left(  \ell_{ij}^{\, * \, \beta} \right)   \left(  \ell_{ij}^{\, * \, \beta} \right)' =  \Ep \left(  - \ell_{ij}^{\, * \, \beta \beta} \right)  $.
}
\begin{align*}
      \overline \Sigma_{\infty}  
      &=   \plim_{I,J \to \infty}  \frac 1 {n}    \sum_{(i,j) \in {\cal D}}  \Ep \left(  \ell_{ij}^{\, * \, \beta} \right)   \left(  \ell_{ij}^{\, * \, \beta} \right)'
       =  \plim_{I,J \to \infty}  \frac 1 {n}    \sum_{(i,j) \in {\cal D}}  \Ep \left(  - \ell_{ij}^{\, * \, \beta \beta} \right)   
       =   \overline W_{\infty} .
\end{align*}
Thus, in case (a) the asymptotic variance of  $\widehat \beta$ simplifies to 
$W_{\infty}^{-1}  \overline \Sigma_{\infty} \overline W_{\infty}^{-1} = \overline W_{\infty}^{-1} $.
For case~(b) of Assumption~\ref{ass:PanelA1}(i) we have
\begin{align*}
      \overline \Sigma_{\infty}  
      &=   \plim_{I,J \to \infty}  \frac 1 {n}    \sum_{(i,j) \in {\cal D}} 
      \left[  \Ep \left(  \ell_{ij}^{\, * \, \beta} \right)   \left(  \ell_{ij}^{\, * \, \beta} \right)'
          +  \Ep \left(  \ell_{ij}^{\, * \, \beta} \right)   \left(  \ell_{ji}^{\, * \, \beta} \right)' \right]
    \\      
       &=  \plim_{I,J \to \infty}   \frac {1} {n}  \sum_{(i,j) \in {\cal D}}  \E \left\{ \left(
              \partial_{z}  \ell_{ij} \tilde X_{ij}  +  \partial_{z}  \ell_{ji} \tilde X_{ji} \right)    \partial_{z}  \ell_{ij}  \tilde X'_{ij} \right\}  ,
\end{align*}
where we use that $\ell_{ij}^{\, * \, \beta}  = \partial_z \ell_{ij} \tilde X_{ij}$. This is the formula for 
$ \overline \Sigma_{\infty}$ given in Theorem~\ref{th:GeneralExpansion}, and this formula covers both case (a)
and case (b), because independence across pairs $(i,j) \leftrightarrow (j,i)$ is of course a special case of dependence
across those pairs.

All the remaining terms in $U$ contribute asymptotic bias but no variance. We consider case~(a)
of Assumption~\ref{ass:PanelA1}(i) in the following, but one can easily verify that the additional bias terms
stemming from correlation across pairs $(i,j) \leftrightarrow (j,i)$ are asymptotically negligible, so that the same 
asymptotic bias expressions are obtained in case (b).

Using  $ \ell_{ij}^{\, * \, \beta_k \alpha_i} =  \gamma^0_j  \partial_{z^2} \ell_{ij} \tilde X_{ij,k}$ and
$ \overline \ell_{ih}^{\, \alpha_i \alpha_i} =  \gamma^0_j \gamma_j^{0 \prime}  \partial_{z^2} \overline \ell_{ij} $
and $\ell_{ij}^{\alpha_i} =   \gamma^0_j  \partial_{z} \ell_{ij}  $
we obtain
\begin{align*}
     \Ep \left[ \left(  \ell_{ij}^{\, * \, \beta_k \alpha_i}  \right)'
             \left( \sum_{h \in {\cal D}_i} \overline \ell_{ih}^{\, \alpha_i \alpha_i} \right)^{-1} \ell_{ij}^{\alpha_i}  \right]
      &=    \gamma^{0 \prime}_j  
             \left( \sum_{h \in {\cal D}_i} \gamma^0_h \gamma_h^{0 \prime}  \partial_{z^2} \ell_{ih} \right)^{-1}   \gamma^0_j 
           \,   \Ep \left( \partial_{z} \ell_{ij} \partial_{z^2} \ell_{ij}  \tilde X_{ij,k}  \right) ,
\end{align*}
and also using $\overline \ell_{ih}^{\, * \, \beta_k \alpha_i \alpha_i} = \gamma^0_j \gamma_j^{0 \prime}  
\Ep\left( \partial_{z^3} \ell_{ij} \tilde X_{ij,k} \right)$  and  the Bartlett identity $ \Ep  \ell_{ij}^{\alpha_i}
   \left( \ell_{ij}^{\alpha_i} \right)' = -  \overline \ell_{ij}^{\, \alpha_i \alpha_i}$,
\begin{align*}
  & \sum_{(i,j) \in {\cal D}}  \Ep \left[ 
   \left( \ell_{ij}^{\alpha_i} \right)'
       \left( \sum_{h \in {\cal D}_i} \overline \ell_{ih}^{\, \alpha_i \alpha_i} \right)^{-1} 
 \left(  \sum_{h \in {\cal D}_i}  \overline \ell_{ih}^{\, * \, \beta_k \alpha_i \alpha_i}  \right)
             \left( \sum_{h \in {\cal D}_i} \overline \ell_{ih}^{\, \alpha_i \alpha_i} \right)^{-1} \ell_{ij}^{\alpha_i}  \right] 
    \\
        &=      
        -   \sum_{i=1}^I  {\rm Tr} \left[  
       \left( \sum_{j \in {\cal D}_i} \overline \ell_{ij}^{\, \alpha_i \alpha_i} \right)^{-1} 
 \left(  \sum_{j \in {\cal D}_i}  \overline \ell_{ij}^{\, * \, \beta_k \alpha_i \alpha_i}  \right)
           \right] 
        =  - \sum_{(i,j) \in {\cal D}}    {\rm Tr} \left[   
       \left( \sum_{h \in {\cal D}_i} \overline \ell_{ih}^{\, \alpha_i \alpha_i} \right)^{-1} 
  \overline \ell_{ij}^{\, * \, \beta_k \alpha_i \alpha_i}             \right]   
  \\ &=
   - \sum_{(i,j) \in {\cal D}} 
     \gamma^{0 \prime}_j  
             \left( \sum_{h \in {\cal D}_i} \gamma^0_h \gamma_h^{0 \prime}  \partial_{z^2} \overline \ell_{ih} \right)^{-1}   \gamma^0_j 
           \,   \Ep \left(  \partial_{z^3}  \ell_{ij}   \tilde X_{ij} \right) ,
\end{align*}
and therefore
\begin{align*}
     & \frac 1 {\sqrt{n}}    \sum_{(i,j) \in {\cal D}}
          \left\{ -    \Ep \left[ \left(  \ell_{ij}^{\, * \, \beta_k \alpha_i}  \right)'
             \left( \sum_{h \in {\cal D}_i} \overline \ell_{ih}^{\, \alpha_i \alpha_i} \right)^{-1} \ell_{ij}^{\alpha_i}  \right]
             \right.
    \\ & \qquad \qquad  \quad
       \left.        
            +  \frac 1 2 \,   \Ep \left[ 
   \left( \ell_{ij}^{\alpha_i} \right)'
       \left( \sum_{h \in {\cal D}_i} \overline \ell_{ih}^{\, \alpha_i \alpha_i} \right)^{-1} 
 \left(  \sum_{h \in {\cal D}_i}  \overline \ell_{ih}^{\, * \, \beta_k \alpha_i \alpha_i}  \right)
             \left( \sum_{h \in {\cal D}_i} \overline \ell_{ih}^{\, \alpha_i \alpha_i} \right)^{-1} \ell_{ij}^{\alpha_i}  \right] 
             \right\}
  \\
  &= 
    -  \frac 1 {\sqrt{n}}   \sum_{(i,j) \in {\cal D}} 
     \gamma^{0 \prime}_j  
             \left( \sum_{h \in {\cal D}_i} \gamma^0_h \gamma_h^{0 \prime}  \partial_{z^2} \overline \ell_{ih} \right)^{-1}   \gamma^0_j 
           \,   \Ep \left(  \partial_{z} \ell_{ij} \partial_{z^2} \ell_{ij} \tilde X_{ij,k} + \frac 1 2 \partial_{z^3}  \ell_{ij}   \tilde X_{ij} \right)           
 \\
  &=  \sqrt{n} \; \frac{I} n         
\underbrace{
  \left[ -  \frac 1 I   \sum_{i=1}^I  \frac 1 {|{\cal D}_i|}  \sum_{j \in {\cal D}_i}
     \gamma^{0 \prime}_j  
             \left(\frac 1 {|{\cal D}_i|} \sum_{h \in {\cal D}_i} \gamma^0_h \gamma_h^{0 \prime}  \partial_{z^2} \overline \ell_{ih} \right)^{-1}   \gamma^0_j 
           \,   \Ep \left( \partial_{z} \ell_{ij} \partial_{z^2} \ell_{ij} \tilde X_{ij,k} + \frac 1 2 \partial_{z^3}  \ell_{ij}   \tilde X_{ij} \right)  
           \right]
       }_{\rightarrow_P  \overline B_\infty}    
       .
\end{align*}
Analogously we obtain
    \begin{align*}
      & \frac 1 {\sqrt{n}}    \sum_{(i,j) \in {\cal D}}
         \left\{    -  \Ep  \left[   \left(  \ell_{ij}^{\, * \, \beta_k \gamma_j}  \right)'
             \left( \sum_{h \in {\cal D}_j} \overline \ell_{hj}^{\, \gamma_j \gamma_j} \right)^{-1} \ell_{ij}^{\gamma_j}  \right]
             \right.
        \\
    \\ & \qquad \qquad  \quad
    + \left. \frac 1 2 \,   \Ep \left[ 
   \left( \ell_{ij}^{\gamma_j} \right)'
       \left( \sum_{h \in {\cal D}_j} \overline \ell_{hj}^{\, \gamma_j \gamma_j} \right)^{-1} 
 \left(  \sum_{h \in {\cal D}_j}  \overline \ell_{hj}^{\, * \, \beta_k \gamma_j \gamma_j}  \right)
             \left( \sum_{h \in {\cal D}_j} \overline \ell_{hj}^{\, \gamma_j \gamma_j} \right)^{-1} \ell_{ij}^{\gamma_j}  \right]
             \right\}   
     \\
   &=  \sqrt{n} \; \frac{J} n         
\underbrace{
  \left[ -  \frac 1 J   \sum_{j=1}^J  \frac 1 {|{\cal D}_j|}  \sum_{i \in {\cal D}_j}
     \alpha^{0 \prime}_i  
             \left(\frac 1 {|{\cal D}_j|} \sum_{h \in {\cal D}_j} \alpha^0_h \alpha_h^{0 \prime}  \partial_{z^2} \overline \ell_{hj} \right)^{-1}   \alpha^0_i 
           \,   \Ep \left( \partial_{z} \ell_{ij} \partial_{z^2} \ell_{ij} \tilde X_{ij,k} + \frac 1 2 \partial_{z^3}  \ell_{ij}   \tilde X_{ij} \right)  
           \right]
       }_{\rightarrow_P  \overline D_\infty}.                  
    \end{align*}
Combining the above gives the statement of the theorem.
\end{proof}

\begin{proof}[\bf Proof of Theorem~\ref{th:DeltaLimit}]
   Analogous to the proof of Theorem~\ref{th:BothEffects} we need to 
   translate the stochastic expansion of $\widehat \delta$ in Theorem~\ref{th:GeneralExpansion} into the notation used in the main text.
   We have
   $ \left( \overline {\Delta}^{\,*\,\beta} \right)' \rightarrow_P    \overline {(D_{\beta} \Delta)}_{\infty} $
   and
   $ \Psi_{ij} 
    = - \psi^{(\alpha) \prime}_i
          \gamma^0_j 
       -  \psi^{(\gamma) \prime}_j
          \alpha_i^0  $,
   and therefore   find for the variance terms that       
   \begin{align*}
        \underbrace{  \left( \overline {\Delta}^{\,*\,\beta} \right)'  \overline W_{\infty}^{-1}  \ell_{ij}^{\, * \, \beta}
        }_{= \overline {(D_{\beta} \Delta)}_{\infty} \overline W_\infty^{-1} \partial_{z} \ell_{ij} \tilde X_{ij}
     }
          +
          \underbrace{  \psi^{(\alpha) \prime}_i
          \ell_{ij}^{\, * \, \alpha_i} 
       +  \psi^{(\gamma) \prime}_j
          \ell_{ij}^{\, * \, \gamma_j}  
          }_{= -    \Psi_{ij}  
          \partial_{z} \ell_{ij}}
          &= \Gamma_{ij} .
   \end{align*}
   Analogous to the proof of Theorem~\ref{th:BothEffects} one
   can show for the bias terms that
    \begin{align*}
     & \frac 1 I    \sum_{(i,j) \in {\cal D}} \left\{
  - 
        \Ep \left[  
             \left(  \Delta_{ij}^{\alpha_i}  +  \ell_{ij}^{\, \alpha_i \alpha_i}  \psi^{(\alpha)}_i
                                   +  \ell_{ij}^{\, \alpha'_i \gamma_j}  \psi^{(\gamma)}_j    \right)'
             \left( \sum_{h \in {\cal D}_i} \overline \ell_{ih}^{\, \alpha_i \alpha_i} \right)^{-1} 
              \ell_{ij}^{\alpha_i}  
         \right]
         \right.
    \\ & \qquad  \qquad  \left.
    + \frac 1 2 \,   \Ep \left[ 
   \left( \ell_{ij}^{\alpha_i} \right)'
       \left( \sum_{h \in {\cal D}_i} \overline \ell_{ih}^{\, \alpha_i \alpha_i} \right)^{-1} 
 \left(  \sum_{h \in {\cal D}_i}  \overline \Delta_{ih}^{\, \# \,\alpha_i \alpha_i}  \right)
             \left( \sum_{h \in {\cal D}_i} \overline \ell_{ih}^{\, \alpha_i \alpha_i} \right)^{-1} \ell_{ij}^{\alpha_i}  \right] 
             \right\}         
             \rightarrow_P \overline B^\delta_{\infty} ,
    \end{align*}
  and
    \begin{align*}
     & \frac 1 J   \sum_{(i,j) \in {\cal D}} \left\{
             -  \Ep  \left[ 
               \left(  \Delta_{ij}^{\gamma_j}  
               +  \ell_{ij}^{\, \gamma_j \alpha_i}  \psi^{(\alpha)}_i
             +  \ell_{ij}^{\, \gamma_j \gamma_j}  \psi^{(\gamma)}_j 
                               \right)' 
             \left( \sum_{h \in {\cal D}_j} \overline \ell_{hj}^{\, \gamma_j \gamma_j} \right)^{-1} 
              \ell_{ij}^{\gamma_j}  
             \right]
         \right.
    \\ & \qquad  \qquad    \left.
    + \frac 1 2 \,   \Ep \left[ 
   \left( \ell_{ij}^{\gamma_j} \right)'
       \left( \sum_{h \in {\cal D}_j} \overline \ell_{hj}^{\, \gamma_j \gamma_j} \right)^{-1} 
 \left(  \sum_{h \in {\cal D}_j}  \overline \Delta_{hj}^{\,\# \,   \gamma_j \gamma_j}  \right)
             \left( \sum_{h \in {\cal D}_j} \overline \ell_{hj}^{\, \gamma_j \gamma_j} \right)^{-1} \ell_{ij}^{\gamma_j}  \right]
             \right\}         
               \rightarrow_P    \overline D^\delta_{\infty} .
    \end{align*}     
  Using the above and the expansion in Theorem~\ref{th:GeneralExpansion} gives the statement of Theorem~\ref{th:DeltaLimit}.
\end{proof}

\begin{proof}[\bf Proof of Theorem~\ref{th:bc}]
Under the conditions of Theorem \ref{th:BothEffects}, 
$\widehat B \,  \rightarrow_P \,  \overline B_{\infty} $, 
$\widehat D \,  \rightarrow_P \,  \overline D_{\infty}$, 
$  \widehat W \,  \rightarrow_P \, \overline W_{\infty}$, 
and $ \widehat \Sigma \,  \rightarrow_P \,  \overline \Sigma_{\infty}$.
If, in addition, the conditions of Theorem~\ref{th:DeltaLimit} hold, then also
$   \widehat V^{\delta} \, \rightarrow_P \,  \overline V_{\infty}^{\delta} $,
and the sample analogs of
$ \overline B^\delta_{\infty} $, 
$\overline D^\delta_{\infty}  $, 
$ \overline {(D_{\beta} \Delta)}_{\infty}$
are also consistent.
 These results follow from an identical argument to the proof of Lemma S.1 and Theorem 4.3    
     in the supplementary material of \cite{FW16}, which are based on a repeated application of  the weak law of large numbers and Slutsky's theorem.

 Once we have established the consistency of the  estimators of the bias terms, the asymptotic distributions of the analytical corrections $\widetilde \beta_{\rm ABC} $ and $\widetilde \delta_{\rm ABC} $ follow as corollaries of Theorems \ref{th:BothEffects} and \ref{th:DeltaLimit}, respectively. For example,
\begin{multline*}
\sqrt{n} \left( \widetilde \beta_{\rm ABC} - \beta^0 
         \right) = \sqrt{n} \left( \widehat \beta  - \frac{I} n \, \widehat W^{-1} \widehat B 
       - \frac{J} n \,  \widehat W^{-1} \widehat D - \beta^0 
                \right) \\ 
         =\sqrt{n} \left( \widehat \beta  - \beta^0 
      - \frac{I} n \,  W^{-1}  B 
       - \frac{J} n \,  W^{-1}  D 
         \right)  -  \frac{I} {\sqrt{n}} \, \left( \widehat W^{-1} \widehat B -   W^{-1}  B \right)
       - \frac{J} {\sqrt{n}} \, \left(  \widehat W^{-1} \widehat D -  W^{-1}  D \right) \\
        \;  \to_d \;
      {\cal N}( 0 ,   \;\overline W_{\infty}^{-1}  \overline \Sigma_{\infty} \overline W_{\infty}^{-1}),
\end{multline*}
by  Slutsky's theorem.
%
%
\end{proof}

\end{document}